\newif\ifsubmission

\ifsubmission
\documentclass{llncs}

\else
\documentclass[11pt]{article}
\usepackage{fullpage}
\fi

\ifsubmission \else
\usepackage{palatino}
\fi
\usepackage[normalem]{ulem}
\usepackage{framed}
\usepackage[utf8]{inputenc}
\usepackage[pdfstartview=FitH,pdfpagemode=UseNone,colorlinks,linkcolor=blue,filecolor=blue,citecolor=Violet,urlcolor=red]{hyperref}
\usepackage{cmap}
\usepackage[T1]{fontenc}
\usepackage{multirow}
\usepackage{float}
\usepackage[section]{placeins} 
\usepackage{mathtools} 
\usepackage[dvipsnames]{xcolor}
\usepackage{amsmath,amsthm,amssymb,algpseudocode,algorithm,cryptocode}
\usepackage{mathrsfs}
\usepackage[capitalize]{cleveref}
\usepackage{dashbox}
\usepackage{braket}
\usepackage{physics}
\usepackage{caption,subcaption}
\usepackage{xspace} 
\usepackage{braket}
\usepackage[normalem]{ulem}
\usepackage{soul,xcolor}
\usepackage{tikz}
\usetikzlibrary{quantikz} 
\usepackage{soul}
\usepackage{enumitem}
\usepackage{breakcites}
\usepackage{makecell}
\usepackage{hhline}
\usepackage[most]{tcolorbox}

\newif\ifnotes
\notestrue

\ifsubmission

\newtheorem{axiom}[theorem]{Axiom}
\newtheorem{physicsaxiom}[theorem]{Physics Axiom}
\newtheorem{importedtheorem}[theorem]{Imported Theorem}
\newtheorem{importedlemma}[theorem]{Imported Lemma}
\newtheorem{informaltheorem}[theorem]{Informal Theorem}
\newtheorem{physicstheorem}[theorem]{Physical Theorem}

\newtheorem{claim}[theorem]{Claim}
\newtheorem{subclaim}[theorem]{SubClaim}
\newtheorem{fact}[theorem]{Fact}
\newtheorem{construction}[theorem]{Construction}
\Crefname{importedtheorem}{Imported Theorem}{Imported Theorems}
\Crefname{importedlemma}{Imported Lemma}{Imported Lemma}
\Crefname{theorem}{Theorem}{Theorems}
\Crefname{proposition}{Proposition}{Propositions}
\Crefname{claim}{Claim}{Claims}
\Crefname{subclaim}{SubClaim}{SubClaims}
\Crefname{subsubclaim}{SubSubClaim}{SubSubClaims}
\Crefname{lemma}{Lemma}{Lemmas}
\Crefname{conjecture}{Conjecture}{Conjectures}
\Crefname{corollary}{Corollary}{Corollaries}
\Crefname{construction}{Construction}{Constructions}
\Crefname{property}{Property}{Properties}

\theoremstyle{definition}

\newtheorem{assumption}[theorem]{Assumption}
\newtheorem{notation}[theorem]{Notation}
\Crefname{definition}{Definition}{Definitions}
\Crefname{assumption}{Assumption}{Assumptions}
\Crefname{notation}{Notation}{Notations}

\theoremstyle{remark}

\Crefname{question}{Question}{Questions}
\Crefname{remark}{Remark}{Remarks}
\Crefname{comment}{Comment}{Comments}
\Crefname{fact}{Fact}{Facts}
\Crefname{step}{Step}{Steps}

\else

\newtheorem{theorem}{Theorem}[section]

\newtheorem{importedtheorem}[theorem]{Imported Theorem}

\newtheorem{claim}[theorem]{Claim}

\newtheorem{lemma}[theorem]{Lemma}
\newtheorem{conjecture}{Conjecture}
\newtheorem{corollary}[theorem]{Corollary}
\newtheorem{fact}[theorem]{Fact}

\newtheorem{definition}[theorem]{Definition}
\newtheorem{remark}[theorem]{Remark}
\Crefname{importedtheorem}{Imported Theorem}{Imported Theorems}
\Crefname{importedlemma}{Imported Lemma}{Imported Lemmas}
\Crefname{theorem}{Theorem}{Theorems}
\Crefname{proposition}{Proposition}{Propositions}
\Crefname{claim}{Claim}{Claims}
\Crefname{subclaim}{SubClaim}{SubClaims}
\Crefname{subsubclaim}{SubSubClaim}{SubSubClaims}
\Crefname{lemma}{Lemma}{Lemmas}
\Crefname{conjecture}{Conjecture}{Conjectures}
\Crefname{corollary}{Corollary}{Corollaries}
\Crefname{construction}{Construction}{Constructions}
\Crefname{property}{Property}{Properties}

\theoremstyle{definition}

\Crefname{definition}{Definition}{Definitions}
\Crefname{assumption}{Assumption}{Assumptions}
\Crefname{notation}{Notation}{Notations}

\theoremstyle{remark}

\Crefname{question}{Question}{Questions}
\Crefname{comment}{Comment}{Comments}
\Crefname{fact}{Fact}{Facts}
\Crefname{step}{Step}{Steps}

\fi

\newcommand{\secp}{\lambda}

\usepackage{soul}

\def\cC{{\cal C}}
\def\cD{{\cal D}}

\def\cH{{\cal H}}
\def\cI{{\cal I}}

\def\cK{{\cal K}}

\def\cP{{\cal P}}

\def\regA{{\cal A}}
\def\regB{{\cal B}}
\def\regC{{\cal C}}
\def\regD{{\cal D}}

\def\regM{{\cal M}}

\def\regX{{\cal X}}
\def\regY{{\cal Y}}
\def\regZ{{\cal Z}}

\newcommand{\sA}{\mathsf{A}}

\newcommand{\sF}{\mathsf{F}}
\newcommand{\sG}{\mathsf{G}}

\newcommand{\sM}{\mathsf{M}}

\def\bbC{{\mathbb C}}

\def\bbF{{\mathbb F}}

\def\bbN{{\mathbb N}}

\def\poly{{\rm poly}}

\def\negl{{\rm negl}}

\newcommand{\sk}{\mathsf{sk}}

\newcommand{\Sim}{\mathsf{Sim}}
\newcommand{\Sign}{\mathsf{Sign}}
\newcommand{\Enc}{\mathsf{Enc}}
\newcommand{\Dec}{\mathsf{Dec}}

\DeclareMathOperator*{\expectation}{\mathbb{E}}
\newcommand{\E}{\expectation}

\newcommand{\Gen}{\mathsf{Gen}}

\newcommand{\vk}{\mathsf{vk}}

\newcommand{\proref}[1]{Figure~\protect\ref{#1}}

\newenvironment{boxfig}[2]{\begin{figure}[#1]\fbox{
    \begin{minipage}{\linewidth}
    \vspace{0.2em}\makebox[0.025\linewidth]{}    \begin{minipage}{0.95\linewidth}{{#2 }}
    \end{minipage}\vspace{0.2em}\end{minipage}}}{\end{figure}}

\newcommand{\pprotocol}[4]{
\begin{boxfig}{h!}{
\begin{center}
\textbf{#1}
\end{center}
    {\footnotesize#4}
\vspace{0.2em} } \caption{\label{#3} #2}
\end{boxfig}
}

\newcommand{\protocol}[4]{
\pprotocol{#1}{#2}{#3}{#4} }

\renewcommand{\partial}{\mathsf{partial}}

\newcommand{\TD}{\mathsf{TD}}

\newcommand{\Adv}{\mathsf{Adv}}

\newcommand{\Ver}{\mathsf{Ver}}

\newcommand{\nonnegl}{\mathsf{non}\text{-}\mathsf{negl}}

\newcommand{\co}{\mathsf{co}}

\newcommand{\QObf}{\mathsf{QObf}}
\newcommand{\QEval}{\mathsf{QEval}}

\newcommand{\TokGen}{\mathsf{TokGen}}
\newcommand{\TokVer}{\mathsf{TokVer}}
\newcommand{\TokSign}{\mathsf{TokSign}}

\newcommand{\FSim}{\mathsf{FSim}}
\newcommand{\GSim}{\mathsf{GSim}}
\newcommand{\ParSim}{\mathsf{ParSim}}
\newcommand{\ParMeas}{\mathsf{ParMeas}}

\newcommand{\LinEval}{\mathsf{LinEval}}
\newcommand{\TV}{\mathsf{TV}}
\newcommand{\LM}{\mathsf{LM}}
\newcommand{\CNOT}{\mathsf{CNOT}}

\newcommand{\Mx}{\mathsf{Mx}}
\newcommand{\univ}{\mathsf{unv}}
\newcommand{\LMEval}{\mathsf{LMEval}}
\newcommand{\MS}{\Phi}

\begin{document}
\setstcolor{red}
\ifsubmission
\title{}
\author{}
\institute{}
\else
\title{Quantum State Obfuscation from Classical Oracles}
\author{
James Bartusek\thanks{bartusek.james@gmail.com} \\ UC Berkeley \and 
Zvika Brakerski\thanks{zvika.brakerski@weizmann.ac.il} \\ Weizmann Institute of Science \and
Vinod Vaikuntanathan\thanks{vinodv@csail.mit.edu} \\ MIT 
}
\date{\today}
\fi
\maketitle
\begin{abstract}
A major unresolved question in quantum cryptography is whether it is possible to obfuscate arbitrary quantum computation. Indeed, there is much yet to understand about the feasibility of quantum obfuscation even in the classical oracle model, where one is given for free the ability to obfuscate any classical circuit.

In this work, we develop a new array of techniques that we use to construct a \emph{quantum state obfuscator}, a powerful notion formalized recently by Coladangelo and Gunn (arXiv:2311.07794) in their pursuit of better software copy-protection schemes. Quantum state obfuscation refers to the task of compiling a \emph{quantum program}, consisting of a quantum circuit $C$ with a classical description and an auxiliary quantum state $\ket{\psi}$, into a functionally-equivalent \emph{obfuscated} quantum program that hides as much as possible about $C$ and $\ket{\psi}$. We prove the security of our obfuscator when applied to any pseudo-deterministic quantum program, i.e.\ one that computes a (nearly) deterministic classical input / classical output functionality. Our security proof is with respect to an efficient classical oracle, which may be heuristically instantiated using quantum-secure indistinguishability obfuscation for classical circuits. 

Our result improves upon the recent work of Bartusek, Kitagawa, Nishimaki and Yamakawa (STOC 2023) who also showed how to obfuscate pseudo-deterministic quantum circuits in the classical oracle model, but only ones with a \emph{completely classical} description. Furthermore, our result answers a question of Coladangelo and Gunn, who provide a construction of quantum state indistinguishability obfuscation with respect to a \emph{quantum} oracle, but leave the existence of a concrete real-world candidate as an open problem. Indeed, our quantum state obfuscator together with Coladangelo-Gunn gives the first candidate realization of a ``best-possible'' copy-protection scheme for all polynomial-time functionalities.

Our techniques deviate significantly from previous works on quantum obfuscation. We develop several novel technical tools which we expect to be broadly useful in quantum cryptography. These tools include a publicly-verifiable, linearly-homomorphic quantum authentication scheme with classically-decodable ZX measurements (which we build from coset states), and a method for compiling any quantum circuit into a "linear + measurement" ($\LM$) quantum program: an alternating sequence of CNOT operations and partial ZX measurements.

\end{abstract}
\thispagestyle{empty}
\newpage
\tableofcontents
\thispagestyle{empty}
\newpage 
\pagenumbering{arabic}

\section{Introduction}
\label{sec:intro}

Whether white-box access to programs is more powerful than black-box access is a fundamental question in computer science. This question motivates the study of {\em program obfuscation}, the task of converting a given program $\mathcal{P}$ into an obfuscated program $\hat{\mathcal{P}}$ that preserves functionality (i.e.\ the input-output behavior of $\mathcal{P}$), yet hides all other information about $\mathcal{P}$. Indeed, the previous decades have produced exciting research in cryptography solving several foundational questions about program obfuscation, including: which notions of obfuscation are possible and which aren't~\cite{JACM:BGIRSVY12,FOCS:GGHRSW13}; which ones are useful and to what extent~\cite{STOC:SahWat14}; and of course, how to construct secure program obfuscators under well-founded hardness assumptions~\cite{STOC:JaiLinSah21,EC:JaiLinSah22}. The study of program obfuscation has given us sophisticated cryptographic primitives such as witness encryption~\cite{STOC:GarGenSahWat13} and functional encryption~\cite{FOCS:GGHRSW13}, and it has shed new light on foundational complexity-theoretic questions~\cite{FOCS:BitPanRos15,STOC:IlaLiWil23}.

Developments in quantum information, computation, and cryptography~\cite{DBLP:journals/sigact/Wiesner83,BEN84,DBLP:conf/focs/Shor94,STOC:AarChr12,DBLP:conf/focs/BrakerskiCMVV18,DBLP:conf/focs/Mahadev18,DBLP:journals/corr/abs-2001-04383,Zha21} bring to fore a similarly foundational question:
\begin{quote}
\begin{center}
 \emph{Is it possible to obfuscate quantum computation, and to what extent?}
\end{center}
\end{quote}
Despite recent progress, there remains much to understand about the feasibility of quantum obfuscation, even in a world where one has the aid of a classical oracle, equivalently a world where ideal obfuscation of classical circuits comes for free. 

The strongest known positive result on obfuscating quantum functionalities is from the recent work of Bartusek, Kitagawa, Nishimaki and Yamakawa~\cite{BKNY23}, who showed how to obfuscate the class of {\em classically described} pseudo-deterministic quantum circuits in the classical oracle model, under the learning with errors (LWE) assumption. (For a description of other related work, see Section~\ref{sec:related}.) A deterministic quantum circuit is one that implements a classical, deterministic map $x \to Q(x)$. A \emph{pseudo}-deterministic circuit is nearly deterministic, i.e.\ its output on every (classical) input $x$ is some (classical) $y$ with all but negligible probability, e.g.\ the Shor circuit for factoring. In the \emph{classical oracle model}, the obfuscated circuit as well as the (possibly adversarial) evaluator have oracle access to an efficiently computable classical function sampled by the obfuscator.\footnote{Importantly, the actual description of this function is not necessarily revealed to the evaluator / adversary.} 

\subsection{Our Results}

In this work, we present a new array of techniques that we use to construct a \emph{quantum state obfuscation} scheme, an efficient compiler that converts a \emph{quantum program}, consisting of a quantum circuit $C$ as well as an auxiliary quantum input $\ket{\psi}$, into an  {obfuscated} quantum program that preserves functionality, but hides everything else. In a nutshell, our main result constructs an ideal quantum state obfuscator for pseudo-deterministic functionalities with respect to an efficient classical oracle. The class of circuits we obfuscate is strictly more general than those obfuscated by \cite{BKNY23}, which have a {\em purely classical description} (see \cref{subsec:discussion} for a more detailed comparison). Moreover, our security is {\em unconditional}\footnote{Technically, we only achieve unconditional security if we allow our oracles to use a true random oracle as a sub-routine. If we insist on the efficiency of the oracles, and thus instantiate the random oracle with a pseudo-random function, then we must assume the existence of a quantum-secure one-way function.} while \cite{BKNY23} additionally relied on the LWE assumption (in the classical oracle model).

\begin{theorem}[Informal]\label{thm:informal1}
There exists a quantum state obfuscator for the class of pseudo-deterministic quantum programs achieving the notion of ideal obfuscation (Definition~\ref{def:obf}) in the classical oracle model.
\end{theorem}

The notion of quantum state obfuscation was recently conceptualized and defined by Coladangelo and Gunn~\cite{CG23} in their pursuit of better software copy-protection schemes~\cite{DBLP:conf/coco/Aaronson09}. They show how a quantum state {\em indistinguishability obfuscator} can be used to construct a ``best-possible'' copy-protection scheme, i.e.\ they show a scheme to copy-protect any classical function that can be copy-protected at all! Their observation is that if there {\em exists} some (as yet unknown) copy-protection scheme for a classical (deterministic) function $x \to Q(x)$ that produces an unclonable quantum program $(\ket{\psi},C)$, then a quantum state obfuscation of $Q$ is just as good a copy-protection scheme. Indeed, this follows from the fact that the obfuscation of $Q$ is indistinguishable from an obfuscation of $(\ket{\psi},C)$. 

However, they leave the \emph{existence} of quantum state indistinguishability obfuscation as an open question, only providing a construction with respect to a quantum oracle that has no known (even heuristic) real-world instantiation. 

Our main result (Theorem~\ref{thm:informal1}) answers the open question from \cite{CG23}, providing a construction of quantum state obfuscation with respect  to a classical oracle. The oracle can be heuristically instantiated using a candidate quantum-secure indistinguishability obfuscation for classical circuits (e.g.\ \cite{TCC:GenGorHal15,TCC:BGMZ18,C:CheVaiWee18,brakerski_et_al:LIPIcs.ICALP.2022.28,10.1145/3406325.3451070,10.1007/978-3-030-77883-5_5}), thus providing the first concrete evidence that quantum state (indistinguishability) obfuscation is achievable. That is, while we don't currently have a proof in the plain model, it is plausible to conjecture the following.

\begin{conjecture}
Instantiating our scheme with one of the above candidate quantum-secure indistinguishability obfuscators for classical circuits yields a secure quantum state indistinguishability obfuscator.
\end{conjecture}

Then, following Coladangelo-Gunn, we obtain the first candidate realization of a ``best-possible'' copy-protection scheme for all polynomial-time programs. 

\begin{theorem}[Following \cite{CG23}]
Assuming the above conjecture, there exists a best-possible copy-protection scheme for all polynomial-time classical programs.
\end{theorem}

We note that the prior work of \cite{DBLP:conf/crypto/AaronsonLLZZ21} also obtains results on general-purpose copy-protection by working in the classical oracle model, and we describe the advantages of our approach (combined with \cite{CG23}) in \cref{subsec:discussion}. 

Our techniques deviate significantly from previous works on quantum obfuscation. For example, while previous work \cite{BM22,BKNY23} builds from the idea of obfuscating the verifier for a classical verification of quantum computation scheme, we directly perform verifiable computation on an authenticated and encrypted input. Our construction is comfortingly reminiscent of Yao's classical garbled circuits construction~\cite{DBLP:conf/focs/Yao86} (in fact, the free-XOR variant~\cite{DBLP:conf/icalp/KolesnikovS08}) which has had a tremendous array of applications in cryptography. We give an overview of our techniques in \cref{sec:tech-overview}, and provide more discussion on how they compare to prior work in \cref{subsec:discussion}.

\paragraph{On the classical oracle model.} Before proceeding, a word is in order about achieving results in the classical oracle model. First, a result in this model can be interpreted as reducing the problem of obfuscating quantum functionalities to classical functionalities. Indeed, instantiating the (efficiently implementable) oracle by its indistinguishability obfuscation (iO) gives us a heuristically secure construction in the plain model. Furthermore, by the ``best-possible'' security guarantee of iO, if there exists a secure implementation of the oracle in the plain model, iO is one such.  

Secondly, results in the classical oracle model have historically been important harbingers of subsequent research that showed analogous results without the aid of an oracle. For example, quantum money~\cite{STOC:AarChr12} was first achieved in a classical oracle model before it was de-oracle-ized~\cite{Zha21}; and copy-protection for unlearnable programs was first achieved in a quantum oracle model~\cite{DBLP:conf/coco/Aaronson09} before it was achieved in a classical oracle model~\cite{DBLP:conf/crypto/AaronsonLLZZ21} and later without oracles, for certain classes of functionalities (e.g. \cite{10.1007/978-3-030-84242-0_20,coladangelo2022quantum,DBLP:conf/tcc/LiuLQZ22,CG23}).

\subsection{Related Work}
\label{sec:related}

Alagic and Fefferman~\cite{AF16} presented definitions for obfuscating quantum circuits (and obfuscating classical circuits using quantum states). Their results, as well as those of Alagic, Brakerski, Dulek and Schaffner~\cite{ABDS21} are negative. The latter, for example, shows that virtual black-box (VBB) obfuscation of {\em classical circuits}, even with the aid of quantum information, is impossible. 

Broadbent and Kazmi~\cite{DBLP:conf/latincrypt/BroadbentK21} showed how to obfuscate quantum circuits that have only a few non-Clifford gates. In their construction, the size of the obfuscated circuit blows up exponentially with the number of $T$ gates. 
Bartusek and Malavolta~\cite{BM22} showed how to achieve indistinguishability obfuscation of null quantum circuits and, most related to our work, Bartusek, Kitagawa, Nishimaki and Yamakawa~\cite{BKNY23} showed how to obfuscate general pseudo-deterministic quantum circuits (with a classical description) in the classical oracle model.  Finally, Coladangelo and Gunn~\cite{CG23}, in a concurrent work, define the notion of quantum state (indistinguishability) obfuscation, show applications of this notion to software copy-protection, and construct a quantum state indistinguishability obfuscator in the {\em quantum oracle model}. We summarize these results, along with our contribution, in \cref{table:summary}.

\begin{table}
\scriptsize
\begin{center}
    \begin{tabular}{|c|c|c|c|c|c|c|c|}
        \hline
        Work & \makecell{Obfuscator \\ input} & \makecell{Obfuscator \\ output} & \makecell{Program \\ input} & \makecell{Program \\ output} & Program class & \makecell{Assumption/ \\ model} & Result \\
        \hhline{|=|=|=|=|=|=|=|=|}
        \cite{DBLP:conf/latincrypt/BroadbentK21} & Classical & Quantum* & Quantum & Quantum & \makecell{Unitaries w/ logarithmically \\ many non-Clifford gates} & \makecell{iO for classical \\ circuits} & iO\\
        \hline
        \cite{BM22} & Classical & Classical & Quantum* & Classical & Null circuits & \makecell{Classical oracle \\ model + LWE} & iO \\
        \hline 
        \cite{BKNY23} & Classical & Quantum & Classical & Classical & \makecell{(Pseudo)-Deterministic \\ circuits} & \makecell{Classical oracle \\ model + LWE} & VBB \\
        \hline
        \cite{CG23} & Quantum & Quantum & Classical & Classical & Deterministic circuits & \makecell{Quantum oracle \\ model} & iO \\
        \hline
        This work & Quantum & Quantum & Classical & Classical & \makecell{(Pseudo)-Deterministic \\ circuits} &  \makecell{Classical oracle \\ model} & Ideal\\
        \hline
    \end{tabular}
    \caption{Summary of work on quantum obfuscation. In \cite{DBLP:conf/latincrypt/BroadbentK21}, the obfuscator outputs a quantum state that can only be used to evaluate the program \emph{on one input}, and then is potentially destroyed. In \cite{BM22}, the obfuscated program can be run on quantum inputs, but requires \emph{multiple copies} of the quantum input. The last column refers to the definition of obfuscation that is actually achieved in each work: iO is indistinguishability obfuscation, and VBB is virtual black-box obfuscation. We note that VBB and the stronger notion of ideal obfuscation are morally very similar, and \cite{BKNY23} could also likely be shown to be an ideal obfuscator. Finally, we note that while achieving the notion of VBB/ideal obfuscation is only possible in the oracle model, the results that are in the \emph{classical} oracle model yield heuristic candidates for iO in the plain model.
    }
    \label{table:summary}
\end{center}
\end{table}
\section{Technical Overview}\label{sec:tech-overview}

During this overview, we'll slowly build up to our construction of quantum state obfuscation, highlighting the main ideas along the way. But first, it may be useful to convey a high level feel for the construction. To obfuscate a quantum program $(\ket{\psi},C)$ that implements the computation $x \to Q(x)$, we first encode the state $$\ket*{\widetilde{\psi}} \gets \Enc_k\left(\ket{\psi}\right)$$ 
using a novel quantum authentication scheme (QAS) that we design with particular properties in mind. Next, we compile $C$ into what we call a "linear + measurement", or $\LM$, quantum program.  Such programs consist solely of operations that can be performed on data authenticated with our QAS. Finally, we prepare a sequence of classical oracles $\sF_1,\dots,\sF_t,\sG$, where $t$ is the number of "measurement layers" in the $\LM$ quantum program. The oracles $\sF_1,\dots,\sF_t$ are designed to help the evaluator implement an encrypted sequence of adaptive measurements. They output random labels encoding the measurement results, which are then fed into downstream oracles. The oracle $\sG$ is designed to return the output $Q(x)$ if the evaluation was performed honestly. The final obfuscation then consists of the state $\ket*{\widetilde{\psi}}$ and the oracles $\sF_1,\dots,\sF_t,\sG$. We will describe each of these pieces and how they fit together in more detail.

We begin this technical overview by presenting our quantum authentication scheme (\cref{subsec:overview-authentication}). Next, we discuss the notion of $\LM$ quantum programs, and describe a compiler that writes any quantum program as an $\LM$ quantum program (\cref{subsec:overview-LM}). Next, we show how to use these building blocks to construct a garbling scheme and then a full-fledged obfuscation scheme for quantum computation (\cref{subsec:overview-GC}), and mention a couple of key ideas behind proving security. We defer a more detailed proof overview to \cref{subsec:proof-overview}. We conclude with a discussion and open problems (\cref{subsec:discussion}). 

\subsection{Quantum Authentication from Random Subspaces}\label{subsec:overview-authentication} 

\paragraph{Encode-encrypt authentication.} Our starting point is the notion of an ``encode-encrypt'' authentication scheme, as defined by Broadbent, Gutoski and Stebila~\cite{C:BroGutSte13}. Such schemes are parameterized by a family of CSS codes $\mathscr{C}$, and operate as follows. To encode a qubit $\ket{\psi}$, sample a random code $C \gets \mathscr{C}$ from the family, sample a quantum one-time pad key $(x,z)$, and output the ``encoded-and-encrypted'' state $X^x Z^x C\ket{\psi}$. As discussed by \cite{C:BroGutSte13}, various choices of the code family give rise to popular quantum authentication schemes (QAS), e.g., the polynomial scheme used for multi-party quantum computation \cite{BCGHS} and verifiable delegation \cite{aharonov2017interactive}, and the trap code used for quantum one-time programs \cite{C:BroGutSte13} and zero-knowledge proofs for QMA \cite{doi:10.1137/18M1193530}.

\paragraph{Our instantiation.} A crucial aspect of obfuscation that does not arise in these other settings is the need to preserve security when we allow the adversary to access the verifier of the authentication scheme an {\em a-priori unbounded} number of times. Indeed, the oracles released as part of our obfuscation scheme include subroutines that perform checks on authenticated data, and hence implicitly give the adversary reusable access to the verifier. This requirement of ``public-verifiability'' is not always satisfied by encode-encrypt schemes: for example, the trap code is completely insecure in this setting, as it is possible to learn the location of the traps via repeated queries to the verifier.

While certain flavors of public-verifiability have been considered previously in the quantum authentication literature (e.g. \cite{dulek2018quantum,C:GarYueZha17}), we find that a particularly simple instantiation of the encode-encrypt framework suffices for us: sample a random subspace $S$, a random shift $\Delta$, and use the CSS code defined by the isometry $E_{S,\Delta}$ that maps $\ket{0} \to \ket{S}, \ket{1} \to \ket{S+\Delta}$.\footnote{Here, we use the standard subspace state notation: for an (affine) subspace $S$, $\ket{S} \propto \sum_{s \in S}\ket{s}$.} That is, to encode an $n$-qubit state $\ket{\psi}$, sample a key $k=(S,\Delta,x,z)$ where $S$ is a $\secp$-dimensional subspace of $\bbF_2^{2\secp+1}$, $\Delta \in \bbF_2^{2\secp+1} \setminus S$, and $x,z \in \{0,1\}^{n \cdot (2\secp+1)}$, and output \[\ket*{\widetilde{\psi}} = X^x Z^z E_{S,\Delta}^{\otimes n}\ket{\psi} \coloneqq \Enc_k(\ket{\psi}).\]

Beyond satisfying a natural notion of public-verifiability (which will be discussed below), the resulting QAS satisfies the following desirable properties: (i) linear-homomorphism, and (ii) classically-decodable standard and Hadamard basis measurements (that is, a classical machine can decode the results of standard and Hadamard basis measurements performed on authenticated data).  We note that these latter properties are in fact endemic to encode-encrypt schemes (see discussion in \cite{C:BroGutSte13}), but we confirm them here for completeness.

\paragraph{Useful properties.} First, since $S$ is a subspace, one can confirm that CNOTs are transversal for this scheme as long as the same $(S,\Delta)$ is used to encode each qubit. That is, applying $2\secp+1$ CNOT gates qubit-wise to an encoding of $b_1$ and $b_2$ yields 
\[X^{x_1,x_2}Z^{z_1,z_2}\ket{S+ b_1 \cdot \Delta}\ket{S + b_2 \cdot \Delta} \to X^{x_1,x_1 \oplus x_2}Z^{z_1 \oplus z_2,z_2}\ket{S + b_1 \cdot \Delta}\ket{S + (b_1 \oplus b_2) \cdot \Delta}~,\]
which is indeed an encoding of the output of the CNOT operation using quantum one-time pad keys $(x_1,z_1\oplus z_2)$ and $(x_1 \oplus x_2, z_2)$. Thus, an evaluator can apply any sequence of CNOT gates, which we refer to as a "linear"\footnote{Of course, all quantum gates are linear with respect to the ambient Hilbert space of exponential dimension. Here, linearity specifically refers to the fact that any sequence of CNOT gates applies a linear function over $\bbF_2$ to each standard basis vector.} function, to authenticated data, as long as the decoder performs the analogous updates to their one-time pad keys.

Next, we note that standard basis measurements of an encoded qubit $X^xZ^z(\alpha\ket{S} + \beta\ket{S+\Delta})$ can be decoded {\em classically}. Indeed, any vector in $S+x$ can be interpreted as a 0, while any vector in $S+\Delta+x$ can be interpreted as a 1.

Finally, we check that the results of a Hadamard basis measurement can also be decoded {\em classically}. To do so, we'll define the "primal" codespace $S_\Delta \coloneqq S \cup (S+\Delta)$, and define the "dual" codespace to consist of  $\widehat{S} \coloneqq S_\Delta^\bot$ and  $\widehat{S} + \widehat{\Delta}$, where $\widehat{\Delta}$ is such that  $$\widehat{S}_{\widehat{\Delta}} \coloneqq \widehat{S} \cup (\widehat{S} + \widehat{\Delta}) = S^\bot~.$$ 
Then, it is not hard to check and confirm that \[H^{\otimes (2\secp+1)}X^xZ^z\left(\alpha\ket{S} + \beta\ket{S+\Delta}\right) = X^z Z^x\left(\frac{\alpha+\beta}{\sqrt{2}}\ket*{\widehat{S}} + \frac{\alpha-\beta}{\sqrt{2}}\ket*{\widehat{S}+\widehat{\Delta}}\right).\] Thus, any vector in $\widehat{S} + z$ can be interpreted as a 0 measurement result in the Hadamard basis, and any vector in $\widehat{S} + \widehat{\Delta}+z$ can be interpreted as a 1 measurement result in the Hadamard basis.

\paragraph{Reusable security.} Now we turn to the security of our scheme. Intuitively, we want to capture the fact that no adversary can successfully tamper with authenticated data, even given the ability to verify authenticated data. In more detail, given an authentication key $k = (S,\Delta,x,z)$, where $x = (x_1,\dots,x_n), z= (z_1,\dots,z_n)$, we define the following classical functionalities, which are parameterized by the key $k$ and a choice of bases $\theta \in \{0,1\}^n$.

\begin{itemize}
    \item $\Dec_{k,\theta}(\widetilde{v})$: On input a tuple of vectors $\widetilde{v}$ parsed as $(\widetilde{v}_1,\dots,\widetilde{v}_n)$, the decoding algorithm defines $v \in \{0,1\}^n$ as follows. For each $i \in [n]$:
    \[\text{if } \theta_i = 0: v_i = \begin{cases}0 \text{ if } \widetilde{v}_i \in S + x_i \\ 1 \text{ if } \widetilde{v}_i \in S + \Delta+ x_i \\ \bot \text{ otherwise}\end{cases} ~~~~~ \text{if } \theta_i = 1: v_i = \begin{cases}0 \text{ if } \widetilde{v}_i \in \widehat{S}+ z_i \\ 1 \text{ if } \widetilde{v}_i \in \widehat{S} + \widehat{\Delta}+ z_i \\ \bot \text{ otherwise}\end{cases}.\] If $v_i = \bot$ for some $i$, then output $\bot$, and otherwise output $v$.

    \item $\Ver_{k,\theta}(\widetilde{v})$: On input a tuple of vectors $\widetilde{v}$ parsed as $(\widetilde{v}_1,\dots,\widetilde{v}_n)$, the verification algorithm defines $v \in \{\top,\bot\}^n$ as follows.
    \[\text{if } \theta_i = 0: v_i = \begin{cases}\top \text{ if } \widetilde{v}_i \in S_\Delta + x_i \\ \bot \text{ otherwise}\end{cases} ~~~~~ \text{if } \theta_i = 1: v_i = \begin{cases}\top \text{ if } \widetilde{v}_i \in \widehat{S}_{\widehat{\Delta}} + z_i \\ \bot \text{ otherwise}\end{cases}.\] If $v_i = \bot$ for some $i$, then output $\bot$, and otherwise output $\top$.

\end{itemize}

That is, the verification algorithm just checks whether its inputs lie in the primal (resp. dual) codespace, while the decoding algorithm additionally computes the logical bits encoded by its inputs. We show that for any state $\ket{\psi}$, sequence of measurement bases $\theta \in \{0,1\}^n$, and adversarial measurement $\mathsf{Adv}$ that samples

\[\widetilde{v} \gets \Adv^{\Ver_{k,\cdot}(\cdot)}\left(X^x Z^z E_{S,\Delta}^{\otimes n}\ket{\psi}\right),\] the decoded value $v \gets \Dec_{k,\theta}(\widetilde{v})$ is either $\bot$, or its distribution is very close in total variation distance to the distribution that results from directly measuring $\ket{\psi}$ in the bases $\theta$. 

In fact, we also consider the possibility that the adversary is supposed to homomorphically apply some  sequence of CNOT gates (that is, a linear function $L$) to the authenticated data before measuring. Thus, in full generality we also parameterize the decoding $\Dec_{k,\theta,L}$ and verification $\Ver_{k,\theta,L}$ algorithms by a linear function $L$, which determines an updated sequence of one-time pad keys $(x_{L,1},\dots,x_{L,n}), (z_{L,1},\dots,z_{L,n})$ to be used in the decoding and verification.

\paragraph{A word on the proof of security.} Our proof combines two useful tricks from the literature: superspace sampling (\cite{C:Zhandry19,10.1007/978-3-030-84242-0_20}) and the Pauli twirl \cite{aharonov2017interactive}. Briefly, our first step is to sample random (say, $(3\secp/2 +1)$-dimensional) superspaces $R \supset S_\Delta$, $\widehat{R} \supset \widehat{S}_{\widehat{\Delta}}$ and use $(R,\widehat{R})$ in lieu of $(S_\Delta,\widehat{S}_{\widehat{\Delta}})$ in the definition of the oracle $\Ver_{k,\cdot,\cdot}(\cdot)$. Since $R$ and $\widehat{R}$ are random and small enough compared to the ambient space, the adversary cannot notice this change except with negligible probability. Next, we imagine sampling each one-time pad vector in two parts: for $x_i$ we sample an $x_{i,R} \gets R$ and an $x_{i,\co(R)} \gets \co(R)$, where $\co(R)$ is a set of coset representatives of $R$, and define $x_i = x_{i,R} + x_{i,\co(R)}$, and for $z_i$ we sample a $z_{i,\widehat{R}} \gets \widehat{R}$ and a $z_{i,\co(\widehat{R})} \gets \co(\widehat{R})$, and define $z_i = z_{i,\widehat{R}} + z_{i,\co(\widehat{R})}$. Finally, we consider the following experiment:
\begin{itemize}
    \item Sample $R,\widehat{R},\{x_{i,\co(R)},z_{i,\co(\widehat{R})}\}_{i \in [n]}$ and give this information to the adversary in the clear. Note that this is now sufficient to implement the oracle $\Ver_{k,\cdot,\cdot}(\cdot)$.
    \item Sample random $S,\Delta$ such that $\widehat{R}^\bot \subset S \subset S_\Delta \subset R$ and $\{x_{i,R},z_{i,\widehat{R}}\}_{i \in [n]}$ to complete the description of the authentication key $(S,\Delta,x,z)$. Send $X^xZ^zE_{S,\Delta}^{\otimes n}\ket{\psi}$ to the adversary, who mounts its attack.
\end{itemize}

At this point, we use the Pauli twirl over the space in between $\widehat{R}^\bot$ and $R$ to show that any adversarial operation can be decomposed into a fixed linear combination of Pauli attacks. To conclude, we use the randomness of $S,\Delta$ to show that any fixed Pauli attack will either be rejected with overwhelming probability or act as the identity on the encoded qubit, which completes the proof. See \cref{sec:authentication} for the full details of our definitions, construction, and security proofs.

\subsection{Linear + Measurement Quantum Programs}\label{subsec:overview-LM}

Next, we discuss our quantum program compiler. We start with any quantum circuit written using the $\{\CNOT,H,T\}$ universal gate set, where $H$ is the Hadamard gate, and $T$ applies a phase of $e^{i\pi/4}$. With the help of magic states, we compile the circuit into an alternating sequence of layers of CNOT gates (i.e.\ linear functions) and partial standard and Hadamard basis measurements, which we refer to as "ZX measurements".\footnote{Here, ZX is not meant to denote the composition of the Z and X operators, rather, it is meant as a shorthand for ``standard + Hadamard basis''.} We refer to the resulting program as a linear + measurement ($\LM$) quantum program. We note that the measurements are in fact partial in two aspects: (i) they may only operate on a subset of the qubits, and (ii) the measurement operators are projectors with rank potentially greater than 1. Furthermore, we allow the measurements to be adaptive, that is, their description may depend on previous measurement results (and the classical input to the computation).

More specifically, our goal will be to write each of the $H$ and $T$ gates as a sequence of CNOT gates, ZX measurements, and Pauli corrections derived from these measurement results. Then, the Pauli corrections can be commuted past future CNOT gates using the update rule $(x_1,z_1),(x_2,z_2) \to (x_1,z_1 \oplus z_2),(x_1 \oplus x_2,z_2)$, and incorporated into the description of future ZX measurements.

\paragraph{Handling the $H$ gate.} Following \cite{C:BroGutSte13}, we prepare the two-qubit magic state $$\ket{\phi_H} \propto \ket{00} + \ket{01} + \ket{10} - \ket{11},$$ and perform the Hadamard gate as shown in \cref{fig:H} (\cpageref{fig:H}), using one CNOT gate and Pauli corrections derived from a standard basis and a Hadamard basis measurement. As remarked in \cite{C:BroGutSte13}, it might seem strange at first that we are replacing a Hadamard gate with a circuit that nonetheless performs a Hadamard basis measurement. However, in our setting this does represent real progress: our authentication scheme does not support applying Hadamard gates directly to authenticated data,\footnote{At least, while preserving its linear-homomorphism. Applying a Hadamard gate transversally to authenticated data would result in an encoding with respect to the dual subspace, which would no longer support transversal CNOTs with data encoded using the primal subspace.} but does support the decoding of Hadamard basis measurements.

\paragraph{Handling the $T$ gate.} As we show on the bottom left of \cref{fig:T} (\cpageref{fig:T}), the $T$ gate can be implemented using the two magic states $$ \ket{\phi_T} \propto \ket{0} + e^{i\pi /4} \ket{1} \hspace{.2in} \mbox{and} \hspace{.2in} \ket{\phi_{PX}} \propto i\ket{0} + \ket{1}, $$ a CNOT gate, a \emph{classically controlled CNOT gate}, and Pauli corrections. 

Unfortunately, controlled CNOT is a "multi-linear" operation that we don't know how to directly implement on data authenticated with our authentication scheme. Therefore, taking inspiration from the "encrypted CNOT" operation introduced in \cite{FOCS:Mahadev18b},\footnote{Those familiar with \cite{FOCS:Mahadev18b}'s encrypted CNOT may notice the parallels: in \cite{FOCS:Mahadev18b}'s setting, these two measurements correspond to the two types of ``claws'' generated by the lattice-based encryption of $c$.} we replace the controlled CNOT operation with a \emph{projective measurement} $\Gamma_c$ controlled on the classical control bit $c$, where

 \begin{itemize}
        \item $\Gamma_0 = \{\ketbra{00}{00} + \ketbra{10}{10}, \ketbra{01}{01} + \ketbra{11}{11}\}$. That is, it measures its second input in the standard basis and has no effect on its first input.
        \item $\Gamma_1 = \{\ketbra{00}{00} + \ketbra{11}{11}, \ketbra{01}{01} + \ketbra{10}{10}\}$. That is, it measures the XOR of its two inputs, partially collapsing both.
\end{itemize}

Note that $\Gamma_1$ can roughly be seen as applying CNOT ``out of place'', writing the result to a third register, and then measuring it. We also remark that both of these measurements are diagonal in the standard basis, and thus can be performed on our authenticated data.

In the implementation of the $T$ gate shown on the bottom left of \cref{fig:T}, the $c$-CNOT operation is applied to the two magic states, followed by a measurement of the second magic state wire in the standard basis, and finally a $Z$ correction to the first magic state wire conditioned on both $c$ and the measurement outcome. One can show that the result of these operations is identical to what is shown on the bottom right of \cref{fig:T}: measure $\Gamma_c$ on the two magic state wires, then measure the second magic state wire in the \emph{Hadamard} basis, and finally apply a $Z$ correction to the first magic state wire conditioned on both $c$ and the \emph{XOR} of the two measurement results. We make this precise in the proof of \cref{claim:T-gate}, showing that our $\Gamma_c$-based implementation of the $T$ gate works as expected.

\paragraph{Formalizing $\LM$ quantum programs.} By combining these observations, we are able to specify any quantum program with $t$ many $T$ gates using an $n$-qubit state $\ket{\psi}$ (which in particular includes all of the necessary magic state) along with a sequence \[L_1,M_{\theta_1,f_1^{(\cdot)}},\dots,L_t,M_{\theta_t,f_t^{(\cdot)}},L_{t+1},M_{\theta_{t+1},g^{(\cdot)}}\] where
\begin{itemize}
    \item Each $L_i$ is a sequence of CNOT gates. 
    \item Each $M_{\theta_i,f_i^{(\cdot)}}$  (and $M_{\theta_i,g^{(\cdot)}}$) describes a partial ZX measurement in the following way:
    \begin{itemize}
        \item $\theta_i \in \{0,1,\bot\}^n$ defines a partial set of measurement bases. We define $\MS_{i,0} \coloneqq \{j : \theta_{i,j} = 0\}$ to be the set of registers measured in the standard basis and $\MS_{i,1} \coloneqq \{j : \theta_{i,j} = 1\}$ to be the set of registers measured in the Hadamard basis, and define $\MS_i \coloneqq (\MS_{i,0},\MS_{i,1})$ to be the total set of registers on which the $i$'th measurement is performed.
        \item Each $f_i^{(\cdot)}$ is a function that assigns measurement outcomes to basis states. The superscript indicates that its description may depend on previously generated information, i.e. the classical input $x$ to the computation and previous measurement results.
        \item To be precise, $M_{\theta_i,f_i^{(\cdot)}}$ can be described by the following measurement operators:
        \[\left\{H^{\MS_{i,1}}\left(\sum_{m : f_i^{(\cdot)}\left(m_{\MS_i}\right) = y}\ketbra{m}{m}\right)H^{\MS_{i,1}}\right\}_y,\] where $H^{\MS_{i,1}}$ applies a Hadamard gate to each qubit in the set $\MS_{i,1}$, and $m_{\MS_i}$ is the substring of $m$ consisting of the indices in $\MS_i$.
    \end{itemize}
\end{itemize}
\noindent

Thus, we have written our quantum program as an alternating sequence of linear operations and partial ZX measurements. We formalize this notion of an "$\LM$ quantum program" in \cref{def:LM-circuit}, and provide an example diagram of an $\LM$ quantum program in \cref{fig:LM}. 

However, looking ahead, it will be convenient to apply our obfuscator not to a completely arbitrary $\LM$ quantum program, but rather to an $\LM$ quantum program that satisfies a particular structural property. This property is described in \cref{def:LM-properties}, and satisfied by $\LM$ quantum programs output by our compiler described above. In order to formalize such programs (and our obfuscator), it will be necessary to introduce some further notation. For the purpose of this technical overview, we will introduce the notation and show how it is applied to the concrete compiler described above, but defer further details and a formalization of the property given by \cref{def:LM-properties} to the body.

It may be helpful to refer to \cref{fig:T} (our implementation of the $T$ gate) and \cref{fig:LM} (the example $\LM$ quantum program) while reading what follows. We begin by specifying (disjoint) sets $V_1,\dots,V_{t+1}$ and (disjoint) sets $W_1,\dots,W_t$ with the following properties.

\begin{itemize}
    \item  $\MS_1 = (V_1,W_1), \MS_2 = (V_1,V_2,W_2), \dots, \MS_t = (V_1,\dots,V_t,W_t), \MS_{t+1} = (V_1,\dots,V_{t+1}) = [n].$
    \item $V_i$ are the set of registers that are \emph{fully} collapsed in the standard or Hadamard basis by the $i$'th measurement. Concretely, $V_i$ consists of the 3rd wire of the $(i-1)$'th $T$-gate circuit (\cref{fig:T}), 1st and 2nd wires of any $H$-gate circuit (\cref{fig:H}) in the $i$'th layer, and 1st wire of the $i$'th $T$-gate circuit (\cref{fig:T}).
    \item $W_i$ are the set of registers that are \emph{partially} collapsed by the $i$'th measurement. Concretely, $W_i$ consists of the two magic state wires used for the $i$'th $T$ gate circuit. Indeed, the $i$'th measurement applies the controlled measurement $\Gamma_{c_i}$ to these wires, which only partially collapses them.
\end{itemize}

\noindent Then, we are able to specify further details about the $f_i^{(\cdot)}$ and $g^{(\cdot)}$ measurements.

\begin{itemize}
    \item Each $f_i^{(\cdot)}$ takes as input some sequence $(v_1,\dots,v_i,w_i)$ and outputs $v_i$ (fully collapsing the $V_i$ registers) along with a bit $r_i$ (partially collapsing the $W_i$ registers).
    \item In order to compute the bit $r_i$, the function $f_i^{(\cdot)}$ first needs to compute the $i$'th control bit $c_i$, which may depend on the input $x$ and all previous measurement results $(v_1,\dots,v_{i},r_1,\dots,r_{i-1})$. While we are able to provide $f_i^{(\cdot)}$ with the values $v_1,\dots,v_{i-1}$ on registers $V_1,\dots,V_{i-1}$ (which have been collapsed by previous measurements), this is not the case for the bits $r_1,\dots,r_{i-1}$, since the $W_1,\dots,W_{i-1}$ registers may have been computed on since previous measurements. To handle this, we remember the previous results $r_1,\dots,r_{i-1}$, and paramaterize $f_i^{x,r_1,\dots,r_{i-1}}$ by the input $x$ and previously generated bits $r_1,\dots,r_{i-1}$. Thus, the actual measurements are specified \emph{adaptively} using the previously generated bits $r_1,\dots,r_{i-1}$.
    \item In a similar manner, the function $g^{x,r_1,\dots,r_t}$ is parameterized by the input $x$ and previously generated bits $r_1,\dots,r_t$. It takes as input some sequence $(v_1,\dots,v_{t+1})$ and instead of performing an intermediate measurement, it computes the final output $y = Q(x)$.
\end{itemize}

Finally, we have set up enough notation to start discussing our actual obfuscation construction, which follows. 

\subsection{Obfuscation Construction}\label{subsec:overview-GC} 

So far, we have discussed a method for authenticating quantum states $\ket*{\widetilde{\psi}} = \Enc_k(\ket{\psi})$ using key $k = (S,\Delta,x,z)$, and a method for writing any quantum program as
\[\ket{\psi},L_1,M_{\theta_1,f_1^{(\cdot)}},\dots,L_t,M_{\theta_t,f_t^{(\cdot)}},L_{t+1},M_{\theta_{t+1},g^{(\cdot)}}\] where $\ket{\psi}$ consists of a quantum state that was part of the description of the original program, as well as some magic states.

\paragraph{Garbling via encrypted measurements.} We will build up to our full obfuscation construction by first describing how to \emph{garble} quantum circuits using our approach. That is, we'll suppose that the evaluator is only interested in computing the output on a particular input $x$, and show how to design oracles $\sF_1[x],\dots,\sF_t[x],\sG[x]$ that, along with the authenticated state $\ket*{\widetilde{\psi}} = \Enc_k(\ket{\psi})$, allow the evaluator to perform the entire computation on top of authenticated data, and eventually obtain $Q(x)$ without learning anything else about the program's implementation.

The basic idea is to implement the measurement $M_{\theta_i,f_i^{(x,\cdot)}}$ on authenticated data using an oracle $\sF_i[x]$, that, instead of outputting the results $(v_i,r_i)$ in the clear, outputs the encoded version $\widetilde{v}_i$ of $v_i$ (that is, $\widetilde{v}_i$ is in the support of the authenticated state that encodes the logical string $v_i$) along with a random \emph{label} $\ell_i$ representing the bit $r_i$. We will always denote vectors that result from measuring authenticated states (but not decoding) with a tilde (e.g. $\widetilde{v}_i$). 

Roughly $\sF_i[x]$ will be implemented as follows. It takes as input vectors $\widetilde{v}_1,\dots,\widetilde{v}_i,\widetilde{w}_i$ from the support of authenticated states (obtained from authenticated wires $V_1,\dots,V_i,W_i$), as well as labels $\ell_1,\dots,\ell_{i-1}$ that encode the results $r_1,\dots,r_{i-1}$ of previous measurements. It first decodes its inputs, and then uses the decoded values to compute the next measurement results $(v_i,r_i)$. Finally, it outputs the encodings $(\widetilde{v}_i,\ell_i)$ where $\ell_i$ is a label for $r_i$ computed via a random oracle $H$. 

We will implement $\sG[x]$ in exactly the same way, except that it directly outputs the result $y$. Sketches of these oracles follow.\\

\noindent \underline{$\sF_i[x](\widetilde{v}_1,\dots,\widetilde{v}_i,\widetilde{w}_i,\ell_1,\dots,\ell_{i-1})$}
\begin{enumerate}
    \item $(v_1,\dots,v_i,w_i) \gets \Dec_{k,\theta_i,L_i\dots L_1}(\widetilde{v}_1,\dots,\widetilde{v}_i,\widetilde{w}_i)$.\footnote{Recall the description of the decoding oracle $\Dec$ from \cref{subsec:overview-authentication}. We additionally parameterize the oracle with a concatenation of the linear functions $L_i \dots L_1$, which determines the sequence of Pauli one-time-pad keys to be used during the decoding. } Abort if the output is $\bot$.
    \item For each $\iota \in [i-1]$, let \[\ell_{\iota,0} = H(\widetilde{v}_1,\dots,\widetilde{v}_{\iota},\ell_1,\dots,\ell_{\iota-1},0), ~~~ \ell_{\iota,1} = H(\widetilde{v}_1,\dots,\widetilde{v}_{\iota},\ell_1,\dots,\ell_{\iota-1},1),\] and let $r_\iota$ be the bit such that $\ell_\iota = \ell_{\iota,r_\iota}$, or abort if there is no such bit.
    \item Compute $(v_i,r_i) = f_i^{x,r_1,\dots,r_{i-1}}(v_1,\dots,v_i,w_i)$.
    \item Set $\ell_i \coloneqq H(\widetilde{v}_1,\dots,\widetilde{v}_i,\ell_1,\dots,\ell_{i-1},r_i),$ and output $(\widetilde{v}_i,\ell_i)$.
\end{enumerate} 

\noindent \underline{$\sG[x](\widetilde{v}_1,\dots,\widetilde{v}_{t+1},\ell_1,\dots,\ell_{t})$}
\begin{enumerate}
    \item $(v_1,\dots,v_{t+1}) \gets \Dec_{k,\theta_i,L_{t+1}\dots L_1}(\widetilde{v}_1,\dots,\widetilde{v}_{t+1})$. Abort if the output is $\bot$.
    \item For each $\iota \in [t]$, let \[\ell_{\iota,0} = H(\widetilde{v}_1,\dots,\widetilde{v}_{\iota},\ell_1,\dots,\ell_{\iota-1},0), ~~~ \ell_{\iota,1} = H(\widetilde{v}_1,\dots,\widetilde{v}_{\iota},\ell_1,\dots,\ell_{\iota-1},1),\] and let $r_\iota$ be the bit such that $\ell_\iota = \ell_{\iota,r_\iota}$, or abort if there is no such bit.
    \item Compute and output $y = g^{x,r_1,\dots,r_t}(v_1,\dots,v_{t+1})$.
\end{enumerate}

Proving the security of this garbled program consists of two main steps: (1) a ``soundness'' argument establishing that no adversary, given $\ket*{\widetilde{\psi}}$ and oracle access to $\sF_1[x],\dots,\sF_t[x]$ should be able to output classical strings $(\widetilde{v}_1,\dots,\widetilde{v}_{t+1},\ell_1,\dots,\ell_t)$ such that $\sG[x](\widetilde{v}_1,\dots,\widetilde{v}_{t+1},\ell_1,\dots,\ell_t) \notin \{Q(x),\bot\}$, and (2) a ``simulation'' argument establishing that the $\sF_1[x],\dots,\sF_t[x]$ oracles can be simulated using a verification oracle $\Ver_{k,\cdot,\cdot}(\cdot)$ for the authentication scheme \emph{instead of} the decoding functionality $\Dec_{k,\cdot,\cdot}(\cdot)$. Indeed, a common theme throughout our proof strategy is understanding how we can replace $\Dec_{k,\cdot,\cdot}(\cdot)$ with $\Ver_{k,\cdot,\cdot}(\cdot)$ so that we can then appeal to the security of the authentication scheme. Further discussion on these two steps can be found in our proof intuition section, \cref{subsec:proof-overview}. Here, we just mention that the main idea for the soundness argument is an inductive strategy, where we perform the first measurement and appeal to soundness of a garbled program with one fewer measurement layer.

\paragraph{From garbling to obfuscation via signature tokens.} To complete our construction of full-fledged obfuscation, it remains to show how to grant the evaluator the ability to execute the circuit on \emph{any} input $x$ of its choice, without risking any other leakage on the description of the program. A natural idea is to re-define $\sF_1,\dots,\sF_t,\sG$ so that they additionally take $x$ as input, and include $x$ in the hashes $H(x,\widetilde{v}_1,\dots,\widetilde{v}_i,\ell_1,\dots,\ell_{i-1},r_i)$ that define the output labels. We intuitively want to sample different output labels for each $x$ so that the resulting obfuscation scheme can roughly be seen as a ``concatenation'' of \emph{independently sampled} garbling schemes for each $x$ (that share the same initial authenticated state). However, it turns out that this is not yet enough to ensure security.

Indeed, nothing is preventing the adversary from applying a type of ``mixed input'' attack, where they evaluate honestly on an input $x$, but at some point insert a measurement implemented by the oracle $\sF_i(x',\cdot)$ on some input $x' \neq x$. That is, at layer $i$, the adversary would first implement $\sF_i(x',\cdot)$ (and ignore the output labels) to collapse the state in some way before continuing with their honest evaluation procedure using $\sF_i(x,\cdot)$. Unfortunately, this rogue call to $\sF_i(x',\cdot)$ wouldn't destroy the current state enough to cause the remaining oracle calls to $\sF_i(x,\cdot),\dots,\sF_t(x,\cdot),\sG(x,\cdot)$ to abort, but \emph{might} collapse the state in a manner inconsistent with an honest evaluation on input $x$, eventually allowing the adversary to break the ``soundness'' of the scheme by finding an input to $\sG(x,\cdot)$ that results in an output $y \neq Q(x)$.

Taking inspiration from \cite{BKNY23} who faced a similar issue, we solve our problem via the use of \emph{signature tokens} \cite{arxiv:BenSat16}. This quantum cryptographic primitive consists of a quantum signing key $\ket{\sk}$ that may be used to produce a classical signature $\sigma_x$ on any \emph{single} message $x$ but never \emph{two} signatures $\sigma_x,\sigma_{x'}$ on two different messages $x,x'$ simultaneously. We include a quantum signing key $\ket{\sk}$ as part of our obfuscation construction, and re-define the oracles $\sF_1,\dots,\sF_t,\sG$ to take $x$ \emph{and} a signature $\sigma_x$ as input, abort if the signature is invalid, and otherwise include both in the hashes $H(x,\sigma_x,\widetilde{v}_1,\dots,\widetilde{v}_i,\ell_1,\dots,\ell_{i-1},r_i)$ that define the output labels.

Intuitively, this prevents the above attack. Once the adversary has begun an honest evaluation on some input $x$, it must ``know'' some valid signature $\sigma_x$, preventing it from querying the oracles on any input that starts with $(x',\sigma_{x'})$. That is, if it actually want to evaluate on $x'$, it must uncompute everything it has computed so far to return to $\ket{\sk}$ before it can produce a signature $\sigma_{x'}$ and begin evaluating on $x'$. 

We provide more details about this approach in the proof intuition section, \cref{subsec:proof-overview}. Of note is the fact that we crucially use a \emph{purified} random oracle \cite{C:Zhandry19} in order to extract a signature token on $x$ from any adversary who has begun evaluating the obfuscated program on $x$. Formalizing this approach is one of trickier aspects of the proof, and we describe a toy problem in \cref{subsec:proof-overview} that may provide some intuition for the actual proof.

\subsection{Discussion and Open Problems}\label{subsec:discussion}

\paragraph{Comparison with \cite{DBLP:conf/crypto/AaronsonLLZZ21}.} In \cref{sec:intro}, we described an application of our construction to ``best-possible'' copy-protection, a notion recently introduced by \cite{CG23}. However, it has already been shown by \cite{DBLP:conf/crypto/AaronsonLLZZ21} that copy-protection for \emph{all unlearnable functions} exists in the classical oracle model. So what is the advantage of our approach? 

The catch with \cite{DBLP:conf/crypto/AaronsonLLZZ21} is that it is also known that copy-protection for all unlearnable functions is \emph{unachievable} in the plain model \cite{10.1007/978-3-030-77886-6_17}. Thus, one cannot expect to instantiate \cite{DBLP:conf/crypto/AaronsonLLZZ21}'s construction in the plain model with, say, an indistinguishably obfuscator, and conjecture that their notion of security still holds. Moreover, it is unclear (to us) how one would prove that \cite{DBLP:conf/crypto/AaronsonLLZZ21} or any other approach is a ``best-possible'' copy-protector without going through the intermediate notion of quantum state indistinguishability obfuscation.

On the other hand, there is no known impossibility for quantum state indistinguishability obfuscation in the plain model. Thus, one can conjecture that our approach (or some variant of it) will eventually yield a quantum state indistinguishability obfuscator in the plain model, which, by the results of \cite{CG23}, would immediately yield best-possible copy-protection for \emph{any} functionality. Hence, we view our results (combined with \cite{CG23}) as marking significant progress towards the long-standing goal of achieving general-purpose copy-protection from concrete cryptographic assumptions.

\paragraph{Comparison with \cite{BKNY23}.} As mentioned in \cref{sec:intro}, our techniques differ significantly from prior work on quantum obfuscation. To be more concrete, let's take the example of \cite{BKNY23}, who constructed obfuscation for classically-described pseudo-deterministic quantum functionalities. 

At a high level, their approach was to encrypt the description of the circuit $Q \to \Enc(Q)$ using a blind quantum computation protocol (e.g. \cite{FOCS:Mahadev18b}) and prepare a classical oracle that is able to verify the result of computing $\Enc(Q) \to \Enc(Q(x))$, and decrypt the result if verification passes. But if the program $\ket{Q}$ is described quantumly, this approach completely breaks down, since the classical oracle cannot even \emph{interpret} the \emph{quantum} statement $\Enc(\ket{Q}) \to \Enc(Q(x))$ to be verified.

Our approach is to not only encrypt $\ket{Q}$ but to \emph{authenticate} it as well. Then, we design oracles that are able to verify the blind computation \emph{along the way}, as opposed to \emph{all at once} at the end of the computation. That is, while both approaches make fundamental use of both blind and verifiable quantum computation, they differ significantly in execution.

Finally, it is worth mentioning that the \cite{BKNY23} approach may prove to be advantageous in the case where we are only interested in obfuscating quantum circuits with a classical description, and desire to produce an obfuscated program that \emph{also} has a classical description. Given the fact that quantum state obfuscation is a best-possible copy-protector, it is actually \emph{inherent} that our obfuscated program includes (unclonable) quantum states, even when applied to a circuit with classical description. On the other hand, while the \cite{BKNY23} approach as currently implemented also produces obfuscated programs with a quantum description, it is reasonable to hope that it could be de-quantized. Indeed, the only reason that \cite{BKNY23}'s construction includes a quantum state is that two of their building blocks - signature tokens and Pauli functional commitments - are constructed using quantum keys. However, it is plausible that both of these building blocks could be instantiated with a classical key (e.g.\ using the ideas of \cite{10.1145/3357713.3384304}, though proving security may be difficult). We leave further exploration of these ideas to future work.

\paragraph{Open problems.} We conclude this overview by mentioning a couple of natural open problems that are motivated by this work. Perhaps most obviously, can we obtain provable security in the plain model from quantum-secure indistinguishability obfuscation (or a different concrete and plausible assumption on classical obfuscators)? Besides representing major progress in our understanding of quantum obfuscation, answering this question would have significant implications for general-purpose software copy-protection (due to \cite{CG23}), another area where positive results are difficult to come by.

Next, can we generalize beyond pseudo-deterministic programs? As a first step, can we obfuscate \emph{sampling} circuits, i.e.\ those with classical inputs that produce a \emph{distribution} over classical outputs? Interestingly, while \cite{BKNY23}'s construction does not even satisfy correctness for (even classically-described) sampling circuits, our approach can plausibly be applied to sampling circuits, although it remains open to obtain any provable guarantees. Finally, the feasibility of obfuscating general-purpose circuits with quantum inputs and/or quantum outputs is still wide open, and remains a fascinating direction for future exploration, with potential applications beyond cryptography, e.g.\ to quantum complexity theory.

\section{Preliminaries}

Let $\secp$ denote the security parameter. We write $\negl(\cdot)$ to denote any \emph{negligible} function, which is a function $f$ such that for every constant $c \in \mathbb{N}$ there exists $N \in \mathbb{N}$ such that for all $n > N$, $f(n) < n^{-c}$. We write $\nonnegl(\cdot)$ to denote any function $f$ that is not negligible. That is, there exists a constant $c$ such that for infinitely many $n$, $f(n) \geq n^{-c}$. Finally, we write $\poly(\cdot)$ to denote any polynomial function $f$. That is, there exists a constant $c$ such that for all $n \in \bbN$, $f(n) \leq n^{-c}$. For two probability distributions $D_0,D_1$ with classical support $S$, let \[\mathsf{TV}\left(D_0,D_1\right) \coloneqq \sum_{x \in S}|D_0(x) - D_1(x)|\] denote the total variation distance. For a set $S$, we let $x \gets S$ denote sampling a uniformly random element $x$ from $S$. If $D$ is a distribution, we let $x \gets D$ denote sampling from $D$, and let \[\left\{x : x \gets D_0\right\} \approx_\epsilon \left\{x : x \gets D_1\right\}\] denote that $\TV(D_0,D_1) \leq \epsilon$. Finally, we denote a linear combination of distributions by 
\[(1-\delta)\{x : x \gets D_0\} + \delta \{x : x \gets D_1\},\] meaning that with probability $1-\delta$, sample from $D_0$ and with probability $\delta$, sample from $D_1$.

\subsection{Quantum Background}

An $n$-qubit register $\regX$ is a named Hilbert space $\bbC^{2^n}$. A pure quantum state on register $\regX$ is a unit vector $\ket{\psi}^{\regX} \in \bbC^{2^n}$. A mixed state on register $\regX$ is described by a density matrix $\rho^{\regX} \in \bbC^{2^n \times 2^n}$, which is a positive semi-definite Hermitian operator with trace 1. 

A \emph{quantum operation} $F$ is a completely-positive trace-preserving (CPTP) map from a register $\regX$ to a register $\regY$, which in general may have different dimensions. That is, on input a density matrix $\rho^{\regX}$, the operation $F$ produces $F(\rho^{\regX}) = \tau^{\regY}$ a mixed state on register $\regY$. A \emph{unitary} $U: \regX \to \regX$ is a special case of a quantum operation that satisfies $U^\dagger U = U U^\dagger = \cI^{\regX}$, where $\cI^{\regX}$ is the identity matrix on register $\regX$. A \emph{projector} $\Pi$ is a Hermitian operator such that $\Pi^2 = \Pi$, and a \emph{projective measurement} is a collection of projectors $\{\Pi_i\}_i$ such that $\sum_i \Pi_i = \cI$. Throughout this work, we will often write an expression like $\Pi\ket{\psi}$, where $\ket{\psi}$ has been defined on some multiple registers, say $\regX$, $\regY$, and $\regZ$, and $\Pi$ has only been defined on a subset of these registers, say $\regY$. In this case, we technically mean $(\cI^\regX \otimes \Pi \otimes \cI^\regZ)\ket{\psi}$, but we drop the identity matrices to reduce notational clutter. 

A family of quantum circuits is in general a sequence of quantum operations $\{C_\secp\}_{\secp \in \bbN}$, parameterized by the security parameter $\secp$. We say that the family is \emph{quantum polynomial time} (QPT) if $C_\secp$ can be implemented with a $\poly(\secp)$-size quantum circuit. A family of \emph{oracle-aided} quantum circuits $\{C^\sF_\secp\}_{\secp \in \bbN}$ has access to an oracle $\sF: \{0,1\}^* \to \{0,1\}^*$ that implements some deterministic classical map. That is, $C_\secp$ can apply a unitary that maps $\ket{x}\ket{y} \to \ket{x}\ket{y \oplus \sF(x)}$. We say that the family is \emph{quantum polynomial query} (QPQ) if $C_\secp$ only makes $\poly(\secp)$-many queries to $\sF$, but is otherwise computationally unbounded.

Let $\Tr$ denote the trace operator. For registers $\regX,\regY$, the \emph{partial trace} $\Tr^{\regY}$ is the unique operation from $\regX,\regY$ to $\regX$ such that for all $(\rho,\tau)^{\regX,\regY}$, $\Tr^{\regY}(\rho,\tau) = \Tr(\tau)\rho$. The \emph{trace distance} between states $\rho,\tau$, denoted $\TD(\rho,\tau)$ is defined as \[\TD(\rho,\tau) \coloneqq \frac{1}{2}\Tr\left(\sqrt{(\rho-\tau)^\dagger(\rho-\tau)}\right).\] The trace distance between two states $\rho$ and $\tau$ is an upper bound on the probability that any (unbounded) algorithm can distinguish $\rho$ and $\tau$.

For any set $S$, we define $O[S]$ to be the boolean function that checks for membership in $S$ and define the projector $$\Pi[S] = \sum_{s \in S}\ketbra{s}{s}.$$

\begin{definition}[Quantum Program]\label{def:quantum-imp}
    A quantum implementation of a functionality with classical inputs and outputs, or, a \emph{quantum program}, is a pair $(\ket{\psi},C)$, where $\ket{\psi}$ is a state and $C$ is the classical description of a quantum circuit. For any classical input $x \in \{0,1\}^{m}$, we write $y \gets C(\ket{x}\ket{\psi})$ to denote the result of running $C$ and then measuring a dedicated $m'$-qubit output register in the standard basis to obtain $y$.
    
    \begin{itemize}
        \item We say that the program is \emph{deterministic} if for all $x$, there exists $y \in \{0,1\}^{m'}$ such that \[\Pr[C(\ket{x}\ket{\psi}) = y] = 1.\]
        \item We say that a family of quantum programs $\{(\ket{\psi_\secp},C_\secp)\}_{\secp \in \bbN}$ is $\epsilon-$\emph{pseudo-deterministic} for some $\epsilon = \epsilon(\secp)$ if for all sequences of inputs $\{x_\secp\}_{\secp \in \bbN}$, there exists a sequence of outputs $\{y_\secp\}_{\secp \in \bbN}$ such that \[\Pr[C_\secp(\ket{x_\secp}\ket{\psi_\secp}) \to y_\secp] \geq 1-\epsilon(\secp).\]

    \end{itemize}

    We will often leave the dependence on $\secp$ implicit, and just refer to (pseudo)-deterministic programs $(\ket{\psi},C)$. We will denote by $Q(x)$ the string $y$ such that $\Pr[C(\ket{x}\ket{\psi}) \to y] \geq 1-\epsilon(\secp)$, and refer to $Q$ as the \emph{map induced by $(\ket{\psi},C)$}.
\end{definition}

\subsection{Useful Lemmas}

\begin{lemma}[Gentle measurement \cite{DBLP:journals/tit/Winter99}]\label{lemma:gentle-measurement}
Let $\rho$ be a quantum state and let $(\Pi,\cI-\Pi)$ be a projective measurement such that $\Tr(\Pi\rho) \geq 1-\delta$. Let \[\rho' = \frac{\Pi\rho\Pi}{\Tr(\Pi\rho)}\] be the state after applying $(\Pi,\cI-\Pi)$ to $\rho$ and post-selecting on obtaining the first outcome. Then, $\TD(\rho,\rho') \leq 2\sqrt{\delta}$.
\end{lemma}

\begin{lemma}[Pauli Twirl over Affine Subspaces]\label{lemma:Pauli-twirl}
Let $R,\widehat{R} \subseteq \bbF_2^n$ be subspaces of $\bbF_2^n$, and let $(x_0,z_0,x_1,z_1)$ be such that either $x_0 \oplus x_1 \notin \widehat{R}^\bot$ or $z_0 \oplus z_1 \notin R^\bot$. Then for any $\Delta \notin R, \widehat{\Delta} \notin \widehat{R}$ and density matrix $\rho$ on $m \geq n$ qubits,
\[\sum_{x \in R + \Delta, z \in \widehat{R} + \widehat{\Delta}}(Z^z X^x)(X^{x_0}Z^{z_0})(X^xZ^z)\rho(Z^z X^x)(Z^{z_1}X^{x_1})(X^x Z^z) = 0.\]
\end{lemma}

\begin{proof}
Using the fact that $X^x Z^z = (-1)^{x \cdot z} Z^z X^x$, we write
\begin{align*}
&\sum_{x \in R + \Delta, z \in \widehat{R} + \widehat{\Delta}}(Z^z X^x)(X^{x_0}Z^{z_0})(X^xZ^z)\rho(Z^z X^x)(Z^{z_1}X^{x_1})(X^x Z^z) \\
&= \sum_{x \in R + \Delta, z \in \widehat{R}+ \widehat{\Delta}} (-1)^{x \cdot z_0 + z \cdot x_0 + x \cdot z_1 + z \cdot x_1}(X^{x_0}Z^{z_0})\rho(Z^{z_1} X^{x_1})\\
&= \left(\sum_{z \in \widehat{R} + \widehat{\Delta}}(-1)^{z \cdot (x_0 \oplus x_1)}\right)\left(\sum_{x \in R + \Delta}(-1)^{x \cdot (z_0 \oplus z_1)}\right)(x^{x_0}Z^{z_0})\rho(Z^{z_1} X^{x_1})\\ 
&= \left(\sum_{z \in \widehat{R}}(-1)^{z \cdot (x_0 \oplus x_1)}\right)\left(\sum_{x \in R}(-1)^{x \cdot (z_0 \oplus z_1)}\right)(-1)^{\widehat{\Delta} \cdot (x_0 \oplus x_1)}(-1)^{\Delta \cdot (z_0 \oplus z_1)}(x^{x_0}Z^{z_0})\rho(Z^{z_1} X^{x_1})\\
&= 0,
\end{align*}
where the last equality follows because either $x_0 \oplus x_1 \notin \widehat{R}^\bot$ or $z_0 \oplus z_1 \notin R^\bot$.
\end{proof}

The following two lemmas are applications of Cauchy-Schwarz. The first is adapted from \cite{cryptoeprint:2022/786}.

\begin{lemma}\label{lemma:map-distinguish}
    Let $\cK$ be a set of keys, $N$ an integer, and $\{\ket{\psi_k},\{\Pi_{k,i}\}_{i \in [N]},O_k\}_{k \in \cK}$ be a set of states $\ket{\psi_k}$, projective submeasurements $\{\Pi_{k,i}\}_{i \in [N]}$, and classical functions $O_k$ such that $\ket{\psi_k} \in \mathsf{Im}(\sum_i \Pi_{k,i})$ for each $k$. Then for any distinguisher $D$, which we take to be an oracle-aided binary outcome projector, it holds that 
    \[\E_{k \gets \cK}\left[\|D^{O_k}\ket{\psi_k}\|^2\right] - \sum_i\E_{k \gets \cK}\|D^{O_k}\Pi_{k,i}\ket{\psi_k}\|^2 \leq N \cdot \left[\sum_{i \neq j}\E_{k \gets \cK}\|\Pi_{k,j}D^{O_k}\Pi_{k,i}\ket{\psi_k}\|^2\right]^{1/2}.\]
\end{lemma}

\begin{proof}
\begingroup
\allowdisplaybreaks
\begin{align*}
&\E_{k \gets \cK}\left[\bigg\| D^{O_k}\ket{\psi_k}\big\|^2\right] \\
&=\E_{k \gets \cK}\left[\bigg| \bra{\psi_k}\left(\sum_i \Pi_{k,i}\right)D^{O_k}\left(\sum_i \Pi_{k,i}\right)\ket{\psi_k}\bigg|\right]\\
&=\E_{k \gets \cK}\left[\bigg| \sum_i \bra{\psi_k}\Pi_{k,i}D^{O_k}\Pi_{k,i}\ket{\psi_k} + \sum_{i \neq j}\bra{\psi_k}\Pi_{k,j}D^{O_k}\Pi_{k,i}\ket{\psi_k}\bigg|\right]\\
&\leq \sum_i \E_{k \gets \cK}\big\| D^{O_k}\Pi_{k,i}\ket{\psi_k}\big\|^2 + \E_{k \gets \cK}\left[\sum_{i \neq j}\bigg| \bra{\psi_k}\Pi_{k,j}D^{O_k}\Pi_{k,i}\ket{\psi_k}\bigg|\right] \\
&\leq \sum_i \E_{k \gets \cK}\big\| D^{O_k}\Pi_{k,i}\ket{\psi_k}\big\|^2  + \E_{k \gets \cK}\left[\sum_{i \neq j}\sqrt{\bra{\psi_k}\Pi_{k,i}D^{O_k}\Pi_{k,j}D^{O_k}\Pi_{k,i}\ket{\psi_k}}\right] \\
&\leq \sum_i \E_{k \gets \cK}\big\| D^{O_k}\Pi_{k,i}\ket{\psi_k}\big\|^2  + N \cdot \E_{k \gets \cK}\left[\left(\sum_{i \neq j}\|\Pi_{k,j}D^{O_k}\Pi_{k,i}\ket{\psi_k}\|^2\right)^{1/2}\right] \\
&\leq \sum_i \E_{k \gets \cK}\big\| D^{O_k}\Pi_{k,i}\ket{\psi_k}\big\|^2  + N \cdot \left[\sum_{i \neq j}\E_{k \gets \cK}\|\Pi_{k,j}D^{O_k}\Pi_{k,i}\ket{\psi_k}\|^2\right]^{1/2},
\end{align*}
\endgroup

where the first inequality is the triangle inequality, the second follows from Cauchy-Schwarz applied to vectors $\ket{\psi_k}$ and $\Pi_{k,j}D^{O_k}\Pi_{k,i}\ket{\psi}$, the third follows from Cauchy-Shwarz applied to the length $N^2$ vector $(1,\dots,1)$ and the vector with $(i,j)$'th entry equal to \[\sqrt{\bra{\psi_k}\Pi_{k,i}D^{O_k}\Pi_{k,j}D^{O_k}\Pi_{k,i}\ket{\psi_k}},\] and the fourth is Jensen's inequality.

\end{proof}

\begin{lemma}\label{lemma:project}
    Let $\cK$ be a set of keys, $N$ an integer, and $\{\ket{\psi_k},\{\Pi_{k,i}\}_{i \in [N]},O_k,\Gamma_k\}_{k \in \cK}$ be a set of states $\ket{\psi_k}$, projective submeasurements $\{\Pi_{k,i}\}_{i \in [N]}$, classical function $O_k$, and projective measurements $\Gamma_k$ such that $\ket{\psi_k} \in \mathsf{Im}(\sum_i \Pi_{k,i})$ for each $k$. Then for any oracle-aided unitary $U$, it holds that
    \[\E_{k \gets \cK}[\|\Gamma_k U^{O_k}\ket{\psi_k}\|^2] \leq N \cdot \sum_i\E_{k \gets \cK}\|\Gamma_k U^{O_k}\Pi_{k,i}\ket{\psi_k}\|^2.\] %
\end{lemma}

\begin{proof}
\begin{align*}
    &\E_{k \gets \cK}[\|\Gamma_k U^{O_k}\ket{\psi_k}\|^2]\\
    &\leq \E_{k \gets \cK}\left[\left(\sum_i\|\Gamma_k U^{O_k}\Pi_{k,i}\ket{\psi_k}\|\right)^2\right]\\
    &\leq \E_{k \gets \cK}\left[\left(\sqrt{N}\sqrt{\sum_i\|\Gamma_k U^{O_k}\Pi_{k,i}\ket{\psi_k}\|^2}\right)^2\right]\\
    &=N \cdot \sum_i\E_{k \gets \cK}\|\Gamma_k U^{O_k}\Pi_{k,i}\ket{\psi_k}\|^2,
\end{align*}

where the first inequality is the triangle inequality, and the second follows from Cauchy-Shwarz applied to the length $N$ vector $(1,\dots,1)$ and the vector with $i$'th entry equal to $\|\Gamma_k U^{O_k}\Pi_{k,i}\ket{\psi_k}\|$.

\end{proof}

We will frequently invoke the following lemma in order to switch between two oracles that differ on hard-to-find inputs. The proof is a standard oracle hybrid argument.

\begin{lemma}\label{lemma:puncture}
    For each $\secp \in \bbN$, let $\cK_\secp$ be a set of keys, and $\{\ket{\psi_k},O_k^0,O_k^1,S_k\}_{k \in \cK_\secp}$ be a set of states $\ket{\psi_k}$, classical functions $O_k^0,O_k^1$, and sets of inputs $S_k$. Suppose that the following properties holds.
    \begin{enumerate}
        \item The oracles $O_k^0$ and $O_k^1$ are identical on inputs outside of $S_k$.
        \item For any oracle-aided unitary $U$ with $q = q(\secp)$ queries, there is some $\epsilon = \epsilon(\secp)$ such that 
        \[\E_{k \gets \cK}\left[\big\| \Pi[S_k]U^{O_k^0}\ket{\psi_k}\big\|^2\right] \leq \epsilon.\]
    \end{enumerate}
    Then, for any oracle-aided unitary $U$ with $q(\secp)$ queries and distinguisher $D$, \[\bigg| \Pr_{k \gets \cK}\left[D\left(k,U^{O_k^0}\ket{\psi_k}\right) = 0\right] - \Pr_{k \gets \cK}\left[D\left(k,U^{O_k^1}\ket{\psi_k}\right) = 0\right]\bigg| \leq 4q\sqrt{\epsilon}.\]
\end{lemma}

\begin{proof}

For each $i \in [0,\dots,q]$, define hybrid $\cH_i$ to sample $k \gets \cK$ and output $(k,U^{(\cdot)}\ket{\psi_k})$, where $U$'s first $q - i$ oracle queries are answered with $O_k^0$ and $U$'s final $i$ oracle queries are answered with $O_k^1$. For each $i \in [0,\dots,q-1]$, define hybrid $\cH_i'$ to be identical to $\cH_i$ except that we apply the measurement $\{\Pi[S_k],\cI-\Pi[S_k]\}$ to $U$'s state right before the $q-i$'th oracle query, and post-select on obtaining the second outcome. Then for any $i \in [0,\dots,q-1]$, 
\begin{itemize}
    \item By condition 2 of the lemma statement and \cref{lemma:gentle-measurement}, it holds that $\TD(\cH_i,\cH_i') \leq 2\sqrt{\epsilon}$.
    \item By conditions 1 and 2 of the lemma statement and \cref{lemma:gentle-measurement}, it holds that $\TD(\cH_i'
    ,\cH_{i+1}) \leq 2\sqrt{\epsilon}$.
\end{itemize}

Thus, the lemma follows by summing over the $2q$ hybrid switches.

\end{proof}

\subsection{Signature Tokens}\label{subsec:sig-tokens}

A signature token scheme consists of algorithms $(\TokGen,\TokSign,\TokVer)$ with the following syntax. 

\begin{itemize}
    \item $\TokGen(1^\secp) \to (\vk,\ket{\sk})$: The $\TokGen$ algorithm takes as input the security parameter $1^\secp$ and outputs a classical verification key $\vk$ and a quantum signing key $\ket{\sk}$.
    \item $\TokSign(b,\ket{\sk}) \to \sigma$: The $\TokSign$ algorithm takes as input a bit $b \in \{0,1\}$ and the signing key $\ket{\sk}$, and outputs a classical signature $\sigma$.
    \item $\TokVer(\vk,b,\sigma) \to \{\top,\bot\}$: The $\TokVer$ algorithm takes as input a verification key $\vk$, a bit $b$, and a signature $\sigma$, and outputs $\top$ or $\bot$.
\end{itemize}

A signature token should satisfy the following definition of correctness.  

\begin{definition}\label{def:token-correctness}
A signature token scheme $(\TokGen,\TokSign,\TokVer)$ is \emph{correct} if for any $b \in \{0,1\}$,
\[\Pr\left[\TokVer(\vk,b,\sigma) = \top : \begin{array}{r}(\vk,\ket{\sk}) \gets \Gen(1^\secp) \\ \sigma \gets \Sign(b,\ket{\sk})\end{array}\right] = 1.\]
\end{definition}

A signature token should satisfy the following definition of security. Note that we give the adversary oracle access to the verification functionality, and ask for exponential security.

\begin{definition}\label{def:unforgeability}
A signature token scheme $(\TokGen,\TokSign,\TokVer)$ satisfies \emph{unforgeability} if for any QPQ adversary $\{\sA_\secp\}_{\secp \in \bbN}$,

\[\Pr\left[\begin{array}{l} \TokVer(\vk,0,\sigma_0) = \top ~~ \wedge \\ \TokVer(\vk,1,\sigma_1) = \top\end{array} : \begin{array}{r} (\vk,\ket{\sk}) \gets \Gen(1^\secp) \\ (\sigma_0,\sigma_1) \gets \sA_\secp^{\TokVer[\vk]}(\ket{\sk}) \end{array}\right] = 2^{-\Omega(\secp)},\] where $\TokVer[\vk]$ is the functionality $\TokVer(\vk,\cdot,\cdot)$.
\end{definition}

\begin{importedtheorem}[\cite{arxiv:BenSat16}]
There exists a signature token scheme that satisfies \cref{def:token-correctness} and \cref{def:unforgeability}. 
\end{importedtheorem}

\section{Authentication Scheme}\label{sec:authentication}

In this section, we introduce the notion of a ``publicly-verifiable linearly-homomorphic QAS (Quantum Authentication Scheme) with classically-decodable ZX measurements.'' We then provide a construction and security proof.

\subsection{Definitions}\label{subsec:auth-def}

The following notation will be heavily referenced both throughout this section, and throughout the remainder of the paper.

\paragraph{Partial ZX measurements.} Given a string $\theta \in \{0,1,\bot\}^n$, define sets \[\MS_\theta \coloneqq \{i : \theta_i \neq \bot\}, \ \ \MS_{\theta,0} = \{i : \theta_i = 0\}, \ \ \MS_{\theta,1} = \{i : \theta_i = 1\}, \ \ \MS_{\theta,\bot} \coloneqq \{i : \theta_i = \bot\}.\] We will often write $\MS,\MS_0,\MS_1,\MS_\bot$ instead of $\MS_\theta,\MS_{\theta,0},\MS_{\theta,1},\MS_{\theta,\bot}$ when the choice of $\theta$ is clear from context. The string $\theta$ will be used to denote the basis of a partial measurement on $n$ qubits, where the 0 indices are measured in the standard basis and the 1 indices are measured in the Hadamard basis. We will also need the following notation.

\begin{itemize}
    \item Given a string $m \in \{0,1\}^n$ and a set $\MS \subseteq [n]$, let $m_\MS$ denote the substring of $m$ consisting of bits $\{m_i\}_{i \in \MS}$. 
    \item Given $\theta \in \{0,1,\bot\}^n$ and a set $\MS \subseteq [n]$, define $\theta[\MS]$ to be equal to $\theta$ on indices in $V$ and $\bot$ everywhere else.
    \item Given a register $\regM$, an operation $U$ on $\regM$, and a subset $\MS \subseteq [n]$, let $U^\MS$ be the operation on $\regM^{\otimes n}$ that applies $U$ to the $i$'th copy of $\regM$ for each $i \in \MS$.
\end{itemize}

For any $\theta \in \{0,1,\bot\}^n$ and classical function $f : \{0,1\}^{|\MS|} \to \{0,1\}^*$, let $M_{\theta,f}$ be the projective measurement on $n$ qubits defined by the operators \[\left\{H^{\MS_1}\left(\sum_{m : f\left(m_\MS\right) = y}\ketbra{m}{m}\right)H^{\MS_1}\right\}_{y}.\]  For any $n$-qubit register $\regM$, we write \[M_{\theta,f}(\regM) \to \regM,y\] to refer to the operation that measures $\regM$ according to $M_{\theta,f}$ and then writes the classical result $y$ to a new register. Sometime we will write $M_{\theta,f}(\regM) \to y$ to denote just the classical measurement result $y$.

\paragraph{Linear operations.} We will use $L$ to denote a sequence of CNOT gates on $n$ qubits, which we refer to as a \emph{linear operation}. While all quantum gates are linear with respect to the ambient Hilbert space of exponential dimension, here linearity specifically refers to the fact that any sequence of CNOT gates applies a linear function over $\bbF_2$ to each standard basis vector. In an abuse of notation, $L$ will either refer to the classical description of a series of CNOT gates or to the actual unitary operation that applies these gates. Which case should be clear from context.

\paragraph{Syntax.} A publicly-verifiable, linearly-homomorphic quantum authentiation scheme (QAS) with classically-decodable ZX measurements has the following syntax. Let $p = p(\secp)$ be a polynomial.

\begin{itemize}
  \item $\mathsf{Gen}(1^\secp,n) \rightarrow k$: The key generation algorithm takes as input a security parameter $1^\secp$ and number of qubits $n = \poly(\secp)$, and outputs an authentication key $k$.
  \item $\mathsf{Enc}_k(\regM)\to\regC$: The encoding algorithm is an isometry parameterized by an authentication key $k$ that maps a state on an $n$-qubit register $\regM \coloneqq \regM_1 \otimes \dots \otimes \regM_n$ to a state on an $np$-qubit register $\regC \coloneqq \regC_1 \otimes \dots \otimes \regC_n$, where each $\regC_i$ is an $p$-qubit register.
  \item $\LinEval_L(\regC) \to \regC$: The linearly-homomorphic evaluation procedure is a unitary parameterized by a linear operation $L$ that operates on register $\regC$.
  \item $\mathsf{Dec}_{k,L,\theta}\left(c\right) \rightarrow m \cup \{\bot\}$: The classical decoding algorithm is parameterized by an authentication key $k$, a linear operation $L$, and a choice of bases $\theta \in \{0,1,\bot\}^{n}$. It takes as input a string $c \in \{0,1\}^{|\MS| \cdot p}$ and outputs either a classical string $m \in \{0,1\}^{|\MS|}$ or $\bot$.
  \item $\Ver_{k,L,\theta}(c) \to \{\top,\bot\}$: The classical verification algorithm is identical to $\Dec$ except that whenever $\Dec$ outputs an $m \neq \bot$, $\Ver$ outputs $\top$.
\end{itemize}

\paragraph{Partial ZX measurements on authenticated states.} First, given the parameter $p$, define \[\widetilde{\MS} \coloneqq \bigcup_{i \in \MS} \{(i-1)p +1, \dots, ip\} \subseteq [np].\] That is, $\widetilde{\MS}$ contains the $i$'th chunk of $p$ indices for each $i \in \MS$. Define $\widetilde{\MS}_0, \widetilde{\MS}_1, \widetilde{\MS}_\bot$ analogously. For any $\theta \in \{0,1,\bot\}^n$, classical function $f : \{0,1\}^{|\MS|} \to \{0,1\}^*$, authentication key $k$, and linear operation $L$, let $\widetilde{M}_{\theta,f,k,L}$ be the projective measurement on $np$ qubits defined by the operators \[\left\{H^{\widetilde{\MS}_1} \left(\sum_{c : f\left(\Dec_{k,L,\theta}\left(c_{\widetilde{\MS}}\right)\right) = y}\ketbra{c}{c}\right)H^{\widetilde{\MS}_1}\right\}_{y} \cup \left\{H^{\widetilde{\MS}_1}\left(\sum_{c:\Dec_{k,L,\theta}\left(c_{\widetilde{\MS}}\right) = \bot}\ketbra{c}{c}\right)H^{\widetilde{\MS}_1}\right\}.\]  For any $np$-qubit register $\regC$, we write \[\widetilde{M}_{\theta,f,k,L}(\regC) \to \regC,y\] to refer to the operation that measures $\regC$ according to $M_{\theta,f,k,L}$ and then writes the classical result $y$ to a new register. Sometimes we will write $\widetilde{M}_{\theta,f,k,L}(\regC) \to y$ to denote just the classical measurement result $y$.

\paragraph{Correctness.} Our definition of correctness roughly states that encoding, applying a linear homomorphism, and then applying a partial measurement to the encoded state is equivalent to first applying the linear operation, applying the partial measurement, and then encoding. This definition supports composition of multiple partial measurements on encoded data, which will be necessary for our application to obfuscation.

\begin{definition}[Correctness]\label{def:auth-correct}
    A publicly-verifiable linearly-homomorphic QAS with classically-decodable ZX measurements is \emph{correct} if the following holds. For any linear operation $L$, bases $\theta \in \{0,1,\bot\}^{n}$, $f: \{0,1\}^{|\MS|} \to \{0,1\}^*$, and $k \in \Gen(1^\secp,n)$, 
    
    \begin{align*}
    &\LinEval_{L}^\dagger \circ \widetilde{M}_{\theta,f,k,L} \circ \LinEval_L \circ \Enc_k  = \Enc_k \circ L^\dagger \circ M_{\theta,f} \circ L.
    \end{align*}

    Note that both sequences of operations above map $\regM \to (\regC,y)$, where $\regM$ is an $n$-qubit register, $\regC$ is an $np$-qubit register, and $y$ is a classical measurement outcome.

\end{definition}

\paragraph{Security} Next, we formalize two security properties. The first roughly states that no adversary with access to the verification oracle can change the distribution resulting from a partial measurement on the encoded state.

\begin{definition}[Security]\label{def:auth-sec}
    A publicly-verifiable linearly-homomorphic QAS with classically-decodable ZX measurements is \emph{secure} if the following holds. For any linear operation $L$, bases $\theta \in \{0,1,\bot\}^n$, $f : \{0,1\}^{|\MS|} \to \{0,1\}^*$, and oracle-aided adversary $A : \regC \to \regC$, there exists an $\epsilon(\secp) \in [0,1]$ such that for any $n$-qubit state $\ket{\psi}$,

    \[\left\{y : \begin{array}{r} k \gets \Gen(1^\secp,n) \\ y \gets \widetilde{M}_{\theta,f,k,L} \circ A^{\Ver_{k,\cdot,\cdot}(\cdot)} \circ \Enc_k(\regM)\end{array}\right\} \approx_{2^{-\Omega(\secp)}} (1-\epsilon(\secp))\left\{y : y \gets M_{\theta,f} \circ L(\regM)\right\} + \epsilon(\secp)\{\bot\}.\]

\end{definition}

\begin{remark}\label{remark:unitary}
Although the adversary $A$ is defined as a (oracle-aided) general quantum map from $\regC \to \regC$, we can without loss of generality take it to be a (oracle-aided) unitary that additionally operates on some workspace register $\regA$ initialized to $\ket{0}$. We leave the workspace register $\regA$ implicit when writing $y \gets \widetilde{M}_{\theta,f,k,L} \circ A^{\Ver_{k,\cdot,\cdot}(\cdot)} \circ \Enc_k(\regM)$, and note that $ \widetilde{M}_{\theta,f,k,L}$ only operates on $\regC$.
\end{remark}

Next, we describe a weaker security property that is immediately implied by \cref{def:auth-sec}, but will be convenient to use in our application to obfuscation.

\begin{definition}[Mapping Security]\label{def:auth-mapping-sec}
    For any linear operation $L$, bases $\theta \in \{0,1,\bot\}^n$, $f : \{0,1\}^{|\MS|} \to \{0,1\}^*$, $n$-qubit state $\ket{\psi}$, and set $B \subset \{0,1\}^*$ such that \[\Pr[y \in B : y \gets M_{\theta,f} \circ L (\ket{\psi})] = 0,\] it holds that for any oracle-aided adversary $A : \regC \to \regC$, \[\Pr\left[y \in B : \begin{array}{r} k \gets \Gen(1^\secp,1^n) \\ y \gets \widetilde{M}_{\theta,f,k,L} \circ A^{\Ver_{k,\cdot,\cdot}(\cdot)} \circ \Enc_k (\ket{\psi})\end{array}\right] = 2^{-\Omega(\secp)}.\]
\end{definition}

Finally, we define a notion of privacy, which states that any two encoded states are indistinguishable, even given the verification oracle.

\begin{definition}[Privacy]\label{def:auth-privacy}
For any $n$-qubit states $\ket{\psi_0},\ket{\psi_1
}$ and oracle-aided binary outcome projector $D$, 
\[\bigg|\Pr_{k \gets \Gen(1^\secp,n)}\left[1 \gets D^{\Ver_{k,\cdot,\cdot}(\cdot)} \circ \Enc_k(\ket{\psi_0})\right] - \Pr_{k \gets \Gen(1^\secp,n)}\left[1 \gets D^{\Ver_{k,\cdot,\cdot}(\cdot)} \circ \Enc_k(\ket{\psi_1})\right]\bigg| = 2^{-\Omega(\secp)}.\]
\end{definition}

\subsection{Construction}

\paragraph{Paulis and updates.} We specify several notational conventions regarding sets $\{x_i\}_{i \in [n]}, \{z_i\}_{i \in [n]}$ that describe Pauli corrections on $n$ registers.

\begin{itemize}
    \item As $n$ will be clear from context, let $x \coloneqq (x_1,\dots,x_n)$ and $z \coloneqq (z_1,\dots,z_n)$.
    \item Given a linear operation $L$ on $n$ qubits, let $L(x,z) \coloneqq (x_L,z_L)$ be the result of starting with $(x,z)$, and, for each  CNOT gate in $L$, sequentially applying the CNOT update rule $(x_i,z_i),(x_j,z_j) \to (x_i, z_i \oplus z_j), (x_i \oplus x_j, z_j)$. Note that this is yet another interpretation for $L$, which in another context could refer to the unitary that applies the sequence of CNOT gates.
    \item Let $L^{-1}$ be the inverse of $L$, and note that $L^{-1}(x_L,z_L) = (x,z)$.
    \item Given $x = (x_1,\dots,x_n)$ or $x_L = (x_{L,1},\dots,x_{L,n})$ and a subset $\MS \subseteq [n]$, let $x_\MS \coloneqq \{x_i\}_{i \in \MS}$ and $x_{L,\MS} \coloneqq \{x_{L,i}\}_{i \in \MS}$. Given disjoint sets $\MS_0,\MS_1 \subset [n]$, we let $x_{\MS_0},x_{\MS_1}$ refer to the union $\{x_i\}_{i \in \MS_0} \cup \{x_i\}_{i \in \MS_1}$.
\end{itemize}

\paragraph{Subspaces.}
Given a $\secp$-dimensional subspace $S \subset \bbF_2^{2\secp+1}$ and a vector $\Delta \in \bbF_2^{2\secp+1} \setminus S$, define the $(\secp+1)$-dimensional subspace $$S_\Delta \coloneqq S \cup (S+\Delta).$$ 
Let the dual subspace of $S_{\Delta}$ be $\widehat{S} \coloneqq S_\Delta^\bot$; note that 
\begin{itemize}
    \item $\widehat{S}$ is $\secp$-dimensional; and
    \item since $S_{\Delta} \supset S$, its dual subspace $\widehat{S} := S_{\Delta}^\perp \subset S^\perp$.
\end{itemize}
\noindent
Let $\widehat{\Delta}$ be an arbitrary choice of a vector such that $S^\bot = \widehat{S} \cup (\widehat{S} + \widehat{\Delta})$, and define $$\widehat{S}_{\widehat{\Delta}} \coloneqq S^\bot = \widehat{S} \cup (\widehat{S} + \widehat{\Delta})~.$$

Given a subspace $S$, define the state \[\ket{S} \coloneqq \frac{1}{\sqrt{|S|}}\sum_{s \in S}\ket{s},\] and note that \[H^{\otimes 2\secp+1}\ket{S} = \ket*{S^\bot}.\]

Next, given any $\secp$-dimensional subspace $S$ and $\Delta \notin S$, define the isometry $E_{S,\Delta}$ from 1 qubit to $2\secp+1$ qubits that maps $\ket{0} \to \ket{S}$ and $\ket{1} \to \ket{S+\Delta}$.

\protocol{Publicly-verifiable linearly-homomorphic QAS with classically-decodable ZX measurements}{Our construction of a publicly-verifiable linearly-homomorphic QAS with classically-decodable ZX measurements.}{fig:authscheme}{

\begin{itemize}
    \item $\Gen(1^\secp,n)$: Sample a uniformly random $\secp$-dimensional subspace $S \subset \bbF_2^{2\secp+1}$, vector $\Delta \gets \bbF_2^{2\secp+1} \setminus S$, and $x_i,z_i \gets \bbF_2^{2\secp+1}$ for each $i \in [n]$. Output $k \coloneqq (S,\Delta,x,z)$.  
    \item $\Enc_k = X^x Z^z E_{S,\Delta}^{\otimes n}.$ 
    \item $\LinEval_L(\regC)$: Parse register $\regC = \regC_1 \otimes \dots \otimes \regC_n$, where each $\regC_i$ is a $(2\secp+1)$-qubit register. For each CNOT in $L$ from qubit $i$ to $j$, apply $\mathsf{CNOT}^{\otimes 2\secp+1}$ from register $\regC_i$ to $\regC_j$.
    \item $\Dec_{k,L,\theta}(c)$: Parse $c = \{c_i\}_{i \in \MS}$. Define $\{m_i\}_{i \in \MS}$ as follows. \[\forall i \in \MS_0: m_i = \begin{cases}0 \text{ if } c_i \in S + x_{L,i} \\ 1 \text{ if } c_i \in S + \Delta+ x_{L,i} \\ \bot \text{ otherwise}\end{cases} ~~~~~ \forall i \in \MS_1: m_i = \begin{cases}0 \text{ if } c_i \in \widehat{S} + z_{L,i} \\ 1 \text{ if } c_i \in \widehat{S} + \widehat{\Delta}+ z_{L,i} \\ \bot \text{ otherwise}\end{cases}\] If any $m_i = \bot$, then output $\bot$, and otherwise output $m = \{m_i\}_{i \in \MS}$.
    \item $\Ver_{k,L,\theta}(c)$:\footnote{This procedure is already determined by $\Dec_{k,L,\theta}(c)$, but we write it explicitly for clarity in the proof.}  Parse $c = \{c_i\}_{i \in \MS}$. For each $i \in \MS_0$, output $\bot$ if $c_i \notin S_\Delta + x_{L,i}$. For each $i \in \MS_1$, output $\bot$ if $c_i \notin \widehat{S}_{\widehat{\Delta}} + z_{L,i}$. Otherwise, output $\top$.
\end{itemize}
}

\begin{theorem}\label{thm:auth-correct}
    The QAS described in \proref{fig:authscheme} satisfies correctness (\cref{def:auth-correct}).
\end{theorem}

\begin{proof}
First, we show two key claims.

\begin{claim}\label{claim:Hadamard}
For any $S$ and $\Delta$, it holds that $H^{\otimes 2\secp+1}E_{\widehat{S},\widehat{\Delta}} = E_{S,\Delta}H$.
\end{claim}

\begin{proof}
We show that the maps are equivalent by checking their behavior on the basis $\{\ket{+},\ket{-}\}$. First, 

\begin{align*}
H^{\otimes 2\secp+1}E_{\widehat{S},\widehat{\Delta}}\ket{+} = H^{\otimes 2\secp+1}\ket*{\widehat{S}_{\widehat{\Delta}}} = \ket{S} = E_{S,\Delta}\ket{0} = E_{S,\Delta}H\ket{+}.
\end{align*}

Next, 
\begin{align*}
H^{\otimes 2\secp+1}E_{\widehat{S},\widehat{\Delta}}\ket{-} = H^{\otimes 2\secp+1}\left(\ket*{\widehat{S}} - \ket*{\widehat{S} + \widehat{\Delta}}\right)= H^{\otimes 2\secp + 1}Z^{\Delta}\ket*{\widehat{S}_{\widehat{\Delta}}} = \ket*{S+\Delta} = E_{S,\Delta}\ket{1} = E_{S,\Delta}H\ket{-}.
\end{align*}

\end{proof}

\begin{claim}\label{claim:linear}
For any $S,\Delta$, and $L$, $\LinEval_LE_{S,\Delta}^{\otimes n} = E_{S,\Delta}^{\otimes n}L$
\end{claim}

\begin{proof}
We show this for the case where $L$ contains a single CNOT gate, and the full proof follows by applying the argument sequentially. We show that the maps are equivalent by checking their behavior on the basis $\{\ket{b_1,b_2}\}_{b_1,b_2 \in \{0,1\}}$.

\begingroup
\allowdisplaybreaks
\begin{align*}
&\mathsf{CNOT}^{\otimes 2\secp+1}E_{S,\Delta}^{\otimes 2}\ket{b_1,b_2}\\
&= \mathsf{CNOT}^{\otimes 2\secp+1}\ket{S+b_1 \cdot \Delta}\ket{S + b_2 \cdot \Delta}\\
&= \frac{1}{2^{\secp}}\mathsf{CNOT}^{\otimes 2\secp+1} \sum_{s_1 \in S}\ket{s_1 + b_1 \cdot \Delta}\sum_{s_2 \in S}\ket{s_2 + b_2 \cdot \Delta}\\
&= \frac{1}{2^{\secp}} \sum_{s_1 \in S}\ket{s_1 + b_1 \cdot \Delta}\sum_{s_2 \in S}\ket{(s_1 + s_2) + (b_1 + b_2) \cdot \Delta}\\
&= \ket{S+b_1 \cdot \Delta}\ket{S + (b_1 + b_2) \cdot \Delta}\\
&= E_{S,\Delta}^{\otimes 2}\mathsf{CNOT}\ket{b_1,b_2}.
\end{align*}
\endgroup

\end{proof}

Now, define measurements $M'_{\theta,f},\widetilde{M}'_{\theta,f,k,L}$ so that \[M_{\theta,f} = H^{\MS_1}M'_{\theta,f}H^{\MS_1}, ~~ \text{ and } ~~ \widetilde{M}_{\theta,f,k,L} = H^{\widetilde{\MS}_1}X^{x_{L,\MS_0},z_{L,\MS_1}}\widetilde{M}'_{\theta,f,k,L}X^{x_{L,\MS_0},z_{L,\MS_1}}H^{\widetilde{\MS}_1}.\] 

To be concrete,

\[M'_{\theta,f} \coloneqq \left\{\sum_{m : f(m_\MS) = y}\ketbra{m}{m}\right\}_{y},\] and \[\widetilde{M}'_{\theta,f,k,L} \coloneqq \left\{\sum_{m : f(m_\MS) = y}\left(\sum_{\substack{c : \{c_i \in S+m_i \cdot \Delta\}_{i \in \MS_0},\\ \{c_i \in \widehat{S}+m_i \cdot \widehat{\Delta}\}_{i \in \MS_1}}}\ketbra{c}{c}\right)\right\}_{y} \cup \left\{\sum_{\substack{c: \exists i \in \MS_0 \text{ s.t. } c_i \notin S_\Delta\\ \ \ \vee \exists i \in \MS_1 \text{ s.t. }c_i \notin \widehat{S}_{\widehat{\Delta}}}}\ketbra{c}{c}\right\}.\]

Observe that 

\[\widetilde{M}'_{\theta,f,k,L}\left(E_{S,\Delta}^{\otimes |\MS_\bot \cup \MS_0|} \otimes E_{\widehat{S},\widehat{\Delta}}^{\otimes |\MS_1|}\right) = \left(E_{S,\Delta}^{\otimes |\MS_\bot \cup \MS_0|} \otimes E_{\widehat{S},\widehat{\Delta}}^{\otimes |\MS_1|}\right)M'_{\theta,f}.\]

Then,

\begingroup
\allowdisplaybreaks
\begin{align*}
&\LinEval_L^\dagger\widetilde{M}_{\theta,f,k,L}\LinEval_L\Enc_k\\
&=\LinEval_L^\dagger\widetilde{M}_{\theta,f,k,L}\LinEval_L X^xZ^zE_{S,\Delta}^{\otimes n}  \\
&=\LinEval_L^\dagger\widetilde{M}_{\theta,f,k,L}X^{x_L}Z^{z_L} \LinEval_LE_{S,\Delta}^{\otimes n}\\
&=\LinEval_L^\dagger\widetilde{M}_{\theta,f,k,L}X^{x_L}Z^{z_L}E^{\otimes n}_{S,\Delta}L \tag{\cref{claim:linear}}\\
&=\LinEval_L^\dagger H^{\widetilde{\MS}_1}X^{x_{L,\MS_0}z_{L,\MS_1}}\widetilde{M}'_{\theta,f,k,L}X^{x_{L,\MS_0},z_{L,\MS_1}}H^{\widetilde{\MS}_1}X^{x_L}Z^{z_L}E^{\otimes n}_{S,\Delta}L\\
&=\LinEval_L^\dagger H^{\widetilde{\MS}_1}X^{x_{L,\MS_0}z_{L,\MS_1}}\widetilde{M}'_{\theta,f,k,L}X^{x_{L,\MS_\bot}}Z^{z_{L,\MS_\bot},z_{L,\MS_0},x_{L,\MS_1}}H^{\widetilde{\MS}_1}E^{\otimes n}_{S,\Delta}L\\
&=\LinEval_L^\dagger H^{\widetilde{\MS}_1}X^{x_{L,\MS_\bot},x_{L,\MS_0},z_{L,\MS_1}}Z^{z_{L,\MS_\bot},z_{L,\MS_0},x_{L,\MS_1}}\widetilde{M}'_{\theta,f,k,L}H^{\widetilde{\MS}_1}E^{\otimes n}_{S,\Delta}L\\
&=\LinEval_L^\dagger X^{x_L}Z^{z_L}H^{\widetilde{\MS}_1}\widetilde{M}'_{\theta,f,k,L}H^{\widetilde{\MS}_1}E^{\otimes n}_{S,\Delta}L\\
&=\LinEval_L^\dagger X^{x_L}Z^{z_L}H^{\widetilde{\MS}_1}\widetilde{M}'_{\theta,f,k,L}\left(E_{S,\Delta}^{\otimes |\MS_\bot \cup \MS_0|} \otimes E_{\widehat{S},\widehat{\Delta}}^{\otimes |\MS_1|}\right)H^{\MS_1}L \tag{\cref{claim:Hadamard}}\\
&=\LinEval_L^\dagger X^{x_L}Z^{z_L}H^{\widetilde{\MS}_1}\left(E_{S,\Delta}^{\otimes |\MS_\bot \cup \MS_0|} \otimes E_{\widehat{S},\widehat{\Delta}}^{\otimes |\MS_1|}\right)M'_{\theta,f}H^{\MS_1}L \\
&=\LinEval_L^\dagger X^{x_L}Z^{z_L}E^{\otimes n}_{S,\Delta}H^{\MS_1}M'_{\theta,f}H^{\MS_1}L \\
&=X^xZ^z\LinEval_L^\dagger E^{\otimes n}_{S,\Delta}M_{\theta,f}L\\
&=X^xZ^z E^{\otimes n}_{S,\Delta}L^\dagger M_{\theta,f}L\\
&=\Enc_{k}L^\dagger M_{\theta,f}L.
\end{align*}
\endgroup

\end{proof}

\subsection{Security}

\begin{theorem}\label{thm:auth-sec}
    The QAS described in \proref{fig:authscheme} satisfies security (\cref{def:auth-sec}).
\end{theorem}

\begin{proof} We begin by modifying the $\Gen$ procedure and $\Ver_{k,\cdot,\cdot}(\cdot)$ oracle, and arguing that the output of the experiment remains (almost) unaffected. In particular, we will "expand" the verification oracle with random superspaces $R \supset S_\Delta$ and $\widehat{R} \supset \widehat{S}_{\widehat{\Delta}}$. Consider the following procedures.

\begin{itemize}
    \item $\Gen'(1^\secp,n)$: Sample a uniformly random $\secp$-dimensional subspace $S \subset \bbF_2^{2\secp+1}$, vector $\Delta \gets \bbF_2^{2\secp+1} \setminus S$, and $x_i,z_i \gets \bbF_2^{2\secp+1}$ for each $i \in [n]$. Sample uniformly random $(3\secp/2 + 1)$-dimensional subspaces $R,\widehat{R} \subset \bbF_2^{2\secp+1}$ conditioned on $S_\Delta \subset R$ and $\widehat{S}_{\widehat{\Delta}} \subset \widehat{R}$. Output $k \coloneqq (S,\Delta,x,z)$ along with $(R,\widehat{R})$.
    \item $\Ver_{k,R,\widehat{R},L,\theta}'(c)$: Parse $c = \{c_i\}_{i \in \MS}$. For each $i \in \MS_0$, output $\bot$ if $c_i \notin R + x_{L,i}$. For each $i \in \MS_1$, output $\bot$ if $c_i \notin \widehat{R} + z_{L,i}$. Otherwise, output $\top$.
\end{itemize}

\begin{claim}\label{claim:increase-subspaces} 
For any $L,\theta,f,A$, and $\ket{\psi}$, it holds that 
\begin{align*}
    \left\{y : \begin{array} {r} k \gets \Gen(1^\secp,n) \\ y \gets \widetilde{M}_{\theta,f,k,L} \circ A^{\Ver_{k,\cdot,\cdot}(\cdot)} \circ \Enc_k(\ket{\psi})\end{array}\right\} \approx_{2^{-\Omega(\secp)}} \left\{y : \begin{array} {r} (k,R,\widehat{R}) \gets \Gen'(1^\secp,n) \\ y \gets \widetilde{M}_{\theta,f,k,L} \circ A^{\Ver'_{k,R,\widehat{R},\cdot,\cdot}(\cdot)} \circ \Enc_k(\ket{\psi})\end{array}\right\}.
\end{align*}
\end{claim}

\begin{proof} Note that these distributions can be sampled by a reduction given oracle access to either $(O[S_\Delta],O[\widehat{S}_{\widehat{\Delta}}])$ or $(O[R],O[\widehat{R}])$. Now, for any fixed $(\secp+1)$-dimensional subspaces $S_\Delta, \widehat{S}_\Delta$ and any vector $v$, \[\Pr_{R,\widehat{R}}[v \in R \setminus S_\Delta \cup \widehat{R} \setminus \widehat{S}_{\widehat{\Delta}}] \leq \frac{|R \setminus S_\Delta|}{|\bbF_2^{2\secp+1} \setminus S_\Delta |} + \frac{|\widehat{R} \setminus \widehat{S}_{\widehat{\Delta}}|}{|\bbF_2^{2\secp+1} \setminus \widehat{S}_{\widehat{\Delta}} |} = 2 \cdot \frac{2^{3\secp/2+1} - 2^{\secp+1}}{2^{2\secp+1} - 2^{\secp+1}} = 2^{-\Omega(\secp)},\] where the probability is over sampling random $(3\secp/2+1)$-dimensional subspaces $R,\widehat{R}$ conditioned on $S_\Delta \subset R$ and $\widehat{S}_{\widehat{\Delta}} \subset \widehat{R}$. Then the claim follows by noting that $(O[S_\Delta],O[\widehat{S}_{\widehat{\Delta}}])$ and $(O[R],O[\widehat{R}])$ are identical outside of $R \setminus S_\Delta$ and $\widehat{R} \setminus \widehat{S}_\Delta$, and applying \cref{lemma:puncture} (a standard oracle hybrid argument).
\end{proof}

Now, fix any $(3\secp/2+1)$-dimensional subspaces $R,\widehat{R}$ such that $\widehat{R}^\bot \subset R$, and consider the following procedure.

\begin{itemize}
    \item $\Gen'_{R,\widehat{R}}(1^\secp,n)$: Sample a uniformly random subspace $S \subset \bbF_2^{2\secp+1}$ conditioned on $\widehat{R}^\bot \subset S \subset R$, sample a uniformly random vector $\Delta \gets R \setminus S$, and sample uniformly random $x_i^R,z_i^{\widehat{R}} \gets (R,\widehat{R})$ for each $i \in [n]$. Set $x^R = (x_1^R,\dots,x_n^R), z^{\widehat{R}} = (z_1^{\widehat{R}},\dots,z_n^{\widehat{R}})$ and output $(S,\Delta,x^R,z^{\widehat{R}})$. 
\end{itemize}

Next, let $\co(R)$ be an arbitrary set of coset representatives of $R$, let $\co(\widehat{R})$ be an arbitrary set of coset representatives of $\widehat{R}$, and fix any

\[x^{\co(R)} = \left(x_1^{\co(R)},\dots,x_n^{\co(R)}\right), ~~~ z^{\co(\widehat{R})} = \left(z_1^{\co(\widehat{R})},\dots,z_n^{\co(\widehat{R})}\right),\] where each $x_i^{\co(R)} \in \co(R)$ and $z_i^{\co(\widehat{R})} \in \co(\widehat{R})$. Then the proof of the theorem follows by combining \cref{claim:increase-subspaces} with the following claim. Notice that the adversary $A$ in the following claim no longer requires access to the "expanded" oracle $\Ver'_{k,R,\widehat{R},\cdot,\cdot}(\cdot)$, since $A$ is allowed to depend on $\left(R,\widehat{R},x^{\co(R)},z^{\co(\widehat{R})}\right)$, which suffice to implement $\Ver'_{k,R,\widehat{R},\cdot,\cdot}(\cdot)$.

\begin{claim}
Fix any $R,\widehat{R},x^{\co(R)},z^{\co(\widehat{R})}$. Then for any $L,\theta,f$, and unitary $A$,\footnote{As noted in \cref{remark:unitary}, by introducing a sufficiently large workspace register $\regA$ initialized to $\ket{0}$, we can assume without loss of generality that the adversary $A$ is unitary. This additional workspace register $\regA$ is left implicit in the description of the claim and proof.} there exists an $\epsilon = \epsilon(\secp) \in [0,1]$ such that for any $\ket{\psi}$,

\begin{align*}
    \left\{y : \begin{array}{r} \left(S,\Delta,x^{R},z^{\widehat{R}}\right) \gets \Gen'_{R,\widehat{R}}(1^\secp,n) \\  x \coloneqq x^{R} + x^{\co(R)}, z \coloneqq z^{\widehat{R}} + z^{\co(\widehat{R})} \\ k \coloneqq (S,\Delta,x,z)\\ y \gets \widetilde{M}_{\theta,f,k,L} \circ A \circ \Enc_k(\ket{\psi})\end{array}\right\} \approx_{2^{-\Omega(\secp)}} (1-\epsilon)\left\{y: y \gets M_{\theta,f} \circ L(\ket{\psi})\right\} + \epsilon \left\{\bot\right\}.
\end{align*}

\end{claim}

\begin{proof}

Let $\cD$ be the distribution described by the LHS of the statement in the claim. Next, define $x^{\co(R)}_L,z^{\co(\widehat{R})}_L \coloneqq L(x^{\co(R)},z^{\co(\widehat{R})})$, and define the distribution $\cK_L$ as follows.

\[\cK_L \coloneqq \left\{(S,\Delta,x_L,z_L) : \begin{array}{r} \left(S,\Delta,x^{R},z^{\widehat{R}}\right) \gets \Gen'_{R,\widehat{R}}(1^\secp,n) \\ x^R_L,z^{\widehat{R}}_L \coloneqq L(x^R,z^{\widehat{R}})  \\ x_L \coloneqq x^R_L + x^{\co(R)}_L, z_L \coloneqq z^{\widehat{R}}_L + z^{\co(\widehat{R})}_L\end{array}\right\}.\]

Observe that the distribution over $k = (S,\Delta,x,z)$ as sampled by $\cD$ is equivalent to the distribution that results from sampling  $(S,\Delta,x_L,z_L) \gets \cK_L$ and setting $(x,z) = L^{-1}(x_L,z_L)$. Thus, we can write $\cD$ equivalently as

\[\cD = \left\{y : \begin{array}{r} \left(S,\Delta,x_L,z_L\right) \gets \cK_L \\ (x,z) \coloneqq L^{-1}(x_L,z_L) \\ k \coloneqq \left(S,\Delta,x,z\right) \\ y \gets \widetilde{M}_{\theta,f,k,L} \circ A \circ \Enc_k(\ket{\psi})\end{array}\right\}.\]

Moreover, the vectors $(x_L,z_L)$ obtained by sampling $(S,\Delta,x_L,z_L) \gets \cK_L$ are such that $x_L$ and $z_L$ are uniformly random over an affine subspaces, namely, \[x_L \gets R^{\oplus n} + x^{\co(R)}_L, \ \ \ \text{and } \ \ \ z_L \gets \widehat{R}^{\oplus n} + z^{\co(\widehat{R})}_L.\] 

This follows because $L$ is full rank, and the vectors $x^R,z^{\widehat{R}} = (x_1^R,\dots,x_n^R),(z_1^{\widehat{R}},\dots,z_n^{\widehat{R}})$ obtained by sampling \[\left(S,\Delta,x^{R},z^{\widehat{R}}\right) \gets \Gen'_{R,\widehat{R}}(1^\secp,n)\] are such that each $x_i^R$ is uniformly random over $R$ and each $z_i^{\widehat{R}}$ is uniformly random over $\widehat{R}$. The fact that $x_L$ and $z_L$ are uniform over affine subspaces will be used later in the proof when we apply the Pauli twirl over affine subspaces (\cref{lemma:Pauli-twirl}).

Next, we introduce some more notation.

\begin{itemize}
    \item For each $y \in \mathsf{range}(f)$, define \[V_y \coloneqq \bigcup_{m: f\left(m_\MS\right)=y}  \left(\bigotimes_{i \in \MS_0} \left(S + m_i \cdot \Delta\right) \bigotimes_{i \in \MS_1} \left(\widehat{S} + m_i \cdot \widehat{\Delta}\right)\right),\] and define
    \[V_\bot \coloneqq  \{0,1\}^{(2\secp + 1)n} \setminus \bigcup_{y \in \mathsf{range}(f)}V_y.\] For $y \in \mathsf{range}(f) \cup \{\bot\}$, define $\ket{V_y} \coloneqq \sum_{v \in V_y}\ket{v}$.
    \item Define the unitary $B \coloneqq A \circ \LinEval_L^\dagger$. Note that the "honest" $A$ operation just applies $\LinEval_L$, so in this case $B$ is the identity.

    \item For any pure state $\ket{\phi}$, define $\Mx[\ket{\phi}] \coloneqq \ketbra{\phi}{\phi}.$
\end{itemize}

Now, given any $(S,\Delta,x_L,z_L) \in \cK_L$, which defines $(x,z) = L^{-1}(x_L,z_L)$, and any $y \in \mathsf{range}(f) \cup \{\bot\}$, we can write the probability that $\cD$ outputs $y$ as  

\begin{align*}
    &\left\| \Pi[V_y]X^{x_{L,\MS_0},z_{L,\MS_1}}H^{\widetilde{\MS}_1}A\Enc_{(S,\Delta,x,z)}\ket{\psi}\right\|^2 \\
    &= \left\| \Pi[V_y]X^{x_{L,\MS_0},z_{L,\MS_1}}H^{\widetilde{\MS}_1}B\LinEval_L X^x Z^z E_{S,\Delta}^{\otimes n}\ket{\psi}\right\|^2 \\
    &= \left\| \Pi[V_y]H^{\widetilde{\MS}_1}X^{x_{L,\MS_0}}Z^{z_{L,\MS_1}}B X^{x_L} Z^{z_L} E_{S,\Delta}^{\otimes n}L\ket{\psi}\right\|^2 \\
    &= \left\| \Pi[V_y]H^{\widetilde{\MS}_1}X^{x_L}Z^{z_L}B X^{x_L} Z^{z_L} E_{S,\Delta}^{\otimes n}L\ket{\psi}\right\|^2,
\end{align*}

where in the last line, we have inserted Pauli $X$ operations on registers that are either measured in the Hadamard basis or not measured at all and Pauli $Z$ operations on registers that are either measured in the standard basis or not measured at all. Doing this has no effect on the outcome. Thus, we can write the distribution $\cD$ concisely as

\begin{align*}
\cD = \sum_{y \in \mathsf{range}(f) \cup \{\bot\}} \ketbra{y}{y} \frac{1}{|\cK_L|}\sum_{\left(S,\Delta,x_L,z_L\right) \in \cK_L}\bra{V_y}\Mx\left[H^{\widetilde{\MS}_1} Z^{z_{L}}X^{x_{L}} BX^{x_L}Z^{z_L}E_{S,\Delta}^{\otimes n}L\ket{\psi}\right]\ket{V_y}.
\end{align*}

To complete the proof, we will decompose $B$ as a sum of Paulis, and factor out terms that will cause $\cD$ to output $\bot$ (with high probability). Eventually, we'll be left with terms that do not affect the outcome of directly measuring $L\ket{\psi}$. To begin with, let 

\[\cP \coloneqq \left\{X^xZ^z : x = (x_1,\dots,x_n),z=(z_1,\dots,z_n) \in \{0,1\}^{(2\secp+1)n}\right\},\] and define the subsets

\[\cP_\bot \coloneqq \left\{X^xZ^z : \exists i \in \MS_0 \text{ s.t. } x_i \notin R \text{ or } \exists i \in \MS_1 \text{ s.t. } z_i \notin \widehat{R}\right\}, ~~~~~~ \cP_\top = \cP \setminus \cP_\bot.\]

Then we can write $B$ as \[B = \sum_{P \in \cP_\top} \alpha_P P  + \sum_{P \in \cP_\bot}\alpha_P P \coloneqq B_\top + B_\bot,\] for some coefficients $\alpha_P$.

Note that for any $y \in \mathsf{range}(f)$, $\left(S,\Delta,x_L,z_L\right) \in \cK_L$, and $P \in P_\bot$,
\[\bra{V_y}H^{\widetilde{\MS}_1} Z^{z_L}X^{x_L}P X^{x_L}Z^{z_L}E_{S,\Delta}^{\otimes n}L\ket{\psi} = 0,\] which follows by definition of $V_y$, since $S_\Delta \subset R$ and $\widehat{S}_{\widehat{\Delta}} \subset \widehat{R}$. Thus there exists an $\epsilon_\bot$ such that

\[\cD = \sum_{y \in \mathsf{range}(f) \cup \{\bot\}} \ketbra{y}{y} \frac{1}{|\cK_L|}\sum_{\left(S,\Delta,x_L,z_L\right)}\bra{V_y}\Mx\left[H^{\widetilde{\MS}_1} Z^{z_{L}}X^{x_{L}} B_\top X^{x_L}Z^{z_L}E_{S,\Delta}^{\otimes n}L\ket{\psi}\right]\ket{V_y} + \epsilon_\bot \ketbra{\bot}{\bot}.\]

Next, we define the following.
\begin{itemize}
    \item Let $C \simeq R / \widehat{R}^\bot$ be a subspace of coset representatives of $\widehat{R}^\bot$ in $R$.%
    \item Let $\widehat{C} \simeq \widehat{R} / R^\bot$ be a subspace of coset representatives of $R^\bot$ in $\widehat{R}$. %
    \item Define the set of Paulis
    \begin{align*}&\cP_{C,\widehat{C}} \coloneqq \left\{X^x Z^z : \begin{array}{c}\left\{x_i \in C, z_i = 0^{2\secp+1}\right\}_{i \in \MS_0}, \left\{x_i = 0^{2\secp+1}, z_i \in \widehat{C}\right\}_{i \in \MS_1}, \\ \left\{x_i = 0^{2\secp+1}, z_i = 0^{2\secp+1}\right\}_{i \in \MS_\bot}\end{array}\right\}.
    \end{align*}
\end{itemize}

Now, for any $(x,z) = (x_1,\dots,x_n,z_1,\dots,z_n)$ such that $P  = X^xZ^z \in \cP_\top$, define $(x',z') = (x_1',\dots,x_n',z_1',\dots,z_n')$ such that $X^{x'}Z^{z'} \in \cP_{C,\widehat{C}}$ as follows.

\begin{itemize}
    \item For $i \in \MS_0$, let $x_i' \in C$ be the representative of $x_i$'s coset (recall that $x_i \in R$ by definition of $\cP_\top$), and let $z_i' = 0^{2\secp+1}$.
    \item For $i \in \MS_1$, let $z_i' \in \widehat{C}$ be the representative of $z_i$'s coset (recall that $z_i \in \widehat{R}$ by definition of $\cP_\top$), and let $x_i' = 0^{2\secp+1}$.
    \item For $i \in \MS_\bot$, let $x'_i = 0^{2\secp+1}, z_i' = 0^{2\secp+1}$.
\end{itemize}

Then note that for all $y \in \mathsf{range}(f) \cup \{\bot\}$ and $\left(S,\Delta,x_L,z_L\right) \in \cK_L$, it holds that 

\[\bra{V_y}H^{\widetilde{\MS}_1} Z^{z_L}X^{x_L}X^{x}Z^z X^{x_L}Z^{z_L}E_{S,\Delta}^{\otimes n}L\ket{\psi} = \bra{V_y}H^{\widetilde{\MS}_1} Z^{z_L}X^{x_L}X^{x'}Z^{z'} X^{x_L}Z^{z_L}E_{S,\Delta}^{\otimes n}L\ket{\psi},\]

which follows by definition of $V_y$ and the fact that $S,\widehat{S}$ are always sampled so that $\widehat{R}^\bot \subset S$ and $R^\bot \subset \widehat{S}$. 

That is, we have identified for any $P \in \cP_\top$ a canonical $P' \in \cP_{C,\widehat{C}}$ for which $P'$ will have the same behavior as $P$ over all $\left(S,\Delta,x_L,z_L\right) \in \cK_L$. Thus, we can replace $B_\top$ in the expression for $\cD$ with \[\sum_{P \in \cP_{C,\widehat{C}}}\alpha_P P\] for some coefficients $\alpha_P$, and write $\cD$ as
\begin{align*}
&\sum_{y \in \mathsf{range}(f) \cup \{\bot\}} \ketbra{y}{y} \frac{1}{|\cK_L|}\sum_{\left(S,\Delta,x_L,z_L\right)}\bra{V_y}\Mx\left[H^{\widetilde{\MS}_1} Z^{z_{L}}X^{x_{L}} \left(\sum_{P \in \cP_{C,\widehat{C}}}\alpha_P P\right) X^{x_L}Z^{z_L}E_{S,\Delta}^{\otimes n}L\ket{\psi}\right]\ket{V_y} \\ & ~~~~~~~ + \epsilon_\bot \ketbra{\bot}{\bot} \\
= &\sum_{y \in \mathsf{range}(f) \cup \{\bot\}} \ketbra{y}{y} \sum_{P_0,P_1 \in \cP_{C,\widehat{C}}}\alpha_{P_0}\alpha_{P_1}^*\frac{1}{|\cK_L|}\\&\hspace{2cm}\left(\sum_{\left(S,\Delta,x_L,z_L\right)}\bra{V_y}H^{\widetilde{\MS}_1} Z^{z_{L}}X^{x_{L}} P_0 X^{x_L}Z^{z_L}\Mx\left[E_{S,\Delta}^{\otimes n}L\ket{\psi}\right] Z^{z_L}X^{x_L} P_1^\dagger X^{x_{L}}Z^{z_{L}}H^{\widetilde{\MS}_1}\ket{V_y}\right)\\ & ~~~~~~~ + \epsilon_\bot \ketbra{\bot}{\bot}
\end{align*}

Now, we are finally ready to apply the the Pauli twirl over affine subspaces (\cref{lemma:Pauli-twirl}). To do so, we make the following observations.

\begin{itemize}
    \item As noted above, $x_L \gets R^{\oplus n} + x^{\co(R)}_L$, and $z_L \gets \widehat{R}^{\oplus n} + z^{\co(\widehat{R})}_L$ are uniformly random over affine subspaces of $R^{\oplus n}$ and $\widehat{R}^{\oplus n}$ respectively. 
    \item Consider any $X^{x_0}Z^{z_0} \neq X^{x_1}Z^{z_1} \in \cP_{C,\widehat{C}}$. If $x_0 \neq x_1$, then there exists some index $i \in [n]$ such that $x_{0,i} \oplus x_{1,i} \notin \widehat{R}^\bot$ and thus, $x_0 \oplus x_1 \notin (\widehat{R}^{\oplus n})^\bot$. Otherwise, $z_0 \neq z_1$, and there exists some index $i \in [n]$ such that $z_{0,i} \oplus z_{1,i} \notin R^\bot$ and thus, $z_0 \oplus z_1 \notin (R^{\oplus n})^\bot$.
\end{itemize}

Then by \cref{lemma:Pauli-twirl}, all the cross-terms $P_0 \neq P_1$ are killed in the above expression for $\cD$, which we can now write as

\begin{align*}
&\sum_{y \in \mathsf{range}(f) \cup \{\bot\}} \ketbra{y}{y} \sum_{P \in \cP_{C,\widehat{C}}}\alpha_P\alpha_P^*\frac{1}{|\cK_L|}\left(\sum_{\left(S,\Delta,x_L,z_L\right)}\bra{V_y}\Mx\left[H^{\widetilde{\MS}_1} Z^{z_{L}}X^{x_{L}} PX^{x_L}Z^{z_L}E_{S,\Delta}^{\otimes n}L\ket{\psi}\right]\ket{V_y}\right)\\ & ~~~~~~~ + \epsilon_\bot \ketbra{\bot}{\bot},
\end{align*}

Finally, since $S,\Delta$ are chosen uniformly at random conditioned on $\widehat{R}^\bot \subset S_\Delta \subset R$, we have that for any fixed $P \in \cP_{C,\widehat{C}} \setminus \cI$, 

\begin{align*}
&\sum_{y \in \mathsf{range}(f) \cup \{\bot\}} \frac{1}{|\cK_L|}\sum_{\left(S,\Delta,x_L,z_L\right)}\bra{V_y}\Mx\left[H^{\widetilde{\MS}_1} Z^{z_{L}}X^{x_{L}} PX^{x_L}Z^{z_L}E_{S,\Delta}^{\otimes n}L\ket{\psi}\right]\ket{V_y}\\ &~~~~~ \leq \frac{|S_\Delta \setminus \widehat{R}_\bot|}{|R \setminus \widehat{R}_\bot|} + \frac{|\widehat{S}_{\widehat{\Delta}} \setminus R_\bot|}{|\widehat{R} \setminus R_\bot|} = 2 \cdot \frac{2^{\secp+1} - 2^\secp}{2^{3\secp/2+1} - 2^\secp} = 2^{-\Omega(\secp)}.
\end{align*}

Thus, $\cD$ is within $2^{-\Omega(\secp)}$ total variation distance of 

\begin{align*}
&(1-\epsilon_\bot)\sum_y \ketbra{y}{y}\frac{1}{|\cK_L|}\sum_{\left(S,\Delta,x_L,z_L\right)}\bra{V_y}\Mx\left[H^{\widetilde{\MS}_1} Z^{z_{L}}X^{x_{L}} \cI X^{x_L}Z^{z_L}E_{S,\Delta}^{\otimes n}L\ket{\psi}\right]\ket{V_y} + \epsilon_\bot \ketbra{\bot}{\bot}\\
&= (1-\epsilon_\bot)\sum_y \ketbra{y}{y}\frac{1}{|\cK_L|}\sum_{\left(S,\Delta,x_L,z_L\right)}\bra{V_y}\Mx\left[H^{\widetilde{\MS}_1}E_{S,\Delta}^{\otimes n}L\ket{\psi}\right]\ket{V_y} + \epsilon_\bot \ketbra{\bot}{\bot} \\
&= (1-\epsilon_\bot)\sum_y \ketbra{y}{y}\left(\sum_{m : f(m) = y}\bra{m}\right)\Mx\left[H^{\MS_1}L\ket{\psi}\right]\left(\sum_{m : f(m) = y}\ket{m}\right) + \epsilon_\bot \ketbra{\bot}{\bot} \\
&= (1-\epsilon_\bot)\{y : y \gets M_{\theta,f} \circ L(\ket{\psi})\} + \epsilon_\bot\{\bot\},
\end{align*} 

which completes the proof.

\end{proof}

\end{proof}

\begin{theorem}\label{thm:auth-privacy}
The QAS described in \proref{fig:authscheme} satisfies privacy (\cref{def:auth-privacy}).
\end{theorem}

\begin{proof}
First, recalling the definitions of $\Gen',\Ver'$ in the proof of \cref{thm:auth-sec}, and applying the oracle indistinguishability argued during the proof of \cref{claim:increase-subspaces}, it suffices to show that 

\[\bigg|\Pr_{k,R,\widehat{R} \gets \Gen'(1^\secp,n)}\left[1 \gets A^{\Ver'_{k,R,\widehat{R},\cdot,\cdot}(\cdot)} \circ \Enc_k(\ket{\psi_0})\right] - \Pr_{k,R,\widehat{R} \gets \Gen'(1^\secp,n)}\left[1 \gets A^{\Ver'_{k,R,\widehat{R},\cdot,\cdot}(\cdot)} \circ \Enc_k(\ket{\psi_1})\right]\bigg| = 0.\]

To see this, we'll show that we can give enough information to $A$ for it to implement $\Ver'_{k,R,\widehat{R},\cdot,\cdot}(\cdot)$ while preserving sufficient randomness to one-time pad the input state. 

Consider the following equivalent description of $\Gen'(1^\secp,n)$.\\

\noindent $\Gen'(1^\secp,n)$:

\begin{itemize}
    \item Sample a uniformly random $\secp$-dimensional subspace $S \subset \bbF_2^{2\secp+1}$, vector $\Delta \gets \bbF_2^{2\secp+1} \setminus S$, and uniformly random $(3\secp/2 + 1)$-dimensional subspaces $R,\widehat{R} \subset \bbF_2^{2\secp+1}$ conditioned on $S_\Delta \subset R$ and $\widehat{S}_{\widehat{\Delta}} \subset \widehat{R}$. 
    \item Let $H_\Delta$ be the $2\secp$-dimensional subspace perpendicular to $\Delta$ and $H_{\widehat{\Delta}}$ be the $2\secp$-dimensional subspace perpendicular to $\widehat{\Delta}$. For each $i \in [n]$, sample $x_{i,\Delta} \gets H_\Delta, b_i \gets \{0,1\}, z_{i,\widehat{\Delta}} \gets H_{\widehat{\Delta}}, c_i \gets \{0,1\}$, and define $x_i = x_{i,\Delta} + b_i \cdot \Delta$ and $z_i = z_{i,\widehat{\Delta}} + c_i \cdot \widehat{\Delta}$.
    \item Output $(S,\Delta,x,z),R,\widehat{R}$.
\end{itemize}

Now, fix any choice of $S,\Delta,R,\widehat{R},x_\Delta,z_{\widehat{\Delta}}$ sampled during the procedure $\Gen'(1^\secp,n)$, where $x_\Delta \coloneqq \left(x_{1,\Delta},\cdots,x_{n,\Delta}\right)$ and $z_{\widehat{\Delta}} \coloneqq \left(z_{1,\widehat{\Delta}},\cdots,z_{n,\widehat{\Delta}}\right)$, and consider the following procedure that completes the sampling of the key.\\

\noindent $\Gen'_{S,\Delta,R,\widehat{R},x_\Delta,z_{\widehat{\Delta}}}(1^\secp,n)$:

\begin{itemize}
    \item For each $i \in [n]$, sample $b_i,c_i \gets \{0,1\}$, and define $x_i = x_{i,\Delta} + b_i \cdot \Delta$ and $z_i = z_{i,\widehat{\Delta}} + c_i \cdot \widehat{\Delta}$.
    \item Output $(S,\Delta,x,z)$.
\end{itemize}

Since the oracle $\Ver'_{k,R,\widehat{R},\cdot,\cdot}(\cdot)$ can be implemented given just the fixed information $S,\Delta,R,\widehat{R},x_\Delta,z_{\widehat{\Delta}}$, it suffices to show that for any $S,\Delta,R,\widehat{R},x_\Delta,z_{\widehat{\Delta}}$ and any adversary $A$ (whose description may depend on this information), it holds that 

\[\bigg|\Pr_{k}\left[1 \gets A (\Enc_k(\ket{\psi_0}))\right] - \Pr_{k}\left[1 \gets A(\Enc_k(\ket{\psi_1}))\right]\bigg| = 0,\]

where the probability is over $k \gets \Gen'_{S,\Delta,R,\widehat{R},x_\Delta,z_{\widehat{\Delta}}}(1^\secp,n)$. Since

\[\Enc_k = X^x Z^z E^{\otimes n}_{S,\Delta} = X^{x_\Delta}Z^{z_{\widehat{\Delta}}}X^{b_1 \cdot \Delta \dots b_n \cdot \Delta}Z^{c_1 \cdots \widehat{\Delta} \dots c_n \cdot \widehat{\Delta}} E^{\otimes n}_{S,\Delta} = X^{x_\Delta}Z^{z_{\widehat{\Delta}}} E^{\otimes n}_{S,\Delta} X^{b_1 \dots b_n}Z^{c_1 \dots c_n},\] this follows from the quantum one-time pad \cite{892142}. That is, we use the fact that 
\[\sum_{b_1,\dots,b_n,c_1,\dots,c_n}\Mx\left[X^{b_1,\dots,b_n}Z^{c_1,\dots,c_n}\ket{\psi_0}\right] = \sum_{b_1,\dots,b_n,c_1,\dots,c_n}\Mx\left[X^{b_1,\dots,b_n}Z^{c_1,\dots,c_n}\ket{\psi_1}\right].\]

\end{proof}

\section{Linear + Measurement Quantum Programs}\label{subsec:circuit-rep}

In this section, we show that any quantum program with classical input and output (\cref{def:quantum-imp}) can be implemented using a "linear + measurement" ($\LM$) quantum program.

In slightly more detail, we make use of magic states in order to write any quantum circuit as an alternating sequence of linear operations $L_i$ (by which we mean a sequence of CNOT gates) and partial ZX measurements $M_{\theta_i,f_i}$, where the description of each $f_i$ may depend on the classical input $x$ as well as previous measurement results. We encourage the reader to review our notation for partial ZX measurements $M_{\theta_i,f_i}$ described at the beginning of \cref{subsec:auth-def}. We remark here that these measurements are "partial" in two aspects: (i) they may only operate on a subset of the qubits, and (ii) each measurement outcome may be associated with multiple basis vectors, meaning that the input qubits are not necessarily fully collapsed. We also remark that for the purpose of this paper, we restrict attention to circuits with classical inputs and outputs, but note that one could consider circuits with quantum inputs and outputs as well. 

We first formally define $\LM$ quantum programs, and accompany this with a diagram in \cref{fig:LM}.

\begin{definition}[$\LM$ quantum program]\label{def:LM-circuit}
    An $\LM$ quantum program with classical input and output is described by:
    \begin{itemize}
        \item A quantum state $\ket{\psi}$ on $n$ qubits.
        \item Linear operations $L_1,\dots,L_{t+1}$, where each $L_i$ is a sequence of CNOT gates.
        \item Partial ZX measurements $M_{\theta_1,f_1^{(\cdot)}},M_{\theta_2,f_2^{(\cdot)}},\dots,M_{\theta_t,f_t^{(\cdot)}}$,$M_{\theta_{t+1},g^{(\cdot)}}$ defined by sets of bases $\{\theta_i\}_{i \in [t+1]}$ and classical functions $\{f_i^{(\cdot)}\}_{i \in [t]},g^{(\cdot)}$, which will be parameterized by the input $x$ as well as previous measurement results. In line with the notation introduced in \cref{subsec:auth-def}, for each $i \in [t+1]$, we define $\MS_i \subseteq [n]$ be the set of wires such that $\theta_i \neq \bot$. That is, $\MS_i$ is the set of wires on which the $i$'th partial measurement operates.
    \end{itemize}
    
    Now, we will find it useful to introduce further notation drawing attention to which wires are simply measured in either the standard or Hamadard basis by the $i$'th partial ZX measurement, and which are not fully collapsed. In particular, we define disjoint sets $V_1,\dots,V_{t+1}$ and sets $W_1,\dots,W_t$ with the following properties.
    \begin{itemize}
        \item $\MS_1 = (V_1,W_1), \MS_2 = (V_1,V_2,W_2), \dots, \MS_t = (V_1,\dots,V_t,W_t), \MS_{t+1} = (V_1,\dots,V_{t+1}) = [n]$.
        \item The $i$'th measurement takes previously collapsed wires $V_1,\dots,V_{i-1}$ as input, "fully" collapses wires $V_i$, and "partially" collapses wires $W_i$. This will be made precise by the evaluation procedure defined below, where the $v_i$ are inputs from the $V_i$ wires and $w_i$ are inputs from the $W_i$ wires.
        \item Each $L_i$ does not operate on $\{V_j\}_{j < i}$. That is, fully collapsed registers are no longer computed on.
    \end{itemize}

    \noindent Finally, given an input $x \in \{0,1\}^m$, let $\LMEval\left(x,\ket{\psi},\{L_i\}_{i \in [t+1]},\{\theta_i\}_{i \in [t+1]},\{f_i\}_{i \in [t]},g\right) \to y$ be the formal evaluation procedure, defined as follows:

    \begin{itemize}
        \item Initialize an $n$-qubit register $\regM$ with $\ket{\psi}$.
        \item Compute $((v_1,r_1),\regM) \gets M_{\theta_1,f_1^x} \circ L_1(\regM)$, where $f_1^x(v_1,w_1) = (v_1,r_1)$.
        \item Compute $((v_2,r_2),\regM) \gets M_{\theta_2,f_2^{x,r_1}} \circ L_2 (\regM)$, where $f_2^{x,r_1}(v_1,v_2,w_2) = (v_2,r_2)$.
        \item \dots
        \item Compute $((v_t,r_t),\regM) \gets M_{\theta_t,f_t^{x,r_1,\dots,r_{t-1}}} \circ L_t(\regM)$, where $f^{x,r_1,\dots,r_{t-1}}_t(v_1,\dots,v_t,w_t) = (v_t,r_t)$.
        \item Compute output $y \gets M_{\theta_{t+1},g^{x,r_1,\dots,r_t}} \circ L_{t+1}(\regM)$, where $g^{x,r_1,\dots,r_t}(v_1,\dots,v_{t+1}) = y$.
    \end{itemize}

\end{definition}

\begin{figure}
\begin{quantikz}[transparent]
     \lstick[8]{$\ket{\psi}$} & \gate[8]{L_1} & \qw & \gate[6]{L_2} & \qw & \gate[4]{L_3}& \gate[8,nwires={5,6,7,8}]{M_{\theta_3,g^{x,r_1,r_2}}}\gateinput[4]{$V_3$} & \cw \rstick[4]{$y$}\\
    & & \qw & & \qw & & & \cw\\
    & & \qw & & \gate[6,nwires={5,6}]{M_{\theta_2,f_2^{x,r_1}}}\gateinput[wires=2]{$W_2$} & & & \cw \\
    & & \qw & & & & & \cw \\
     &  & \gate[4]{M_{\theta_1,f_1^x}}\gateinput[wires=2]{$W_1$}  & & \gateinput[wires=2]{$V_2$} & \cw \midstick[2]{$v_2$}& \gateinput[wires=2]{$V_2$}\cw & \\
    &  & & & & \cw & \cw & \\
    &  & \gateinput[wires=2]{$V_1$} & \cw \midstick[2]{$v_1$}  & \gateinput[wires=2]{$V_1$} \cw &  \cw \midstick[2]{$v_1$} & \gateinput[wires=2]{$V_1$}\cw & \\
     & & \vcw{1} & \cw & \cw \vcw{1} & \cw & \cw & \\
    & & & & & & & \\
    & & \text{$r_1$}& & \text{$r_2$} & & &
\end{quantikz}
\caption{Diagram of an $\LM$ quantum program. We use some non-standard quantum circuit notation, so we provide some explanation. Each partial ZX measurement $M_{\theta_1,f_1^{(\cdot)}}, M_{\theta_2,f_2^{(\cdot)}}, M_{\theta_3,g^{(\cdot)}}$ is applied to the wires coming from the left of the corresponding box, some of which may be classical. Some wires (namely, $V_1$, $V_2$, and $V_3$) are fully collapsed by the measurement, producing classical output wires coming from the right. Other wires (namely, $W_1$ and $W_2$) are only partially collapsed, so their corresponding output wires are still quantum. Additional classical outputs (namely, $r_1$ and $r_2$) are produced by these measurements, which are denoted by classical wires coming out of the bottom. Note that the description of later measurements depend on $r_1,r_2$. Finally, we remark that one could instead introduce explicit ancillary wires for the input $x$ and intermediate measurement results $r_1,r_2$, but writing the circuit in the manner above is visually suggestive of the structure of our eventual obfuscation scheme.}
\label{fig:LM}
\end{figure}
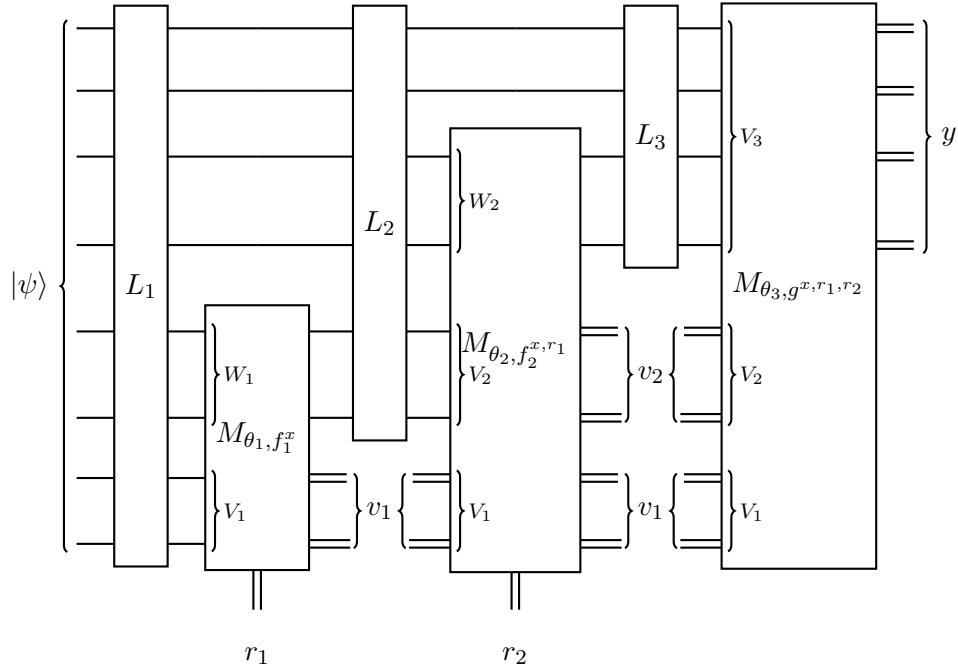

Observe that any alternating sequence of linear operations and partial ZX measurements may be written in the form introduced in the above definition, by simply defining each of the sets $V_i$ to be empty, and writing each function as $f_i^{x,r_1,\dots,r_{i-1}}(w_i) \to r_i$ (that is, we can always choose not to treat any of the wires as fully collapsed in the above formalism). So, why did we bother explicitly defining the $V_i$ and $W_i$ sets? The reason is that we will actually be interested in a ``subclass'' of $\LM$ quantum programs whose partially collapsed wires (the $W_i$ wires) have a particularly simple structure. The diagram shown in \cref{fig:LM} is indeed an example of such an $\LM$ quantum program.

\begin{definition}[$\LM$ quantum program with standard-basis-collapsible $W$ wires]\label{def:LM-properties}
    An $\LM$ quantum program has \emph{standard-basis-collapsible $W$ wires} if:
    \begin{itemize}
        \item $W_1,\dots,W_n$ consist of only standard basis indices, that is, $\theta_{i,j} = 0$ for $i \in [t]$ and $j \in W_i$.
        \item For $i \in [t]$, $W_i$ is disjoint from $\MS_1 \cup \dots \cup \MS_{i-1}$, and the operations $L_1,\dots,L_{i-1}$ are either classically controlled on or do not operate on $W_i$. In particular, for each $i \in [t]$, the entire operation of the $\LM$ program up to and including the $i$'th measurement is diagonal in the standard basis on the wires $W_i$.
    \end{itemize}
\end{definition}

This \emph{standard-basis-collapsible $W$ wires} property ensures that if one were to measure ("collapse") the $W_1,\dots,\allowbreak W_n$ wires in the standard basis before executing the program, the $W_i$ wires would remain completely unaffected throughout the execution of the program up to and \emph{including} the $i$'th measurement (though they could be affected after the $i$'th measurement). Note that this is not a correctness property, indeed, collapsing the $W_1,\dots,W_n$ wires at the beginning of the computation would likely completely change the desired functionality. However, it turns out that this property will be crucial for arguing the \emph{security} of our obfuscation scheme in the following sections (in particular, refer to the "Collapsing the $\sF$ oracles" discussion in the proof intuition section, \cref{subsec:proof-overview}).

\begin{theorem}\label{thm:LM-compiler}
Any quantum program $(\ket{\psi},C)$ (\cref{def:quantum-imp}) can be compiled into an equivalent $\LM$ quantum program $(\ket{\psi'},\{L_i\}_{i \in [t+1]},\{\theta_i\}_{i \in [t+1]},\{f_i\}_{i \in [t]},g)$ with standard-basis-collapsible $W$ wires, where $\{L_i\}_{i \in [t+1]},\allowbreak\{\theta_i\}_{i \in [t+1]},\allowbreak\{f_i\}_{i \in [t]},\allowbreak g$ only depend on the description of $C$ (and not $\ket{\psi}$). Moreover, the compiler runs in polynomial time in the size of its input $(\ket{\psi},C)$.
\end{theorem}

\begin{proof}

We will use a circuit representation very similar to that described in \cite{C:BroGutSte13}, except for a key difference in how we implement the $T$ gate, inspired by the encrypted CNOT operation introduced in \cite{FOCS:Mahadev18b}. We write the quantum circuit $C$ using the $\{\CNOT,H,T\}$ universal gate set, where $T$ is the gate that applies a phase of $e^{i\pi/4}$. Given magic states, we'll show how to implement $H$ and $T$ gates using only CNOT gates and Pauli ($X$ and $Z$) gates controlled on the results of partial ZX measurements. Then, we will observe that the Pauli gates can be subsumed into the description of the measurements, leaving only layers of CNOT gates and partial ZX measurements. 

First, we'll describe our implementations of the $H$ and $T$ gates and prove that they are correct. Then, we'll complete the proof with an inductive argument, showing how to build an $\LM$ quantum program one gate at a time.

\begin{figure}
    \centering
    \begin{quantikz}
     \lstick{$\ket{\psi}$} & \ctrl{1} & \gate{H}  & \meter{} \vcw{2}\\
     \lstick[2]{$\ket{\phi_H}$}& \targ{} & \meter{} \vcw{1} &  \\
     & \qw & \gate{Z} & \gate{X} & \qw \rstick{$H\ket{\psi}$}
    \end{quantikz}
   \caption{Implementation of the $H$ gate with the $H$-magic state $\ket{\phi_H} \propto \ket{00} + \ket{01} + \ket{10} - \ket{11}$.} \label{fig:H}
\end{figure}

\paragraph{Implementing the $H$ gate.} Following \cite{C:BroGutSte13}, we use a two-qubit magic state $$\ket{\phi_H} \propto \ket{00} + \ket{01} + \ket{10} - \ket{11}$$ to implement the Hadamard gate, via the circuit in Figure~\ref{fig:H}. For completeness, we show that the circuit indeed implements the Hadamard gate.

\begin{claim}\label{claim:H-gate}
    The circuit in Figure~\ref{fig:H} implements the Hadamard gate.
\end{claim}
\begin{proof}
    Write $\ket{\psi} = \alpha\ket{0} + \beta\ket{1}$. After the CNOT gate, the joint state of all three qubits can be written as 
    \begin{align*} \alpha \ket{000} + &\alpha\ket{001} + \alpha\ket{010} - \alpha \ket{011} + \beta\ket{100} - \beta\ket{101} + \beta\ket{110} + \beta \ket{111} \\
    = & \Big((\alpha+\beta) \ket{+}\ket{00} + (\alpha-\beta)\ket{+}\ket{01}\Big) + \Big((\alpha+\beta)\ket{+}\ket{10} - (\alpha-\beta) \ket{+}\ket{11}\Big) \\
    + & \Big((\alpha-\beta) \ket{-}\ket{00} + (\alpha + \beta)\ket{-}\ket{01}\Big) + \Big(\alpha-\beta)\ket{-}\ket{10} - (\alpha+\beta)\ket{-}\ket{11}\Big)
    \end{align*}
    After the Hadamard basis measurement on the first wire resulting in a bit $x$ and the standard basis measurement on the second wire resulting in a bit $z$, the resulting state on the third wire is 
    \begin{align*}
       (\alpha+\beta) \ket{0} + (\alpha-\beta)\ket{1} & = H\ket{\psi} & \hspace{.3in} \mbox{if $x=0$ and $z=0$} \\
       (\alpha+\beta)\ket{0} - (\alpha-\beta) \ket{1} & = ZH\ket{\psi} & \hspace{.3in} \mbox{if $x=0$ and $z=1$}  \\ 
       (\alpha-\beta) \ket{0} + (\alpha + \beta)\ket{1} & = XH\ket{\psi} & \hspace{.3in} \mbox{if $x=1$ and $z=0$} \\
       (\alpha-\beta)\ket{0} - (\alpha+\beta)\ket{1}  & = ZXH\ket{\psi} & \hspace{.3in} \mbox{if $x=1$ and $z=1$}
    \end{align*}
    Applying the $Z$ and $X$ corrections now gives the state $H\ket{\psi}$.
\end{proof}

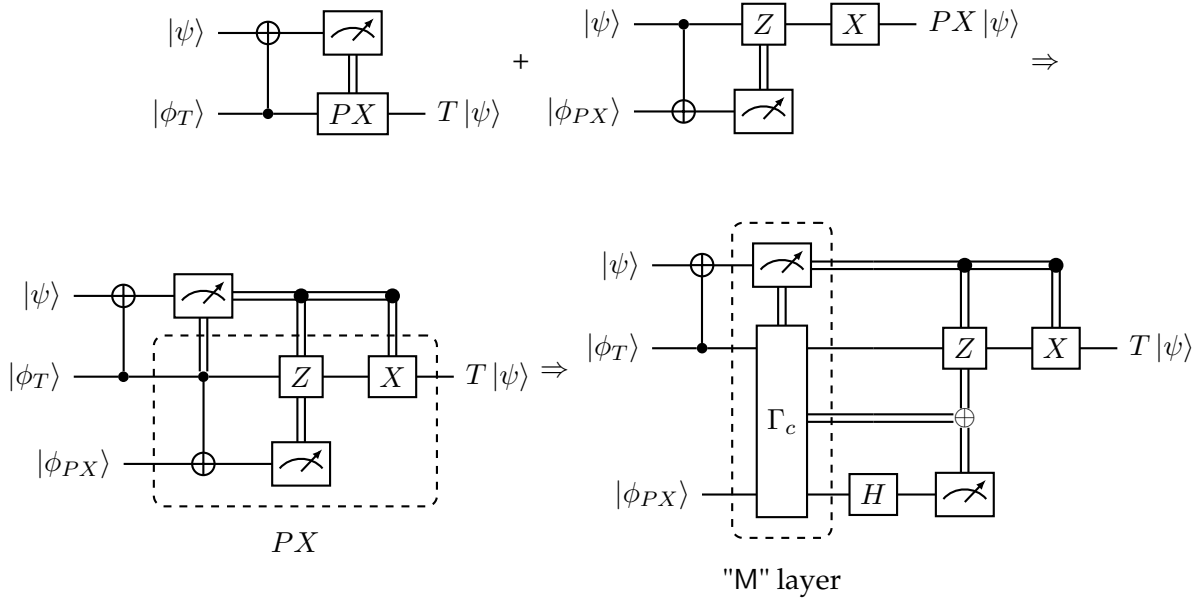
\begin{figure}
    \begin{minipage}{\textwidth}
    \centering 
    \begin{quantikz}
        \lstick{$\ket{\psi}$} & \targ{} & \meter{} \vcw{1}   \\
        \lstick{$\ket{\phi_T}$} & \ctrl{-1} & \gate{PX} &  \qw \rstick{$T\ket{\psi}$}
    \end{quantikz}
    +
    \begin{quantikz}
        \lstick{$\ket{\psi}$} & \ctrl{1} &  \gate{Z} & \gate{X} & \qw \rstick{$PX\ket{\psi}$} \\
        \lstick{$\ket{\phi_{PX}}$} & \targ{} & \meter{} \vcw{-1} & 
    \end{quantikz}
    $\Rightarrow$ \\ 
    \vspace{1cm}
    \begin{quantikz}
        \lstick{$\ket{\psi}$} & \targ{} & \meter{} \vcw{1} & \cwbend{1} & \cwbend{1}   \\
        \lstick{$\ket{\phi_T}$} & \ctrl{-1} & \ctrl{1}\gategroup[wires=2,steps=3,style={dashed, rounded corners},label style={yshift=-3 cm}]{$PX$} & \gate{Z}\vcw{1} & \gate{X}& \qw \rstick{$T\ket{\psi}$} \\
        & \lstick{$\ket{\phi_{PX}}$} & \targ{} & \meter{} &
    \end{quantikz}
    $\Rightarrow$
    \begin{quantikz}
        \lstick{$\ket{\psi}$} & \targ{} & \meter{}\gategroup[wires=4,style={dashed, rounded corners},label style={yshift=-5 cm}]{"$\sM$" layer} \vcw{1} & \cw & \cwbend{1} & \cwbend{1}   \\
        \lstick{$\ket{\phi_T}$} & \ctrl{-1} & \gate[wires=3,nwires=2][3]{\Gamma_c} & \qw & \gate{Z}\vcw{1} & \gate{X}& \qw \rstick{$T\ket{\psi}$} \\
        & & & \cw & \cw \text{$\oplus$} & \\
        & \lstick{$\ket{\phi_{PX}}$} & & \gate{H} & \meter{}\vcw{-1} 
    \end{quantikz}
    \caption{Implementation of the $T$ gate. First, we combine an implementation of the $T$ gate using the $T$-magic state $\ket{\phi_T} \propto \ket{0} + e^{i\pi/4}\ket{1}$ (upper left) with an implementation of the $PX$ gate using the $PX$-magic state $\ket{\phi_{PX}} \propto i\ket{0} + \ket{1}$ (upper right) to obtain the circuit on the bottom left. This circuit includes a classically controlled CNOT gate, which is not supported by $\LM$ quantum programs. We replace the classically controlled CNOT with a classically controlled \emph{projective measurement} to arrive at the circuit on the bottom right. Here, $\Gamma_c$ represents a measurement controlled on the bit $c$ from the first wire to be applied to the second and third wires. These wires are only partially collapsed by this measurement, so they remain quantum wires. However, $\Gamma_c$ also produces a measurement result $r$, which is carried on the classical wire coming from the right. The final $Z$ gate is controlled on the bit $c \cdot (r \oplus h)$, where $h$ is the result of measuring the third wire in the Hadamard basis. The dashed box will eventually become a measurement layer in our implementation of an $\LM$ circuit (though we remark that the input for this measurement will also include wires from previous $H$-gate and $T$-gate circuits). }
    \label{fig:T}
    \end{minipage}
\end{figure}

\paragraph{Implementing the $T$ gate.} We will use two magic states 
$$ \ket{\phi_T} \propto \ket{0} + e^{i\pi /4} \ket{1} \hspace{.2in} \mbox{and} \hspace{.2in} \ket{\phi_{PX}} \propto i\ket{0} + \ket{1} 
$$ 
and the circuit on the bottom right of Figure~\ref{fig:T}. First, we clarify notation in the figure. $\Gamma_c$ is a projective measurement controlled on the bit $c$ from the first wire, and is defined as follows.

    \begin{itemize}
        \item $\Gamma_0 = \{\ketbra{00}{00} + \ketbra{10}{10}, \ketbra{01}{01} + \ketbra{11}{11}\}$. That is, it measures its second input in the standard basis.
        \item $\Gamma_1 = \{\ketbra{00}{00} + \ketbra{11}{11}, \ketbra{01}{01} + \ketbra{10}{10}\}$. That is, it measures the XOR of its two inputs.
    \end{itemize}

    The measurement $\Gamma_c$ is applied to the second and third wires, which remain quantum wires, and produces a classical bit $r$ indicating which of the two measurement results was observed. In this figure, this bit $r$ is carried on the classical wire coming from the right of $\Gamma_c$. In an abuse of notation, we will also use $\Gamma_c$ as a function to define measurement outcomes:
    \begin{align*}\Gamma_0(00) & = \Gamma_0(10) = 0, \ \ \ \Gamma_0(01) = \Gamma_0(11) = 1, \\ 
    \Gamma_1(00) & = \Gamma_1(11) = 0, \ \ \ \Gamma_1(01) = \Gamma_1(10) = 1
    \end{align*}

    The control logic for the $Z$ gate is $c \cdot (r \oplus h)$, where $c$ is the result of measuring the first wire, $r$ is the result of measuring $\Gamma_c$, and $h$ is the result of measuring the third wire in the Hadamard basis. We will now confirm this representation of the $T$ gate works as expected. 
    
\begin{claim}\label{claim:T-gate}
    The bottom right circuit in Figure~\ref{fig:T} implements the $T$ gate.
\end{claim}

\begin{proof}
    Write $\ket{\psi} = \alpha\ket{0}+\beta\ket{1}$. Applying the first CNOT yields

    \begin{align*}&\alpha\ket{00} + e^{i\pi/4}\beta\ket{01} + \beta\ket{10} + e^{i\pi/4}\alpha\ket{11}\\ &= \ket{0}(\alpha\ket{0} + e^{i\pi/4}\beta\ket{1}) + \ket{1}(\beta\ket{0} + e^{i\pi/4}\alpha\ket{1}).
    \end{align*}

    If the result of measuring the first wire is $c = 0$, the state on the second wire is already \[\alpha\ket{0} + e^{i\pi/4}\beta\ket{1} = T\ket{\psi}.\] In this case, we measure $\Gamma_0$, which only collapses the third wire, and neither the $Z$ nor $X$ corrections is applied to the second wire, which remains in the state $T\ket{\psi}$. 
    
    If the result of measuring the first wire is $c=1$, then the second wire is in the state \[\beta\ket{0} + e^{i\pi/4}\alpha\ket{1}.\]
    
    In this case, we measure $\Gamma_1$ on \[(\beta\ket{0} + e^{i\pi/4}\alpha\ket{1})(i\ket{0}+\ket{1}).\] 
    
    If the result is $r=0$, the state has collapsed to

    \begin{align*}
        &i\beta\ket{00} + e^{i\pi/4}\alpha\ket{11} \\
        &=e^{i\pi/4} \beta\ket{00} + \alpha\ket{11}\\
        &=\left(e^{i\pi/4}\beta\ket{0} + \alpha\ket{1}\right)\ket{+} + \left(e^{i\pi/4}\beta\ket{0} - \alpha\ket{1}\right)\ket{-}\\
        &=XT\ket{\psi}\ket{+} + ZXT\ket{\psi}\ket{-},
    \end{align*}

    so applying the $Z$ correction controlled on $c \cdot(r \oplus h) = h$ followed by the $X$ correction results in  $T\ket{\psi}$. If the result is $r=1$, the state has collapsed to

    \begin{align*}
        &\beta\ket{01} + ie^{i\pi/4}\alpha\ket{10} \\
        &= e^{i\pi/4}\beta\ket{01} -\alpha\ket{10} \\
        &= \left(e^{i\pi/4}\ket{0} - \alpha\ket{1}\right)\ket{+} - \left(e^{i\pi/4}\beta\ket{0} + \alpha\ket{1}\right)\ket{-}\\
        &= ZXT\ket{\psi}\ket{+} - ZT\ket{\psi}\ket{-},
    \end{align*}

    so applying the $Z$ correction controlled on $c \cdot(r \oplus h) = 1\oplus h$ followed by the $X$ correction results in $T\ket{\psi}$. 
\end{proof}

\paragraph{Inductive argument.} In order to carry out an inductive argument, we will first generalize the notion of an $\LM$ quantum program to support quantum output. We will actually allow the output to be correct up to some Pauli errors that can be computed by measuring some ancillary registers in the standard or Hadamard basis and applying a classical function to the measurement results.

First, we fix some notation. Throughout the proof, we'll keep track of disjoint sets of wires $V_1,\dots,V_t,V^*_{t+1},O$, where $V_1 \cup \dots \cup V_t \cup V^*_{t+1} \cup O = [n]$. We will also keep track of strings $v,\widehat{x},\widehat{z} \in \{0,1,\bot\}^n$, where $v$ denotes a subset of measurement results (from wires $V_1 \cup \dots \cup V_t \cup V^*_{t+1}$), and $\widehat{x},\widehat{z}$ denote Pauli corrections to be applied to the wires in $O$. For any set $V$, we let $v^{(V)}$ (resp.\ $\widehat{x}^{(V)},\widehat{z}^{(V)}$) be the string restricted to indices in the set $V$. In particular, for an index $i \in [n]$, $v^{(i)}$ is just the $i$'th entry of $v$. Finally, for each $i \in [t]$, we define $v_i \coloneqq v^{(V_i)}$, and we define $v^*_{t+1} \coloneqq v^{(V^*_{t+1})}$.

\begin{definition}[$\LM$ quantum program with Pauli-encoded quantum output]
    An $\LM$ quantum program with Pauli-encoded quantum output is defined like a standard $\LM$ quantum program (\cref{def:LM-circuit}), except for the following differences.
    \begin{itemize}
        \item There is no final measurement $M_{\theta_{t+1},g^{(\cdot)}}$, and thus $\theta_{t+1}$ and $g^{(\cdot)}$ are undefined.
        \item The set $[n] \setminus (V_1 \cup \dots \cup V_t)$ consists of disjoint sets  $V^*_{t+1}$ and $O$, where $O$ contains the Pauli-encoded output. We will refer to $O$ as the "active" set of wires.
        \item There is a string $\theta^*_{t+1} \in \{0,1,\bot\}^n$ that is 0 or 1 on $V^*_{t+1}$ (determining whether these registers are measured in the standard or Hadamard basis) and $\bot$ everywhere else.
        \item There is a classical function $h(x,v_1,\dots,v_t,v^*_{t+1},r_1,\dots,r_t) \to \{0,1\}^{2|O|}$ that operates on the program input $x$ and previous measurement results, and outputs Pauli corrections $\widehat{x}^{(O)},\widehat{z}^{(O)} \in \{0,1\}^{|O|}$ to be applied to the $O$ registers.
    \end{itemize}
    Note that the program is defined by a state $\ket{\psi}$ along with $\{L_i\}_{i \in [t+1]},\{\theta_i\}_{i \in [t]},\{f_i^{(\cdot)}\}_{i \in [t]},\theta^*_{t+1},h$.
\end{definition}

First, the following claim will confirm that it suffices to compile the quantum program $(\ket{\psi},C)$ into an $\LM$ quantum program with Pauli-encoded quantum output.

\begin{claim}
    Consider any $\LM$ quantum program with Pauli-encoded quantum output that computes a classical output functionality. That is, the (Pauli encoding of the) output we are interested in is determined by measuring the $O$ registers in the standard basis. Then, this program can be written as a standard $\LM$ quantum program (\cref{def:LM-circuit}).
\end{claim}

\begin{proof}
    It suffices to define the final measurement $M_{\theta_{t+1},g^{(\cdot)}}$. Set $\theta_{t+1} \in \{0,1\}^n$ to be equal to $\theta_t$ on the sets $V_1,\dots,V_t$, equal to $\theta^*_{t+1}$ on the set $V^*_{t+1}$ and equal to 0 on the set $O$. Let $g^{x,r_1,\dots,r_t}(v_1,\dots,v_{t+1})$ be defined as follows. Parse $v_{t+1}$ as $(v^*_{t+1},y')$, compute $h(x,v_1,\dots,v_t,v^*_{t+1},r_1,\dots,r_t) = (\widehat{x}^{(O)},\widehat{z}^{(O)})$, and output $y = y' \oplus \widehat{x}^{(O)}$.
\end{proof}

Now, we show how to compile any quantum program $(\ket{\psi},C)$ into an $\LM$ quantum program with Pauli-encoded quantum output. We proceed by induction over the number of gates $\ell$ in $C$.

\paragraph{Base case.} Suppose that $C$ contains 0 gates. That is, there is no state $\ket{\psi}$ and the functionality is just the identity applied to input $x$. In this case, $n = |x|$, $t=0$, $L_1$ is empty, $\theta^*_1 = \bot^n$, and the $\LM$ quantum program is defined by $(\ket{0^n},h)$, where $h(x) = (x,0^n)$.

\paragraph{Inductive step.} Consider a quantum program $(\ket{\psi},C)$ with $\ell+1$ gates, and begin by writing $C$ as $(C_\ell,G)$, where $C_\ell$ contains the first $\ell$ gates, and $G \in \{\mathsf{CNOT},H,T\}$ is the final gate. By the inductive hypothesis, we know that $(\ket{\psi},C_\ell)$ can be written as an $\LM$ quantum program with Pauli-encoded quantum output: $\ket{\psi'},\{L_i\}_{i \in [t+1]},\{\theta_i\}_{i \in [t]},\{f_i^{(\cdot)}\}_{i \in [t]},\theta^*_{t+1},h$, where $t$ is the number of $T$ gates in $C_\ell$. Now, we consider three cases corresponding to the gate $G$, which by definition will be applied to one (or two) wire(s) in the "active" set $O$. In each case, we describe how to update the description of the $\LM$ quantum program for $(\ket{\psi},C_\ell)$ so that it now has the same functionality as the full quantum program $(\ket{\psi},C)$. The fact that these updates implement the desired functionality follow from \cref{claim:H-gate} and \cref{claim:T-gate} above. 

\begin{itemize}
    \item $\CNOT$ from wire $i$ to $j$: Append the description of this gate to the end of $L_{t+1}$ and append the operation  $(\widehat{x}^{(i)},\widehat{z}^{(i)}), (\widehat{x}^{(j)},\widehat{z}^{(j)}) \to (\widehat{x}^{(i)},\widehat{z}^{(i)} \oplus \widehat{z}^{(j)}), (\widehat{x}^{(i)} \oplus \widehat{x}^{(j)}, \widehat{z}^{(j)})$ to the end of $h$.
    \item $H$ on wire $i$: Refer to \cref{fig:H}. 
    \begin{itemize}
        \item Introduce two new wires $(n+1,n+2)$ and append $\ket{\phi_H}$ to $\ket{\psi'}$.
        \item Append the description of a $\CNOT$ gate from wire $i$ to $n+1$ to the end of $L_{t+1}$.
        \item Remove wire $i$ from and add wire $n+2$ to $O$. 
        \item Add wires $i$ and $n+1$ to $V^*_{t+1}$. \item For $\tau \in [t]$, define $\theta_{\tau,n+1} = \theta_{\tau,n+2} = \bot$. Define $\theta^*_{t+1,i} = 1$, $\theta^*_{t+1,n+1} = 0$, and $\theta^*_{t+1,n+2} = \bot$.
        \item Update $h$ to $h'$ as follows. The function $h'$ will now take two additional input bits $v^{(i)}, v^{(n+1)}$ as part of $v^*_{t+1}$ and its output will now include $(\widehat{x}^{(n+2)},\widehat{z}^{(n+2)})$ rather than $(\widehat{x}^{(i)},\widehat{z}^{(i)})$, computed as follows. Let $\widehat{x}^{(O)},\widehat{z}^{(O)} = h(x,v_1,\dots,v_t,v^*_{t+1},r_1,\dots,r_t)$ be the output of the original $h$, which includes $(\widehat{x}^{(i)},\widehat{z}^{(i)})$. Then the output of $h'$ includes $\widehat{x}^{(n+2)} \coloneqq v^{(i)} \oplus \widehat{z}^{(i)}$ and $\widehat{z}^{(n+2)} \coloneqq v^{(n+1)} \oplus \widehat{x}^{(i)}$.
    \end{itemize}

    \item $T$ on wire $i$: Refer to \cref{fig:T}. 
    \begin{itemize}
        \item Introduce two new wires $(n+1,n+2)$ and append $\ket{\phi_T}\ket{\psi_{PX}}$ to $\ket{\psi'}$.
        \item Append the description of a $\CNOT$ gate from wire $n+1$ to $i$ to the end of $L_{t+1}$. 
        \item Remove wire $i$ from and add wire $n+1$ to $O$.
        \item Define $V_{t+1} \coloneqq V^*_{t+1} \cup \{i\}$ and define $W_{t+1} \coloneqq \{n+1,n+2\}$.
        \item For $\tau \in [t]$, define $\theta_{\tau,n+1} = \theta_{\tau,n+2} = \bot$. Define $\theta_{t+1}$ to be equal to $\theta_t$ on the sets $V_1,\dots,V_t$, equal to $\theta^*_{t+1}$ on the set $V^*_{t+1}$, equal to $0$ on index $i$, and equal to $\bot$ everywhere else.
        \item Define a function $f_{t+1}^{x,r_1,\dots,r_t}(v_1,\dots,v_{t+1},w_{t+1})$ as follows, where $v_{t+1} = (v^*_{t+1},v^{(i)})$. First, compute $(\widehat{x}^{(O)},\widehat{z}^{(O)}) = h(x,v_1,\dots,v_t,v^*_{t+1},r_1,\dots,r_t)$, which includes $(\widehat{x}^{(i)},\widehat{z}^{(i)})$. Then, set $c = v^{(i)} \oplus \widehat{x}^{(i)}$, and output $(v_{t+1},\Gamma_c(w_{t+1}))$. This defines measurement $M_{\theta_{t+1},f_{t+1}^{(\cdot)}}$.
        \item Initialize $V^*_{t+2} \coloneqq \{n+2\}$, $\theta^*_{t+2}$ to be equal to 1 at index $n+2$ and $\bot$ everywhere else, and $L_{t+2}$ to be empty.
        \item Update $h$ to $h'$ as follows. The function $h'$ will now take an additional input bit $v^{(i)}$ as part of $v_{t+1}$, an additional input bit $v^{(n+2)}$ as part of $v^*_{t+2}$, and an additional input bit $r_{t+1}$. Its output will now include $(\widehat{x}^{(n+1)},\widehat{z}^{(n+1)})$ rather than $(\widehat{x}^{(i)},\widehat{z}^{(i)})$, computed as follows. Let $\widehat{x}^{(O)},\widehat{z}^{(O)} = h(x,v_1,\dots,v_t,v^*_{t+1},r_1,\dots,r_t)$ be the output of the original $h$, which includes $(\widehat{x}^{(i)},\widehat{z}^{(i)})$. Then the output of $h'$ includes $\widehat{x}^{(n+1)} \coloneqq v^{(i)} \oplus \widehat{x}^{(i)}$ and $\widehat{z}^{(n+1)} \coloneqq (v^{(i)} \oplus \widehat{x}^{(i)}) \cdot (v^{(n+2)} \oplus r_{t+1})$.
    \end{itemize}

\end{itemize}

This completes the description of the compiler. Observe that the $W_i$ wires consist of the two magic state wires used to implement the $i$'th $T$ gate (one initialized with $\ket{\phi_T}$ and the other initialized with $\ket{\phi_{PX}}$). The $\ket{\phi_T}$ wire is used as the control for a single CNOT gate just prior to the $i$'th measurement, and the $\ket{\phi_{PX}}$ wire is not touched until the $i$'th measurement. Thus, since $\Gamma_c$ is a standard basis projector, all the requirements of the \emph{standard-basis-collapsible $W$ wires} property (\cref{def:LM-properties}) are fulfilled.

\end{proof}

\section{Quantum State Obfuscation: Construction}

In this section, we define quantum state obfuscation, describe our construction, and show that it is correct. We will prove security in the following section.

\begin{definition}[Quantum State Obfuscation]\label{def:obf}
    For any $\epsilon = \epsilon(\secp)$, let $\cC_\epsilon$ be the set of families of $\epsilon$-pseudo-deterministic quantum programs (\cref{def:quantum-imp}), where each family $\{\ket{\psi_\secp},C_\secp\}_{\secp \in \bbN} \in \cC_\epsilon$ is associated with an induced family of maps $\{Q_\secp : \{0,1\}^{m(\secp)} \to \{0,1\}^{m'(\secp)}\}_{\secp \in \bbN}$. A quantum state obfuscator is a pair of QPT algorithms $(\QObf,\QEval)$ with the following syntax.
    \begin{itemize}
        \item $\QObf\left(1^\secp,\ket{\psi},C\right) \to \ket*{\widetilde{\psi}}$: The obfuscator takes as input the security parameter $1^\secp$ and a quantum program $(\ket{\psi},C)$, and outputs an obfuscated state $\ket*{\widetilde{\psi}}$.
        \item $\QEval\left(x,\ket*{\widetilde{\psi}}\right) \to y$: The evaluation algorithm takes an input $x \in \{0,1\}^{m(\secp)}$ and an obfuscated state $\ket*{\widetilde{\psi}}$, and outputs $y \in \{0,1\}^{m'(\secp)}$.
    \end{itemize}
    \textbf{Correctness} is defined as follows for any quantum program $\{\ket{\psi_\secp},C_\secp\}_{\secp \in \bbN}$.
    \[\forall x \in \{0,1\}^{m(\secp)}, \Pr\left[\QEval\left(x,\ket*{\widetilde{\psi}}\right) = Q_\secp(x) : \ket*{\widetilde{\psi}} \gets \QObf\left(1^\secp,\ket{\psi_\secp},C_\secp\right)\right] = 1-\negl(\secp).\]
    We define two notions of security with respect to some pseudo-deterministic parameter $\epsilon = \epsilon(\secp)$. 
    \begin{itemize}
        \item \textbf{Ideal Obfuscation}: For any QPT adversary $\{A_\secp\}_{\secp \in \bbN}$, there exists a QPT simulator $\{\Sim_\secp\}_{\secp \in \bbN}$ such that for any polynomial $n(\secp)$, program $\{\ket{\psi_\secp},C_\secp\}_{\secp \in \bbN} \in \cC_\epsilon$ with induced family of maps $\{Q_\secp : \{0,1\}^{m(\secp)} \to \{0,1\}^{m'(\secp)}\}_{\secp \in \bbN}$ such that $\ket{\psi_\secp}$ has at most $n(\secp)$ qubits and $C_\secp$ has at most $n(\secp)$ gates, and QPT distinguisher $\{D_\secp\}_{\secp \in \bbN}$,

        \begin{align*}&\bigg|\Pr\left[1 \gets D_\secp\left(A_\secp\left(\QObf\left(1^\secp,\ket{\psi_\secp},C_\secp\right)\right)\right)\right]\\ & ~~~~~ - \Pr\left[1 \gets D_\secp\left(\Sim_\secp^{Q_\secp}\left(1^\secp,n(\secp),m(\secp),m'(\secp)\right)\right)\right]\bigg| = \negl(\secp).\end{align*}
    
        \item \textbf{Indistinguishability Obfuscation}: For any polynomial $n(\secp)$, pair of families $\{\ket{\psi_{\secp,0}},C_{\secp,0}\}_{\secp \in \bbN},\allowbreak \{\ket{\psi_{\secp,1}},C_{\secp,1}\}_{\secp \in \bbN} \in \cC_\epsilon$ with the same induced map $\{Q_\secp : \{0,1\}^{m(\secp)} \to \{0,1\}^{m'(\secp)}\}_{\secp \in \bbN}$ such that $\ket{\psi_{\secp,0}}$ and $\ket{\psi_{\secp,1}}$ both have at most $n(\secp)$ qubits and $C_{\secp,0}$ and $C_{\secp,1}$ both have at most $n(\secp)$ gates, and QPT adversary $\{A_\secp\}_{\secp \in \bbN}$,
        \[\bigg|\Pr\left[1 \gets A_\secp\left(\QObf\left(1^\secp,\ket{\psi_{\secp,0}},C_{\secp,0}\right)\right)\right] - \Pr\left[1 \gets A_\secp\left(\QObf\left(1^\secp,\ket{\psi_{\secp,1}},C_{\secp,1}\right)\right)\right]\bigg| = \negl(\secp).\]
    \end{itemize}
\end{definition}

\begin{remark}[Classical Oracle Model]
In this work, we construct quantum state obfusation in the \emph{classical oracle model}. In this model, we allow $\QObf$ to additionally output the description of a classical deterministic functionality $O$, and both $\QEval$ and the adversary $A_\secp$ are granted quantum-accessible oracle access to $O$. Any scheme in the classical oracle model may be heuristically instantiated in the plain model by using a post-quantum indistinguishability obfuscator to obfuscate $O$ and include its obfuscation in the description of the state $\ket*{\widetilde{\psi}}$.
\end{remark}

Our construction of quantum state obfuscation in the classical oracle model makes use of the following ingredients.

\begin{itemize}
    \item Publicly-verifiable, linearly-homomorphic QAS with classically-decodable ZX measurements $(\Gen,\allowbreak\Enc,\allowbreak\LinEval,\allowbreak\Dec,\allowbreak\Ver)$, defined in \cref{sec:authentication}.
    \item Signature token $(\TokGen,\TokSign,\TokVer)$, defined in \cref{subsec:sig-tokens}.
    \item A pseudorandom function $F_k$ secure against superposition-query attacks \cite{6375347}.
\end{itemize}

For any polynomials $n=n(\secp)$, $m=m(\secp)$, and $m'=m'(\secp)$, let \[\left\{\ket*{\psi^\univ_{\secp,n,m,m'}},C^\univ_{\secp,n,m,m'}\right\}_{\secp \in \bbN}\] be the family of \emph{universal} $(n,m,m')$ quantum programs. That is, for any family of quantum programs $\{\ket{\phi_\secp},C_\secp\}_{\secp \in \bbN}$ where $\ket{\phi_\secp}$ has at most $n$ qubits, $C_\secp$ has as most $n$ gates, and $Q: \{0,1\}^{m} \to \{0,1\}^{m'}$, it holds that for all $\secp$ and inputs $x$, \[C^\univ_{\secp,n,m,m'}\left(\ket{x}\ket{C_\secp}\ket{\phi_\secp}\ket*{\psi^\univ_{\secp,n,m,m'}}\right) = C_\secp\left(\ket{x}\ket{\phi_\secp}\right).\]

By \cref{thm:LM-compiler}, for any $(n,m,m')$ family $\{\ket{\phi_\secp},C_\secp\}_{\secp \in \bbN}$, we can write each quantum program \[\ket{C_\secp}\ket{\phi_\secp}\ket*{\psi^\univ_{\secp,n,m,m'}},C^\univ_{\secp,n,m,m'}\] as an $\LM$ quantum program \[\left(\ket{\psi},\allowbreak\{L_i\}_{i \in [t+1]},\allowbreak\{\theta_i\}_{i \in [t+1]},\allowbreak\{f_i\}_{i \in [t]},g\right)\] that satisfies the \emph{standard-basis-collapsible $W$ wires} property (\cref{def:LM-properties}), where we have dropped the indexing by $\secp$ to reduce notational clutter. Note that \cref{thm:LM-compiler} guarantees that $\ket{\psi}$ contains the complete description of $(\ket{\phi_\secp},C_\secp)$, and that the classical part of the program $(\{L_i\}_{i \in [t+1]},\allowbreak\{\theta_i\}_{i \in [t+1]},\allowbreak\{f_i\}_{i \in [t]},\allowbreak g)$ only depends on $C^\univ_{\secp,n,m,m'}$. Thus, we consider everything but $\ket{\psi}$ to be public, and our obfuscator will take as input a quantum state $\ket{\psi}$, and its goal is to hide $\ket{\psi}$. Finally, we assume without loss of generality that each $r_i$ (being part of the output of $f_i$) is a single bit, which is convenient (though not strictly necessary) for describing the construction and proof, and is satisfied by the output of the compiler given in \cref{thm:LM-compiler}. 

The construction is given in \proref{fig:qobf}, and incorporates the following public parameters:
 \begin{itemize}
        \item Security parameter $\secp$.
        \item Classical part of the $\LM$ quantum program $\{L_i\}_{i \in [t+1]},\{\theta_i\}_{i \in [t+1]},\{f_i\}_{i \in [t]},g$. This information determines the number of qubits $n = \poly(\secp)$ in the input state $\ket{\psi}$, classical input size $m = \poly(\secp)$, and classical output size $m'(\secp)$. Recall from \cref{subsec:auth-def} that each $\theta_i$ defines subsets \[\MS_{\theta_i},\MS_{\theta_i,0},\MS_{\theta_i,1},\MS_{\theta_i,\bot},\widetilde{\MS}_{\theta_i},\widetilde{\MS}_{\theta_i,0},\widetilde{\MS}_{\theta_i,1},\widetilde{\MS}_{\theta_i,\bot},\] and in what follows we drop the $\theta$ and write these subsets as \[\MS_i,\MS_{i,0},\MS_{i,1},\MS_{i,\bot},\widetilde{\MS}_i,\widetilde{\MS}_{i,0},\widetilde{\MS}_{i,1},\widetilde{\MS}_{i,\bot}.\]
        \item Derived security parameter $\kappa \coloneqq \max\{\secp,n^4\} = \poly(\secp)$. 
    \end{itemize}

We will also make use of the following notation. Given a register $\regX$ and classical functionality $\sF$, we let \[(y,\regX) \gets \sF(\regX)\] denote the result of initializing a new register $\regY$, coherently applying the map \[\ket{x}^\regX\ket{0}^\regY \to \ket{x}^\regX\ket{\sF(x)}^\regY,\] and then measuring register $\regY$ to obtain output $y$.

\protocol{Quantum State Obfuscation}{Construction of quantum state obfuscation.}{fig:qobf}{

\noindent $\QObf\left(1^\secp,\ket{\psi}\right)$: 

\begin{itemize}
    \item Sample $k \gets \Gen(1^\kappa,n)$, and compute $\ket*{\psi_k} = \Enc_k(\ket{\psi})$.
    \item Sample a signature token $(\vk,\ket{\sk}) \gets \TokGen(1^\kappa)$.
    \item Sample $k' \gets \{0,1\}^\secp$ for PRF $F_{k'} : \{0,1\}^* \to \{0,1\}^\kappa$ and let $H(\cdot) \coloneqq F_{k'}(\cdot)$.

    \item For each $i \in [t]$, define the function $\sF_i\left(x,\sigma_x,\widetilde{v}_1,\dots,\widetilde{v}_i,\widetilde{w}_i,\ell_1,\dots,\ell_{i-1}\right)$:
    \begin{itemize}
        \item Output $\bot$ if $\TokVer(\vk,x,\sigma_x) = \bot$.
        \item For each $\iota \in [i-1]$, let \[\ell_{\iota,0} = H(x,\sigma_x,\widetilde{v}_1,\dots,\widetilde{v}_{\iota},\ell_1,\dots,\ell_{\iota-1},0),  ~~ \ell_{\iota,1} = H(x,\sigma_x,\widetilde{v}_1,\dots,\widetilde{v}_{\iota},\ell_1,\dots,\ell_{\iota-1},1),\]  and output $\bot$ if $\ell_{\iota,0} = \ell_{\iota,1}$ or $\ell_\iota \notin \{\ell_{\iota,0},\ell_{\iota,1}\}$. Otherwise, let $r_\iota$ be such that $\ell_\iota = h_{\iota,r_\iota}$.%
        \item Compute $(v_1,\dots,v_i,w_i) = \Dec_{k,L_i \dots L_1,\theta_i}(\widetilde{v}_1,\dots,\widetilde{v}_i,\widetilde{w}_i)$ and output $\bot$ if the result is $\bot$.
        \item Compute $(\cdot,r_i) = f_i^{x,r_1,\dots,r_{i-1}}(v_1,\dots,v_i,w_i)$.
        \item Set $\ell_i \coloneqq H(x,\sigma_x,\widetilde{v}_1,\dots,\widetilde{v}_i,\ell_1,\dots,\ell_{i-1},r_i),$ and output $(\widetilde{v}_i,\ell_i)$.
   
    \end{itemize}
    \item Define the function $\sG(x,\sigma_x,\widetilde{v}_1,\dots,\widetilde{v}_{t+1},\ell_1,\dots,\ell_t)$:
\begin{itemize}
        \item Output $\bot$ if $\TokVer(\vk,x,\sigma_x) = \bot$.
        \item For each $\iota \in [t]$, let \[\ell_{\iota,0} = H(x,\sigma_x,\widetilde{v}_1,\dots,\widetilde{v}_\iota,\ell_1,\dots,\ell_{\iota-1},0), ~~ \ell_{\iota,1} = H(x,\sigma_x,\widetilde{v}_1,\dots,\widetilde{v}_\iota,\ell_1,\dots,\ell_{\iota-1},1),\] and output $\bot$ if $\ell_{\iota,0} = \ell_{\iota,1}$ or $\ell_\iota \notin \{\ell_{\iota,0},\ell_{\iota,1}\}$. Otherwise, let $r_\iota$ be such that $\ell_\iota = \ell_{\iota,r_\iota}$.
        \item Compute $(v_1,\dots,v_{t+1}) = \Dec_{k,L_{t+1} \dots L_1,\theta_{t+1}}(\widetilde{v}_1,\dots,\widetilde{v}_{t+1})$ and output $\bot$ if the result is $\bot$.
        \item Output $y = g^{x,r_1,\dots,r_t}(v_1,\dots,v_{t+1})$.
    \end{itemize}
    \item Output $\ket*{\widetilde{\psi}} = \ket*{\psi_k}\ket{\sk}, O = \left(\sF_1,\dots,\sF_t,\sG\right).$
\end{itemize}

$\QEval^O\left(x,\ket*{\widetilde{\psi}}\right)$:%
\begin{itemize}
    \item Parse $\ket*{\widetilde{\psi}} = \ket*{\psi_k}\ket{\sk}, O = \left(\sF_1,\dots,\sF_t,\sG\right)$.
    \item Sample $\sigma_x \gets \TokSign(x,\ket{\sk})$.
    \item Initialize a register $\regC$ with $\ket*{\psi_k}$, and do the following for $i \in [t]$:
    \begin{itemize}
        \item $\regC \gets H^{\widetilde{\MS}_{i,1}}\LinEval_{L_i}(\regC)$.
        \item Measure $\left(\widetilde{v}_i,\ell_i,\regC_{\widetilde{\MS}_i}\right) \gets \sF_i\left(x,\sigma_x,\regC_{\widetilde{\MS}_i},\ell_1,\dots,\ell_{i-1}\right)$.
        \item $\regC \gets H^{\widetilde{\MS}_{i,1}}(\regC)$.
    \end{itemize}
    \item $\regC \gets H^{\widetilde{\MS}_{t+1,1}}\LinEval_{L_{t+1}}(\regC)$.
    \item Measure $y \gets \sG\left(x,\sigma_x,\regC_{\widetilde{\MS}_{t+1}},\ell_1,\dots,\ell_t\right)$ and output $y$.
\end{itemize}

}

\begin{theorem}
The scheme described in \cref{fig:qobf} is a quantum state obfuscator that satisfies correctness (\cref{def:obf}).
\end{theorem}

\begin{proof}
Fix the classical part of an $\LM$ quantum program $\{L_i\}_{i \in [t+1]},\{\theta_i\}_{i \in [t+1]},\{f_i\}_{i \in [t]},g$, and let $x,\ket{\psi}$ be such that there exists $y$ such that 
\[\Pr\left[\LMEval\left(x,\ket{\psi},\{L_i\}_{i \in [t+1]},\{\theta_i\}_{i \in [t+1]},\{f_i\}_{i \in [t]},g\right) \to y\right] = 1-\negl(\secp).\] Technically, we mean an infinite family of programs, inputs, states, and outputs, parameterized by the security parameter $\secp$, but we keep this implicit. We will show via a sequence of hybrids that 

\[\Pr\left[y^* = y : \begin{array}{r}\ket*{\widetilde{\psi}},O \gets \QObf(1^\secp,\ket{\psi}) \\ y^* \gets \QEval^O\left(x,\ket*{\widetilde{\psi}}\right) \end{array}\right] = 1-\negl(\secp).\]

Each hybrid will describe a distribution over $y^*$, beginning with the distribution above, which we denote $\cH_0$. 

\begin{itemize}
    \item $\cH_1$: This is the same as $\cH_0$ except that $H(\cdot)$ is defined to be a uniformly random function with range $\{0,1\}^\kappa$ rather than the PRF $F_{k'}$.
    \item $\cH_2$: This is the same as $\cH_1$ except that the functions $\sF_1,\dots,\sF_t,\sG$ ignore their input $\sigma_x$ and don't apply $\TokVer$.
    \item $\cH_3$: This is the same as $\cH_2$ except that instead of inputting and outputting the labels $\ell_\iota$, the functions $F_1,\dots,F_t,G$ directly input and output the bits $r_\iota$. That is, these functions are defined as follows.\\

    $\sF_i(x,\sigma_x,\widetilde{v}_1,\dots,\widetilde{v}_i,\widetilde{w}_i,r_1,\dots,r_{i-1})$:
    \begin{itemize}
        \item Compute $(v_1,\dots,v_i,w_i) = \Dec_{k,L_i \dots L_1,\theta_i}(\widetilde{v_1},\dots,\widetilde{v}_i,\widetilde{w}_i)$ and output $\bot$ if the result is $\bot$.
        \item Compute $(\cdot,r_i) = f_i^{x,r_1,\dots,r_{i-1}}(v_1,\dots,v_i,w_i)$.
        \item Output $(\widetilde{v}_i,r_i)$.
    \end{itemize}

    $\sG(x,\sigma_x,\widetilde{v}_1,\dots,\widetilde{v}_{t+1},r_1,\dots,r_t):$
    \begin{itemize}
        \item Compute $(v_1,\dots,v_{t+1}) = \Dec_{k,L_{t+1} \dots L_1,\theta_{t+1}}(\widetilde{v}_1,\dots,\widetilde{v}_{t+1})$ and output $\bot$ if the result is $\bot$.
        \item Output $y = g^{x,r_1,\dots,r_t}(v_1,\dots,v_{t+1})$.
    \end{itemize}

    \item $\cH_4$: This is the same as $\cH_3$ except that the functions $\sF_1,\dots,\sF_t$ output $v_\iota$ rather than $\widetilde{v}_\iota$. That is, these functions are defined as follows.\\

    $\sF_i(x,\sigma_x,\widetilde{v}_1,\dots,\widetilde{v}_i,\widetilde{w}_i,r_1,\dots,r_{i-1})$:
    \begin{itemize}
        \item Compute $(v_1,\dots,v_i,w_i) = \Dec_{k,L_i \dots L_1,\theta_i}(\widetilde{v}_1,\dots,\widetilde{v}_i,\widetilde{w}_i)$ and output $\bot$ if the result is $\bot$.
        \item Compute $(\cdot,r_i) = f_i^{x,r_1,\dots,r_{i-1}}(v_1,\dots,v_i,w_i)$.
        \item Output $(v_i,r_i)$.
    \end{itemize}

    $\sG(x,\sigma_x,\widetilde{v}_1,\dots,\widetilde{v}_t,\widetilde{v}_{t+1},r_1,\dots,r_t):$
    \begin{itemize}
        \item Compute $v_1,\dots,v_{t+1} = \Dec_{k,L_{t+1} \dots L_1,\theta_{t+1}}(\widetilde{v}_1,\dots,\widetilde{v}_{t+1})$ and output $\bot$ if the result is $\bot$.
        \item Output $y = g^{x,r_1,\dots,r_t}(v_1,\dots,v_{t+1})$.
    \end{itemize}

\end{itemize}

To complete the proof, we combine the following observations.

\begin{itemize}
    \item $\cH_0 \approx_{\negl(\secp)} \cH_1$: This follows from the (superposition-query) security of the PRF.
    \item $\cH_1 \equiv \cH_2$: This follows from the correctness of the signature token (\cref{def:token-correctness}).
    \item $\cH_2 \approx_{\negl(\secp)} \cH_3$: The only difference between these hybrids occurs if in $\cH_2$, a query to $\sF_i$ or $\sG$ outputs $\bot$ due to the fact that $\ell_{\iota,0} = \ell_{\iota,1}$. Since $H$ is a uniformly random function, each $\ell_{\iota,0} = \ell_{\iota,1}$ with probability $1/2^\kappa = \negl(\secp)$, and the observation follows because $t = \poly(\secp)$.
    \item $\cH_3 \equiv \cH_4$: Starting with $\cH_4$, in which the logical measurements of $v_1,\dots,v_{t+1}$ are performed, we can imagine, after the $i$'th measurement, further collapsing the $V_i$ register to obtain the outcome $\widetilde{v}_i$ as in $\cH_3$. This has no effect on the rest of the computation, since these registers are no longer computed on after the $i$'th measurement, and will continue to decode to $v_i$.

    \item $\cH_4 \equiv \LMEval\left(x,\ket{\psi},\{L_i\}_{i \in [t+1]},\{\theta_i\}_{i \in [t+1]},\{f_i\}_{i \in [t]},g\right)$: This follows from the correctness of the authentication scheme (\cref{def:auth-correct}). Indeed, for a given key $k \in \Gen(1^\kappa,n)$, we can write the distribution sampled by $\cH_4$ as follows:

    \begin{itemize}
        \item Initialize register $\regC$ to $\Enc_k(\ket{\psi})$.
        \item Compute $((v_1,r_1),\regC) \gets \widetilde{M}_{\theta_1,f_1^x,k,L_1} \circ \LinEval_{L_1}(\regC)$.
        \item \dots
        \item Compute $((v_t,r_t),\regC) \gets \widetilde{M}_{\theta_t,f_t^{x,r_1,\dots,r_{t-1}},k,L_t \dots L_1} \circ \LinEval_{L_t}(\regC)$.
        \item Compute output $y \gets \widetilde{M}_{\theta_{t+1},g^{x,r_1,\dots,r_t},k,L_{t+1} \dots L_1} \circ \LinEval_{L_{t+1}}(\regC)$.
    \end{itemize}

    Now, we apply the expression in the definition of correctness (\cref{def:auth-correct}) to the first measurement to obtain an equivalent sampling procedure:

    \begin{itemize}
        \item Initialize register $\regM$ to $\ket{\psi}$.
        \item Compute $((v_1,r_1),\regM) \gets L_1^\dagger \circ M_{\theta_1,f_1^x} \circ L_1(\regM)$.
        \item Compute $\regC \gets \Enc_k(\regM)$
        \item Compute $((v_2,r_2),\regC) \gets \widetilde{M}_{\theta_2,f_2^{x,r_1},k,L_2L_1} \circ \LinEval_{L_2 L_1}(\regC)$.
        \item \dots
        \item Compute $((v_t,r_t),\regC) \gets \widetilde{M}_{\theta_t,f_t^{x,r_1,\dots,r_{t-1}},k,L_t \dots L_1} \circ \LinEval_{L_t}(\regC)$.
        \item Compute output $y \gets \widetilde{M}_{\theta_{t+1},g^{x,r_1,\dots,r_t},k,L_{t+1} \dots L_1} \circ \LinEval_{L_{t+1}}(\regC)$.
    \end{itemize}

    By applying the expression iteratively for each measurement, we obtain:

    \begin{itemize}
        \item Initialize register $\regM$ to $\ket{\psi}$.
        \item Compute $((v_1,r_1),\regM) \gets L_1^\dagger \circ M_{\theta_1,f_1^x} \circ L_1(\regM)$.
        \item Compute $((v_2,r_2),\regM) \gets L_1^\dagger L_2^\dagger \circ M_{\theta_2,f_2^{x,r_1}} \circ L_2L_1(\regM)$.
        \item \dots
        \item Compute $((v_t,r_t),\regM) \gets L_1^\dagger \dots L_t^\dagger \circ M_{\theta_t,f_t^{x,r_1,\dots,r_{t-1}}} \circ L_t\dots L_1(\regM)$.
        \item Compute output $y \gets M_{\theta_{t+1},g^{x,r_1,\dots,r_t}} \circ L_{t+1} \dots L_1(\regM)$.
    \end{itemize}

    By canceling $L_i^\dagger L_i = \cI$, we obtain $ \LMEval\left(x,\ket{\psi},\{L_i\}_{i \in [t+1]},\{\theta_i\}_{i \in [t+1]},\{f_i\}_{i \in [t]},g\right)$:

    \begin{itemize}
        \item Initialize register $\regM$ to $\ket{\psi}$.
        \item Compute $((v_1,r_1),\regM) \gets M_{\theta_1,f_1^x} \circ L_1(\regM)$.
        \item Compute $((v_2,r_2),\regM) \gets M_{\theta_2,f_2^{x,r_1}} \circ L_2(\regM)$.
        \item \dots
        \item Compute $((v_t,r_t),\regM) \gets M_{\theta_t,f_t^{x,r_1,\dots,r_{t-1}}} \circ L_t(\regM)$.
        \item Compute output $y \gets M_{\theta_{t+1},g^{x,r_1,\dots,r_t}} \circ L_{t+1}(\regM)$.
    \end{itemize}
    
\end{itemize}

\end{proof}

\section{Quantum State Obfuscation: Security}\label{sec:security}

In this section, we prove that the scheme described in \cref{fig:qobf} is an ideal quantum state obfuscator in the classical oracle model. We provide some intuition in \cref{subsec:proof-overview} and set up some notation in \cref{subsec:notation} before coming to the formal proof in \cref{subsec:main-theorem} - \cref{subsec:mapping-hardness}.

\subsection{Proof Intuition}\label{subsec:proof-overview}

We begin by discussing three main ideas used in our proof. This will not be a step-by-step outline of the proof, rather, it will try to convey the main intuitive ideas. Broadly speaking, we will want to simulate the oracles $\sF_1,\dots,\sF_t,\sG$ so that they no longer require access to the decoding functionality of the authentication scheme, and $\sG$ can get by with just oracle access to the induced functionality $Q$. Once this is done, we can appeal to privacy of the authentication scheme (\cref{def:auth-privacy}) in order to switch $\ket{\psi}$ to $\ket{0^n}$, thus removing all information about the input state. 

The first idea below will help us simulate the $\sG$ oracle, the second will us help simulate the $\sF$ oracles, and the third is a way to extract signature tokens from the adversary using a purified random oracle, which we will use when proving the indistinguishability of the simulated oracles.

\paragraph{Proving soundness by induction.} As mentioned in the technical overview (\cref{subsec:overview-GC}), one of the main steps in our proof of security is to show the following soundness guarantee. Fix any input $x^*$. Then we would like to show that, given $\ket*{\psi}$ and oracle access to $\sF_1,\dots,\sF_t,\sG$, the adversary cannot prepare a "bad" query $(x^*,\sigma_{x^*},\widetilde{v}_1,\dots,\widetilde{v}_{t+1},\ell_1,\dots,\ell_t)$ with the property that \[\sG(x^*,\sigma_{x^*},\widetilde{v}_1,\dots,\widetilde{v}_{t+1},\ell_1,\dots,\ell_t) \notin \{Q(x^*),\bot\}.\] That is, if $\sG$ does not abort on an input that starts with $x^*$, then it better be the case that it returns the correct output $Q(x^*)$.

To simplify the discussion for now, let's consider the simpler case of a garbled program, which only allows the adversary to evaluate on a single input $x^*$. That is, suppose we hard-code $x^*$ into the oracles $\sF_1[x^*],\dots,\sF_t[x^*],\sG[x^*]$, which now only accept inputs that begin with $x^*$.

Our goal will be to show that the adversary is ``forced'' to follow an honest evaluation path on input $x^*$. Since the honest evaluation path actually branches at each measurement, we will essentially analyze each of these possible branching executions. To do so, let's suppose by induction that the soundness condition holds for any program with $t-1$ measurement layers. Then, for each possible outcome $(v_1,r_1)$ of the first measurement (using input $x^*$) that occurs with non-zero probability, define $\Pi[x^*,v_1,r_1]$ to be the projector onto the space of initial states $\ket{\psi}$ that produce that outcome. Thus, we can write \[\ket{\psi} = \sum_{v_1,r_1} \Pi[x^*,v_1,r_1]\ket{\psi},\] and analyze each component $\ket*{\widetilde{\psi}[x^*,v_1,r_1]} \coloneqq \Enc_k(\Pi[x^*,v_1,r_1]\ket{\psi})$ separately. 

The key step is to show that, if the adversary is initialized with $\ket*{\widetilde{\psi}[x^*,v^*_1,r^*_1]}$ for some $(v_1^*,r_1^*)$, then we can \emph{hard-code} the measurement results $(v^*_1,r^*_1)$ into the oracles $\sF_1[x^*],\dots,\sF_t[x^*]$ without the adversary noticing. That is, we define oracles $\sF_1[x^*,v^*_1,r^*_1],\dots,\sF_t[x^*,v^*_1,r^*_1]$ that operate like $\sF_1[x^*],\dots,\sF_t[x^*]$, except that $\sF_1[x^*,v^*_1,r^*_1]$ always outputs the label representing $r^*_1$, and $\sF_2[x^*,v^*_1,r^*_1],\dots,\sF_t[x^*,v^*_1,r^*_1]$ use $(v^*_1,r^*_1)$ instead of decoding their inputs $\widetilde{v}_1$ and $\ell_1$.

We will prove the indistinguishability of $\sF_1[x^*],\dots,\sF_t[x^*]$ and $\sF_1[x^*,v^*_1,r^*_1],\dots,\sF_t[x^*,v^*_1,r^*_1]$ by reducing to security of the authentication scheme. Note that distinguishing these oracles requires the adversary to find a \emph{differing input} to one of the oracles. Now, assuming that the oracles can be simulated using $\Ver_{k,\cdot,\cdot}(\cdot)$ instead of $\Dec_{k,\cdot,\cdot}(\cdot)$ (which we have yet to argue, but will address in the following section), this means that it suffices to show that the adversary cannot map \[\ket*{\widetilde{\psi}[x^*,v^*_1,r^*_1]} \to \ket*{\widetilde{\psi}[x^*,v_1,r_1]}\] for some $(v_1,r_1) \neq (v^*_1,r^*_1)$ just given access to the verification oracle $\Ver_{k,\cdot,\cdot}(\cdot)$. However, doing so would certainly imply that the adversary can change the measurement results of an authenticated state (given the verification oracle), which violates the security of the authentication scheme. 

Finally, we view the state $\ket*{\widetilde{\psi}[x^*,v^*_1,r^*_1]}$ and oracles $\sF_1[x^*,v^*_1,r^*_1],\allowbreak\dots,\allowbreak\sF_t[x^*,v^*_1,r^*_1]$ as an example of a garbled $(t-1)$-layer program, and finish the proof of soundness by appealing to the induction hypothesis.

The formal inductive argument is given in \cref{subsec:induction}.

\paragraph{Collapsing the $\sF$ oracles.} Now, we address the claim made above that the $\sF$ oracles can be simulated given $\Ver_{k,\cdot,\cdot}(\cdot)$ instead of $\Dec_{k,\cdot,\cdot}(\cdot)$. Note that we can only hope that this simulation is indistinguishable to an adversary with no access to the $\sG$ oracle, since the real $\sF$ oracles can be used to actually implement the computation $x \to Q(x)$, while the simulated oracles cannot since they don't actually decode their inputs. However, it turns out that this suffices for us, since we can simulate the $\sG$ oracle when we need to apply this indistinguishability.

The main idea is to ``collapse'' the oracles $\sF_1,\dots,\sF_t$, showing that the adversary cannot distinguish them from oracles $\FSim_1,\dots,\FSim_t$ that \emph{always} output either the ``zero'' label \[H(x,\sigma_x,\widetilde{v}_1,\dots,\widetilde{v}_i,\ell_1,\dots,\ell_{i-1},0)\] or $\bot$. These oracles now do not have to actually run $f_i^{x,r_1,\dots,r_{i-1}}(v_1,\dots,v_i,w_i)$ to compute the bit $r_i$, meaning that the decoding operation in $\sF_i$ can be replaced with a verification operation.

But how do we show that the oracles can be collapsed? Again, we will use the idea of splitting $\ket{\psi}$ up into orthogonal components and analyzing each component separately. Here is where we make use of the \emph{standard-basis-collapsible $W$ wires} property of the $\LM$ quantum program (\cref{def:LM-properties}). In particular, we will define the orthogonal components by measuring the wires $W_1,\dots,W_t$ in the standard basis. That is, we will write \[\ket{\psi} = \sum_{w}\Pi[w]\ket{\psi},\] where $\Pi[w]$ is the projection of wires $W_1,\dots,W_t$ onto standard basis measurement results $w = (w_1,\dots,w_t)$. 

The point is that if the oracle $\sF_i$ only receives inputs that include encodings $\widetilde{w} = (\widetilde{w}_1,\dots,\widetilde{w}_t)$ of some \emph{fixed} $w = (w_1,\dots,w_t)$, then, for each "prefix" $(x,\sigma_x,\widetilde{v}_1,\dots,\widetilde{v}_i,\ell_1,\dots,\ell_{i-1})$, it will only ever query the random oracle $H$ on \[\text{either} \ \ \  (x,\sigma_x,\widetilde{v}_1,\dots,\widetilde{v}_i,\ell_1,\dots,\ell_{i-1},0) \ \ \ \text{or} \ \ \ (x,\sigma_x,\widetilde{v}_1,\dots,\widetilde{v}_i,\ell_1,\dots,\ell_{i-1},1),\] where the last bit is a deterministic function of $w$ and the prefix. But since each of these values is distributed as a uniformly random string, this behavior is identical (from the adversary's perspective) to, for each prefix, always querying the zero label \[H(x,\sigma_x,\widetilde{v}_1,\dots,\widetilde{v}_i,\ell_1,\dots,\ell_{i-1},0).\]

Thus, it suffices to show, roughly, that given $\ket*{\widetilde{\psi}[w]} \coloneqq \Enc_k(\Pi[w]\ket{\psi})$, the adversary cannot map \[\ket*{\widetilde{\psi}[w]} \to \ket*{\widetilde{\psi}[w']}\] for $w' \neq w$, given access to the verification oracle $\Ver_{k,\cdot,\cdot}(\cdot)$. This again follows directly from security of our authentication scheme.

The ideas sketched here are used to simulate the $\sF$ oracles in our main sequence of hybrids given in \cref{subsec:main-theorem}, and also in lower-level hybrids in \cref{subsec:mapping-hardness}.

\paragraph{Extracting signature tokens.} Recall that the inductive argument sketched above for proving the soundness condition assumed that the oracles only respond on a single fixed input $x^*$. Unfortunately, as discussed in \cref{subsec:overview-GC}, the situation gets more complicated when we grant the adversary access to the oracles on any input $x$ of their choice. The reason is that it is no longer clear that the adversary cannot perform the map 

\[\ket*{\widetilde{\psi}[x^*,v^*_1,r^*_1]} \to \ket*{\widetilde{\psi}[x^*,v_1,r_1]}\] by using an oracle query on an input $x \neq x^*$. Indeed, while the $(V_1,W_1)$ registers of $\ket*{\widetilde{\psi}[x^*,v^*_1,r^*_1]}$ are collapsed to a state that yields a fixed $(v^*_1,r^*_1) \gets f_1^{x^*}(V_1,W_1)$, applying $f_1^{x}(V_1,W_1)$ for some $x \neq x^*$ might \emph{disturb} the registers $V_1,W_1$ (since $f_1^{x^*}$ and $f_1^{x}$ might be different functions!), thus changing the outcome of $f_1^{x^*}(V_1,W_1)$.

Now, as discussed in \cref{subsec:overview-GC}, we use signature tokens to prevent this potential attack. Recall that the oracle $\sF_1$ only responds on input $x$ if additionally given a valid signature token $\sigma_x$. Thus, we will want to formalize the following claim: if the adversary uses $\sF_1(x,\dots)$ to perform some "non-trivial" operation on the authenticated state $\ket*{\widetilde{\psi}[x^*,v^*_1,r^*_1]}$, it is possible to \emph{extract} a valid signature $\sigma_x$ from the adversary. Once this $\sigma_x$ is extracted, the security of the signature token scheme implies that the adversary won't be able to continue evaluating on $x^*$, and, in particular, we won't have to worry about the adversary breaking the soundness condition for input $x^*$.

We will show this claim by purifying the random oracle \cite{C:Zhandry19}, introducing a ``database'' register that is initialized with a different state in uniform superposition for each input to $H$. Then, we'll argue that if the adversary has used $\sF_i(x,\cdot)$ to execute a measurement on the authenticated state, the database register must be disturbed at some inputs that begin with $(x,\sigma_x)$. Thus, an extractor can simply measure the database register in the Hadamard basis, and obtain a signature on $x$ by observing which registers were no longer in uniform superposition.

For the purpose of this overview, we consider a simplified version of this problem that still conveys the fundamental ideas in our proof. We'll take the random oracle to have a single bit of output, only consider the authentication of a single qubit state, and analyze the concrete authentication scheme based on coset states (though we stress that our eventual proof just makes use of generic properties of the authentication scheme). Here is the setup:
\begin{itemize}
    \item An authentication key $k = (S,\Delta,x,z)$ is sampled. 
	\item The adversary $A$ is given an authenticated 0 state $X^xZ^z\ket{S}$ along with access to the following oracle $O$ that can be used to implement a logical Hadamard basis measurement:
	\[O(\widetilde{v}) = \begin{cases}H(0) \text{ if } \widetilde{v} \in \widehat{S} + z \\ H(1) \text{ if } \widetilde{v} \in \widehat{S} + \widehat{\Delta}+ z \\ \bot \text{ otherwise}\end{cases}, \] where $H$ is a random oracle $\{0,1\} \to \{0,1\}$. $H$ will be purified and implemented using a database register initialized to $\ket{++}$.
	\item We claim that the adversary cannot produce any vector in $S+\Delta + x$ (that is, a vector in the support of an authenticated 1 state) at the same time that the database register is in the state $\ket{++}$:
	\[\E\left[\big\| \left(\Pi[S+\Delta+x] \otimes \ketbra{++}{++} \right)A^{O}\left(X^xZ^z\ket{S}\right)\ket{++}\big\|^2\right] = \negl(\secp).\]  To be clear, $A$ operates on input state $X^xZ^z\ket{S}$ (along with some potential extra workspace), and the database registers are operated on by $O$ when answering $A$'s queries.
\end{itemize}

Notice that it is easy for the adversary to produce just a vector in $S+\Delta+x$ (with constant probability) by using the oracle $O$ to honestly to implement a Hadamard basis measurement. The trick is to show that once they do this, it is impossible for them to make queries to $O$ that return the state of the database to $\ket{++}$, while still remembering their vector in  $S+\Delta+x$.

Our first step is to decompose $X^xZ^z\ket{S}$ into orthogonal components corresponding to the authenticated plus and minus state. That is, \[X^xZ^z\ket{S} = \frac{1}{\sqrt{2}}H^{\otimes 2\secp+1}X^zZ^x\ket*{\widehat{S}} + \frac{1}{\sqrt{2}}H^{\otimes 2\secp+1}X^zZ^x\ket*{\widehat{S} + \widehat{\Delta}} \coloneqq  \ket*{\widetilde{+}} + \ket*{\widetilde{-}}.\] Then, for $b \in \{0,1\}$, we define oracles 

\[O[b](\widetilde{v}) = \begin{cases}H(b) \text{ if } \widetilde{v} \in \widehat{S}_{\widehat{\Delta}} + z  \\ \bot \text{ otherwise}\end{cases}, \] that are identical to $O$, except that they always query the random oracle on bit $b$. Then we observe that  

\begin{align*}
	A^{O}\ket*{\widetilde{+}}\approx_{\negl(\secp)} A^{O[0]}\ket*{\widetilde{+}}, ~~~~~ \text{and} ~~~~ A^{O}\ket*{\widetilde{-}}\approx_{\negl(\secp)} A^{O[1]}\ket*{\widetilde{-}}.
\end{align*}

This follows from the security of the authentication scheme, which implies that $A$ cannot map between $\ket*{\widetilde{+}}$ and $\ket*{\widetilde{-}}$. That is, on input $\ket*{\widetilde{+}}$, $A$ won't be able to find any input on which $O$ and $O[0]$ differ, and on input $\ket*{\widetilde{-}}$, $A$ won't be able to find any input on which $O$ and $O[1]$ differ.

Next, we observe that the state of the system that results from using oracle $O[1]$ is actually equivalent to the state that results from first swapping the database registers of $H$, using $O[0]$, and then swapping back. Thus, it holds that

\begin{align*}
	&\E\left[\big\| \left(\Pi[S+\Delta+x] \otimes \ketbra{++}{++} \right)A^{O}\left(X^xZ^z\ket{S}\right)\ket{++}\big\|^2\right] \\
	&= \E\left[\big\| \left(\Pi[S+\Delta+x] \otimes \ketbra{++}{++} \right)\left(A^{O}\ket*{\widetilde{+}}\ket{++} + A^{O}\ket*{\widetilde{-}}\ket{++}\right)\big\|^2\right]\\
	&\approx_{\negl(\secp)} \E\left[\big\| \left(\Pi[S+\Delta+x] \otimes \ketbra{++}{++} \right)\left(A^{O[0]}\ket*{\widetilde{+}}\ket{++} + A^{O[1]}\ket*{\widetilde{-}}\ket{++} \right)\big\|^2\right]\\
	&=\E\left[\big\| \left(\Pi[S+\Delta+x] \otimes \ketbra{++}{++} \right)\left(A^{O[0]}\ket*{\widetilde{+}}\ket{++} + \mathsf{SWAP}A^{O[0]}\ket*{\widetilde{-}}\mathsf{SWAP}\ket{++} \right)\big\|^2\right]\\
	&=\E\left[\big\| \Pi[S+\Delta+x] \left(\ketbra{++}{++}A^{O[0]}\ket*{\widetilde{+}}\ket{++} + \ketbra{++}{++}\mathsf{SWAP}A^{O[0]}\ket*{\widetilde{-}}\mathsf{SWAP}\ket{++} \right)\big\|^2\right]\\
	&=\E\left[\big\| \Pi[S+\Delta+x] \left(\ketbra{++}{++}A^{O[0]}\ket*{\widetilde{+}}\ket{++} + \ketbra{++}{++}A^{O[0]}\ket*{\widetilde{-}}\ket{++} \right)\big\|^2\right]\\
	&=\E\left[\big\| \left(\Pi[S+\Delta+x] \otimes \ketbra{++}{++}\right)\left(A^{O[0]}\ket*{\widetilde{+}}\ket{++} + A^{O[0]}\ket*{\widetilde{-}}\ket{++} \right)\big\|^2\right]\\
	&=\E\left[\big\| \left(\Pi[S+\Delta+x] \otimes \ketbra{++}{++}\right)A^{O[0]}X^xZ^z\ket{S}\ket{++}\big\|^2\right]\\
	&=\negl(\secp),
\end{align*}

where the last step follows from security of the authentication scheme, since $O[0]$ can be implemented with just the verification oracle of the authentication scheme. Note that we crucially used the fact that we are projecting back onto $\ketbra{++}{++}$ in the step where we remove the left-most $\mathsf{SWAP}$ operation, which follows because $\mathsf{SWAP}\ket{++} = \ket{++}$. Indeed, as discussed above, the claim would not be true without the projection onto $\ketbra{++}{++}$, since the adversary \emph{can} obtain a vector in $\Pi[S+\Delta+x]$ while disturbing the database register.

To conclude, we note that this same logic can be extended to more general measurements on more general authenticated states, which ultimately will be used to extract a valid signature on $x$ from any adversary that is actively using the oracles $\sF_1,\dots,\sF_t,\sG$ to evaluate the computation on input $x$. This implies that the adversary cannot launch mixed input attacks, which is one of the main hurdles to overcome in proving the security of our quantum state obfuscator.

The ideas sketched here are used in \cref{subsec:mapping-hardness}, and in particular in the proof of \cref{lemma:mapping-hardness-3}.

\subsection{Notation}\label{subsec:notation}

In this section, we review some important notation and define new notation that will be used throughout the proof. 
\begin{itemize}
    \item $\ket{\psi}$ is the $n$-qubit input state.
    \item $\{L_i\}_{i \in [t+1]},\{\theta_i\}_{i \in [t+1]},\{f_i\}_{i \in [t]},g$ is the classical part of the description of an $\LM$ quantum program.
    \item Parameters: $m$ is the size of the classical input, and $\kappa$ is a derived security parameter that is sufficiently larger than $\secp,n$. We will often use the facts that $n \geq t,m$ and $\kappa = \omega(n)$.
    \item See \cref{def:LM-circuit} for the definition of sets $(V_1,\dots,V_{t+1})$ and $(W_1,\dots,W_t)$. The $i$'th measurement for $i \in [t]$ operates on $(V_1,\dots,V_i,W_i)$ and the $t+1$'st measurement operates on $(V_1,\dots,V_{t+1}) = [n]$. We will assume that the $\LM$ quantum program has \emph{standard-basis-collapsible $W$ wires} (\cref{def:LM-properties}) so in particular, $\theta_{i,j} = 0$ for all $j \in W_i$ (that is, $W_i$ are standard basis registers for the $i$'th measurement).
    \item $k = (S,\Delta,x,z)$ is a key for the authentication scheme.
    \item See \cref{subsec:auth-def} for the description of our partial $ZX$ measurement notation $M_{\theta,f}$ and $\widetilde{M}_{\theta,f,k,L}$. In particular, \[\left\{M_{\theta_i,f_i^{x,r_1,\dots,r_{i-1}}}\right\}_{i \in [t]},M_{\theta_{t+1},g^{x,r_1,\dots,r_t}}\] are the sequence of measurements applied by the $\LM$ quantum program, and \[\left\{\widetilde{M}_{\theta_i,f_i^{x,r_1,\dots,r_{i-1}},k,L_i\dots L_1}\right\}_{i \in [t]},\widetilde{M}_{\theta_{t+1},g_i^{x,r_1,\dots,r_t},k,L_{t+1}\dots L_1}\] are the corresponding sequence of measurements applied to qubits authenticated using the key $k$. 
    \item We will often refer to a partial set of measurement results $\{v^*_\iota,r^*_\iota\}_{\iota \in [\tau]}$ for some $\tau \in [t+1]$. Note that in the case $\tau = t+1$, the value $r^*_{t+1}$ will always be empty, since the final measurement of an $\LM$ quantum program only outputs $v_{t+1}$. We only include this empty value so that the notation is consistent across each measurement layer. 
    \item Fix any input $x^*$ and partial set of measurement results $\{v^*_\iota,r^*_\iota\}_{\iota \in [\tau]}$ for some $\tau \in [t+1]$. First, we explicitly define the projectors that constitute the $M$ measurements:
    \[M_{\theta_\tau,f_\tau^{x^*,r^*_1,\dots,r^*_{\tau-1}}} \coloneqq \left\{\Pi^{x^*,r^*_1,\dots,r^*_{\tau-1}}[v_\tau,r_\tau]\right\}_{v_\tau,r_\tau},\]
    
    \[ M_{\theta_{t+1},g^{x^*,r^*_1,\dots,r^*_t}} \coloneqq \left\{\Pi^{x^*,r^*_1,\dots,r^*_t}[v_{t+1}]\right\}_{v_{t+1}}.\]

    Next, we define a (sub-normalized) initial state for each set of partial measurement results $\{v^*_\iota,r^*_\iota\}_{\iota \in [\tau]}$: \[\ket*{\psi[x^*,\{v^*_\iota,r^*_\iota\}_{\iota \in [\tau]}]} \coloneqq L_1^\dagger\dots L_\tau^\dagger\Pi^{x^*,r^*_1,\dots,r^*_{\tau-1}}[v^*_\tau,r^*_\tau]L_\tau\dots\Pi^{x^*}[v^*_1,r^*_1]L_1\ket{\psi}.\] 
     
    \item During the proof, we will make use of the collapsible $W$ wires property, and consider measuring (some subset of) the $W$ wires in the standard basis at the beginning of the computation. For any string of measurement results $w^*_i \in \{0,1\}^{|W_i|}$, define corresponding sets 

    \[P[w^*_i] \coloneqq \left\{\widetilde{w}_i : \Dec_{k,\emptyset,\theta_i[W_i]}(\widetilde{w}_i) = w^*_i\right\}, ~~ P[\neg w^*_i] \coloneqq \left\{\widetilde{w}_i : \Dec_{k,\emptyset,\theta_i[W_i]}(\widetilde{w}_i) \notin \{w^*_i,\bot\}\right\},\] where $\emptyset$ indicates an empty sequence of CNOT gates, in other words, the identity. Also recall the notation $\theta[V]$ defined in \cref{subsec:auth-def} indicating a string that is equal to $\theta$ on the subset of indices $V$ and $\bot$ everywhere else. Next, define the projectors 

    \[\Pi[w^*] \coloneqq \Pi[P[w^*_i]], ~~ \Pi[\neg w^*] \coloneqq \Pi[P[\neg w^*_i]].\]

    Finally, for any set $S \subseteq [t]$ and $\{w^*_i\}_{i \in S}$ where each $w^*_i \in \{0,1\}^{|W_i|}$, define 

    \[\Pi[\{w^*_i\}_{i \in S}] \coloneqq \bigotimes_{i \in S}\Pi[w^*_i], ~~ \Pi[\neg \{w^*_i\}_{i \in S}] \coloneqq \sum_{i \in S} \Pi[\neg w_i^*].\]

\end{itemize}

\subsection{Main Theorem}\label{subsec:main-theorem}

\begin{theorem}
For any $\epsilon = \epsilon(\secp) = \negl(\secp) \cdot 2^{-2m(\secp)}$, the scheme described in \proref{fig:qobf} is a quantum state obfuscator that satisfies ideal obfuscation (\cref{def:obf}) for $\epsilon$-pseudo-deterministic families of quantum programs.
\end{theorem}

\begin{remark}
One might hope that the above theorem could be shown for any $\epsilon = \negl(\secp)$, and we leave this open. However, we remark that in the case where the input program has a completely classical description (e.g.\ the case handled by \cite{BKNY23}), one can first repeat the circuit $\poly(\secp)$ times to generically go from $\negl(\secp)$-pseudo-determinism to $\negl(\secp)\cdot 2^{-2m(\secp)}$-pseudo-determinism. Thus, this result captures a strictly more general class of programs than \cite{BKNY23}. Moreover, the application to best-possible copy-protection \cite{CG23} only requires obfuscating \emph{fully} deterministic computation. 
\end{remark}

\begin{proof}
Throughout this proof, we will often drop the dependence of functions and circuit families on the parameter $\secp$ in order to reduce notational clutter. Let $n, m, m'$ be any polynomials (in $\secp$), and suppose that $\ket{\psi},\{L_i\}_{i \in [t+1]},\{\theta_i\}_{i \in [t+1]},\{f_i\}_{i \in [t]},g$ is an $\LM$ quantum program such that $\ket{\psi}$ has at most $n$ qubits, there are at most $n$ gates in the circuit, and the classical input has $m$ bits. Suppose that this $\LM$ quantum program is $\epsilon$-pseudo-deterministic for a small enough $\epsilon$ as specified by the theorem statement, and let $Q$ be the induced map of this $\LM$ quantum program. Consider any QPT\footnote{The only reason that we restrict our adversary to be quantum polynomial-\emph{time} as opposed to quantum polynomial-\emph{query} is the very first step in our proof, where we replace the PRF with a random oracle. If we allow the obfuscation scheme to use a true random oracle (thus sacrificing the efficiency of the oracles), then we obtain security against any QPQ adversary.} adversary $A$ and distinguisher $D$, and define $D[A]$ to be the procedure that runs $A$, feeds its output to $D$, and then runs $D$ to produce a binary-valued outcome. Thus, we can write the "real" obfuscation experiment as 

\[\Pr\left[1 \gets D\left(A\left(\QObf\left(1^\secp,\ket{\psi}\right)\right)\right)\right] = \E_{\left(\ket*{\widetilde{\psi}},O\right)\gets \QObf(1^\secp,\ket{\psi})}\left[\big\|D[A]^{O}\ket*{\widetilde{\psi}}\big\|^2\right].\]

Now, we will consider a sequence of hybrid distributions over (state, oracle) pairs $(\ket*{\widetilde{\psi}},O)$, beginning with the real distribution $\QObf_0 \coloneqq \QObf$, and ending with a fully simulated distribution $\QObf_6$ that no longer needs to take $\ket{\psi}$ as input (and instead uses oracle access to $Q$). Our first step will be to switch the oracle $\sG$ to a simulated oracle $\GSim$ that \emph{verifies} rather than \emph{decodes} the intermediate labels and final authenticated measurement, and uses oracle access to $Q$ to respond in the case that verification passes. Next, we'll "collapse" the oracles $\sF_1,\dots,\sF_t$ as described in \cref{subsec:proof-overview}, using a strategy derived from the collapsible $W$ wires property of the $\LM$ quantum program. Finally, we'll replace the input $\ket{\psi}$ with the all zeros input $\ket{0^n}$.

The description of these distributions follow (but no claims about indistinguishability yet). The difference between adjacent distributions are highlighted in red, and whenever we write $w^*$, we parse it at $w^* = \{w^*_i\}_{i \in [t]}$, where each $w^*_i \in \{0,1\}^{|W_i|}$.\\

\begin{tcolorbox}[breakable, enhanced]\footnotesize
    $\QObf_1(1^\secp,\ket{\psi})$:
    \begin{itemize}
        \item Sample $k \gets \Gen(1^\kappa,n)$, and compute $\ket*{\psi_k} = \Enc_k(\ket{\psi})$.
        \item Sample a signature token $(\vk,\ket{\sk}) \gets \TokGen(1^\kappa)$.
        \item \textcolor{red}{Let $H : \{0,1\}^* \to \{0,1\}^\kappa$ be a random oracle.}
        \item Define $\sF_1,\dots,\sF_t$ as in $\QObf_0$.
        \item Define $\sG$ as in $\QObf_0$.
        \item Output $\ket*{\widetilde{\psi}} = \ket*{\psi_k}\ket{\sk}, O = \left(\sF_1,\dots,\sF_t,\sG\right).$
    \end{itemize}

    \vspace{0.2cm}\hrule\vspace{0.3cm}

    $\QObf_2(1^\secp,\ket{\psi})$:
    \begin{itemize}
        \item Sample $k \gets \Gen(1^\kappa,n)$, and compute $\ket*{\psi_k} = \Enc_k(\ket{\psi})$.
        \item Sample a signature token $(\vk,\ket{\sk}) \gets \TokGen(1^\kappa)$.
        \item Let $H : \{0,1\}^* \to \{0,1\}^\kappa$ be a random oracle.
        \item Define $\sF_1,\dots,\sF_t$ as in $\QObf_0$.
        \item Define the function $\textcolor{red}{\GSim}\left(x,\sigma_x,\widetilde{v}_1,\dots,\widetilde{v}_{t+1},\ell_1,\dots,\ell_t\right):$
        \begin{itemize}
            \item Output $\bot$ if $\TokVer(\vk,x,\sigma_x) = \bot$.
            \item For each $\iota \in [t]$, let \[\ell_{\iota,0} = H(x,\sigma_x,\widetilde{v}_1,\dots,\widetilde{v}_\iota,\ell_1,\dots,\ell_{\iota-1},0), ~~ \ell_{\iota,1} = H(x,\sigma_x,\widetilde{v}_1,\dots,\widetilde{v}_\iota,\ell_1,\dots,\ell_{\iota-1},1),\] and output $\bot$ if $\ell_{\iota,0} = \ell_{\iota,1}$ or $\ell_\iota \notin \{\ell_{\iota,0},\ell_{\iota,1}\}$. 
            \item \textcolor{red}{Output $\bot$ if $\Ver_{k,L_{t+1} \dots L_1,\theta_{t+1}}(\widetilde{v}_1,\dots,\widetilde{v}_{t+1}) = \bot$.} 
            \item Output \textcolor{red}{$Q(x)$}.
        \end{itemize}
        \item Output $\ket*{\widetilde{\psi}} = \ket*{\psi_k}\ket{\sk}, O = \left(\sF_1,\dots,\sF_t,\textcolor{red}{\GSim}\right).$
    \end{itemize}

    \vspace{0.2cm}\hrule\vspace{0.3cm}

    $\QObf_3(1^\secp,\ket{\psi})$:
    \begin{itemize}
        \item Sample $k \gets \Gen(1^\kappa,n)$, and compute $\ket*{\psi_k} = \Enc_k(\ket{\psi})$.
        \item Sample a signature token $(\vk,\ket{\sk}) \gets \TokGen(1^\kappa)$.
        \item Let $H : \{0,1\}^* \to \{0,1\}^\kappa$ be a random oracle.
        \item For each $i \in [t]$, define the function $\sF_i\textcolor{red}{[w^*]}\left(x,\sigma_x,\widetilde{v}_1,\dots,\widetilde{v}_i,\widetilde{w}_i,\ell_1,\dots,\ell_{i-1}\right)$:
        \begin{itemize}
            \item Output $\bot$ if $\TokVer(\vk,x,\sigma_x) = \bot$.
            \item For each $\iota \in [i-1]$, let \[\ell_{\iota,0} = H(x,\sigma_x,\widetilde{v}_1,\dots,\widetilde{v}_{\iota},\ell_1,\dots,\ell_{\iota-1},0),  ~~ \ell_{\iota,1} = H(x,\sigma_x,\widetilde{v}_1,\dots,\widetilde{v}_{\iota},\ell_1,\dots,\ell_{\iota-1},1),\]  and output $\bot$ if $\ell_{\iota,0} = \ell_{\iota,1}$ or $\ell_\iota \notin \{\ell_{\iota,0},\ell_{\iota,1}\}$. Otherwise, let $r_\iota$ be such that $\ell_\iota = \ell_{\iota,r_\iota}$.
            \item Compute $(v_1,\dots,v_i,\cdot) = \Dec_{k,L_i \dots L_1,\theta_i}(\widetilde{v}_1,\dots,\widetilde{v}_i,\widetilde{w}_i)$ and output $\bot$ if the result is $\bot$.
            \item Compute $(\cdot,r_i) = f_i^{x,r_1,\dots,r_{i-1}}(v_1,\dots,v_i,\textcolor{red}{w^*_i})$.
            \item Set $\ell_i \coloneqq H(x,\sigma_x,\widetilde{v}_1,\dots,\widetilde{v}_i,\ell_1,\dots,\ell_{i-1},r_i),$ and output $(\widetilde{v}_i,\ell_i)$.
       
        \end{itemize}
        \item Define $\GSim$ as in $\QObf_2$.
        \item Output $\ket*{\widetilde{\psi}} = \ket*{\psi_k}\ket{\sk}, O = \left(\sF_1\textcolor{red}{[\cdot]},\dots,\sF_t\textcolor{red}{[\cdot]},\GSim\right).$
    \end{itemize}

    \vspace{0.2cm}\hrule\vspace{0.3cm}

    $\QObf_4(1^\secp,\ket{\psi})$:
    \begin{itemize}
        \item Sample $k \gets \Gen(1^\kappa,n)$, and compute $\ket*{\psi_k} = \Enc_k(\ket{\psi})$.
        \item Sample a signature token $(\vk,\ket{\sk}) \gets \TokGen(1^\kappa)$.
        \item Let $H : \{0,1\}^* \to \{0,1\}^\kappa$ be a random oracle.
        \item For each $i \in [t]$, define the function $\sF_i[w^*]\left(x,\sigma_x,\widetilde{v}_1,\dots,\widetilde{v}_i,\widetilde{w}_i,\ell_1,\dots,\ell_{i-1}\right)$:
        \begin{itemize}
            \item Output $\bot$ if $\TokVer(\vk,x,\sigma_x) = \bot$.
            \item For each $\iota \in [i-1]$, let \[\ell_{\iota,0} = H(x,\sigma_x,\widetilde{v}_1,\dots,\widetilde{v}_{\iota},\ell_1,\dots,\ell_{\iota-1},0),  ~~ \ell_{\iota,1} = H(x,\sigma_x,\widetilde{v}_1,\dots,\widetilde{v}_{\iota},\ell_1,\dots,\ell_{\iota-1},1),\]  and output $\bot$ if $\ell_{\iota,0} = \ell_{\iota,1}$ or $\ell_\iota \notin \{\ell_{\iota,0},\ell_{\iota,1}\}$. \st{Otherwise, let $r_\iota$ be such that $\ell_\iota = \ell_{\iota,r_\iota}$.}
            \item Compute $(v_1,\dots,v_i,\cdot) = \Dec_{k,L_i \dots L_1,\theta_i}(\widetilde{v}_1,\dots,\widetilde{v}_i,\widetilde{w}_i)$ and output $\bot$ if the result is $\bot$.
            \item \textcolor{red}{For $\iota \in [i]$, compute $(\cdot,r_\iota) = f_\iota^{x,r_1,\dots,r_{\iota-1}}(v_1,\dots,v_\iota,w^*_\iota)$.}
            \item Set $\ell_i \coloneqq H(x,\sigma_x,\widetilde{v}_1,\dots,\widetilde{v}_i,\ell_1,\dots,\ell_{i-1},r_i),$ and output $(\widetilde{v}_i,\ell_i)$.
       
        \end{itemize}
        \item Define $\GSim$ as in $\QObf_2$.
        \item Output $\ket*{\widetilde{\psi}} = \ket*{\psi_k}\ket{\sk}, O = \left(\sF_1[\cdot],\dots,\sF_t[\cdot],\GSim\right).$
    \end{itemize}

    \vspace{0.2cm}\hrule\vspace{0.3cm}
    
    $\QObf_5(1^\secp,\ket{\psi})$:
    \begin{itemize}
        \item Sample $k \gets \Gen(1^\kappa,n)$, and compute $\ket*{\psi_k} = \Enc_k(\ket{\psi})$.
        \item Sample a signature token $(\vk,\ket{\sk}) \gets \TokGen(1^\kappa)$.
        \item Let $H : \{0,1\}^* \to \{0,1\}^\kappa$ be a random oracle.
        \item For each $i \in [t]$, define the function $\textcolor{red}{\FSim_i}\left(x,\sigma_x,\widetilde{v}_1,\dots,\widetilde{v}_i,\widetilde{w}_i,\ell_1,\dots,\ell_{i-1}\right)$:
        \begin{itemize}
            \item Output $\bot$ if $\TokVer(\vk,x,\sigma_x) = \bot$.
            \item For each $\iota \in [i-1]$, let \[\ell_{\iota,0} = H(x,\sigma_x,\widetilde{v}_1,\dots,\widetilde{v}_\iota,\ell_1,\dots,\ell_{\iota-1},0), ~~ \ell_{\iota,1} = H(x,\sigma_x,\widetilde{v}_1,\dots,\widetilde{v}_\iota,\ell_1,\dots,\ell_{\iota-1},1),\] and output $\bot$ if $\ell_{\iota,0} = \ell_{\iota,1}$ or $\ell_\iota \notin \{\ell_{\iota,0},\ell_{\iota,1}\}$. 
            \item \textcolor{red}{Output $\bot$ if $\Ver_{k,L_i \dots L_1,\theta_i}(\widetilde{v}_1,\dots,\widetilde{v}_i,\widetilde{w}_i) = \bot$.}
            \item Set $\ell_i \coloneqq H(x,\sigma_x,\widetilde{v}_1,\dots,\widetilde{v}_i,\ell_1,\dots,\ell_{i-1},\textcolor{red}{0}),$ and output $(\widetilde{v}_i,\ell_i)$.
        \end{itemize}
        \item Define $\GSim$ as in $\QObf_2$.
        \item Output $\ket*{\widetilde{\psi}} = \ket*{\psi_k}\ket{\sk}, O = \left(\textcolor{red}{\FSim_1},\dots,\textcolor{red}{\FSim_t},\GSim\right).$
    \end{itemize}

    \vspace{0.2cm}\hrule\vspace{0.3cm}
    
    $\QObf_6(1^\secp)$:
    \begin{itemize}
        \item Sample $k \gets \Gen(1^\kappa,n)$, and compute $\ket*{\psi_k} = \Enc_k(\textcolor{red}{\ket{0^n}})$.
        \item Sample a signature token $(\vk,\ket{\sk}) \gets \TokGen(1^\kappa)$.
        \item Let $H : \{0,1\}^* \to \{0,1\}^\kappa$ be a random oracle.
        \item Define $\FSim_1,\dots,\FSim_t$ as in $\QObf_4$.
        \item Define $\GSim$ as in $\QObf_2$.
        \item Output $\ket*{\widetilde{\psi}} = \ket*{\psi_k}\ket{\sk}, O = \left(\FSim_1,\dots,\FSim_t,\GSim\right).$
    \end{itemize}
\end{tcolorbox}

Later, we will use these distributions to define a sequence of hybrids starting with the real obfuscation experiment and ending with the simulated obfuscation experiment. But first, we will establish several claims about these distributions that will be useful while arguing indistinguishability of the hybrids. 

This first claim establishes that, once the oracles are simulated, no adversary can map a state whose $W$ wires have been collapsed to outcome $w^*$ onto the support of a state with different outcomes $w \neq w^*$. At the end of this sequence of claims, we will have established that this property holds \emph{even in $\QObf_2$}, where the $\sF_i$ oracles are not yet simulated.

\begin{claim}\label{claim:prelim-wmapping}
    For any QPQ unitary $U$ and any $w^* = \{w^*_i\}_{i \in [t]}$, it holds that 
    \[\E\left[\big\| \Pi[\neg w^*]U^O\Pi[w^*]\ket*{\widetilde{\psi}} \big\|^2 : \ket*{\widetilde{\psi}},O \gets \QObf_5(1^\secp,\ket{\psi})\right]= 2^{-\Omega(\kappa).}\] 
\end{claim}

\begin{proof}
    The key point is that in $\QObf_5$, none of the oracles $\FSim_1,\dots,\FSim_t,\GSim$ require access to the the decryption oracle $\Dec_{k,\cdot,\cdot}(\cdot)$ for the authentication scheme. Rather, they can be implemented just given access to the verification oracle $\Ver_{k,\cdot,\cdot}(\cdot)$. Now, since the $W$ wires are authenticated, the hardness of mapping from the support of $w^*$ to $w$ for any $w \neq w^*$ follows directly from the security of the authentication scheme (\cref{thm:auth-sec}), and in particular that it satisfies \cref{def:auth-mapping-sec} (mapping security).
\end{proof}

In the proof of the next two claims, we will use the following notation. Fix a key $k$ and $w^* = \{w^*_i\}_{i \in [t]}$, and define the following function.\\

\noindent $R[k,w^*] : (x,\widetilde{v}_1,\dots,\widetilde{v}_i) \to r_i$
\begin{itemize}
    \item Compute $(v_1,\dots,v_i) = \Dec_{k,L_i\dots L_1,\theta_i[V_1,\dots,V_i]}(\widetilde{v}_1,\dots,\widetilde{v}_i)$ and output $\bot$ if the result is $\bot$.
    \item For $\iota \in [i],$ compute $(\cdot,r_\iota) = f_\iota^{x,r_1,\dots,r_{\iota-1}}(v_1,\dots,v_\iota,w^*_\iota)$.
    \item Output $r_i$.
\end{itemize}

This function determines the bit $r_i$ when the $w^*$ outcomes have been hard-coded into the oracles, and we will use it when showing indistinguishability between $\QObf_3,\QObf_4$, and $\QObf_5$ in the case where the $W$ wires of the input state have been collapsed to outcome $w^*$.

\begin{claim}\label{claim:prelim4-5}
    For any (unbounded) distinguisher $D$ and any $w^* = \{w^*_i\}_{i \in [t]}$, it holds that 
    \begin{align*}&\E\left[\big\|D^{\sF_1[w^*],\dots,\GSim}\left(k,\Pi[w^*]\ket*{\widetilde{\psi}}\right)\big\|^2 : \ket*{\widetilde{\psi}},(\sF_1[\cdot],\dots,\sF_t[\cdot],\GSim) \gets \QObf_4(1^\secp,\ket{\psi})\right]\\ &= \E\left[\big\|D^{\FSim_1,\dots,\GSim}\left(k,\Pi[w^*]\ket*{\widetilde{\psi}}\right)\big\|^2 : \ket*{\psi},(\FSim_1,\dots,\FSim_t,\GSim) \gets \QObf_5(1^\secp,\ket{\psi})\right],\end{align*}
    where $D$'s input includes the key $k$ sampled by $\QObf_4,\QObf_5$.
\end{claim}

\begin{proof}

    In $\QObf_4$, we have that 

    \[\sF_i[w^*](x,\sigma_x,\widetilde{v}_1,\dots,\widetilde{v}_i,\widetilde{w}_i,\ell_1,\dots,\ell_{i-1}) \in \{H(x,\sigma_x,\widetilde{v}_1,\dots,\widetilde{v}_i,\ell_1,\dots,\ell_{i-1},R[k,w^*](\widetilde{v}_1,\dots,\widetilde{v}_i)),\bot\},\] while in $\QObf_5$, we have that 

    \[\FSim_i[w^*](x,\sigma_x,\widetilde{v}_1,\dots,\widetilde{v}_i,\widetilde{w}_i,\ell_1,\dots,\ell_{i-1}) \in \{H(x,\sigma_x,\widetilde{v}_1,\dots,\widetilde{v}_i,\ell_1,\dots,\ell_{i-1},0),\bot\}.\]

    Both implementations of the oracles will always output $\bot$ on the same set of inputs, since this is true of $\Dec_{k,\cdot,\cdot}(\cdot)$ and $\Ver_{k,\cdot,\cdot}(\cdot)$ by definition. Finally, their non-$\bot$ answers are identically distributed over the  randomness of the random oracle $H$, since each $(x,\sigma_x,\widetilde{v}_1,\dots,\widetilde{v}_i,\widetilde{w}_i,\ell_1,\dots,\ell_{i-1})$ fixes a single choice of bit $R[k,w^*](\widetilde{v}_1,\dots,\widetilde{v}_i) \in \{0,1\}$, and for any $(x,\sigma_x,\widetilde{v}_1,\dots,\widetilde{v}_i,\widetilde{w}_i,\allowbreak\ell_1,\dots,\allowbreak\ell_{i-1})$, \[H(x,\sigma_x,\widetilde{v}_1,\dots,\widetilde{v}_i,\widetilde{w}_i,\ell_1,\dots,\ell_{i-1},0) \ \ \ \text{and} \ \ \  H(x,\sigma_x,\widetilde{v}_1,\dots,\widetilde{v}_i,\widetilde{w}_i,\ell_1,\dots,\ell_{i-1},1)\] are identically distributed (each is a uniformly random string).
\end{proof}

\begin{claim}\label{claim:prelim3-4}
    For any QPQ distinguisher $D$ and any $w^* = \{w^*_i\}_{i \in [t]}$, it holds that 
    \begin{align*}&\bigg|\E\left[\big\|D^{\sF_1[w^*],\dots,\GSim}\left(k,\Pi[w^*]\ket*{\widetilde{\psi}}\right)\big\|^2 : \ket*{\widetilde{\psi}},(\sF_1[\cdot],\dots,\sF_t[\cdot],\GSim) \gets \QObf_3(1^\secp,\ket{\psi})\right]\\ &- \E\left[\big\|D^{\sF_1[w^*],\dots,\GSim}\left(k,\Pi[w^*]\ket*{\widetilde{\psi}}\right)\big\|^2 : \ket*{\widetilde{\psi}},(\sF_1[\cdot],\dots,\sF_t[\cdot],\GSim) \gets \QObf_4(1^\secp,\ket{\psi})\right]\bigg| = 2^{-\Omega(\kappa)},\end{align*}
    where $D$'s input includes the key $k$ sampled by $\QObf_3,\QObf_4$.
\end{claim}

\begin{proof}
    Observe that the oracles $\sF_i[w^*]$ in these experiments are identical except for on inputs \[(x,\sigma_x,\widetilde{v}_1,\dots,\widetilde{v}_i,\widetilde{w}_i,\ell_1,\dots,\ell_{i-1})\] such that there exists an $\iota \in [i-1]$ with \[\ell_{\iota} = H(x,\sigma_x,\widetilde{v}_1,\dots,\widetilde{v}_\iota,\ell_1,\dots,\ell_{\iota-1},1-R[k,w^*](\widetilde{v}_1,\dots,\widetilde{v}_\iota)).\] However, the oracles $\sF_i[w^*]$ in $\QObf_4$ are defined to never output such an $\ell_\iota$, and thus such an input can only be guessed with probability $2^{-\kappa}$ over the randomness of $H$. The claim follows by applying \cref{lemma:puncture} (a standard oracle hybrid argument).
\end{proof}

Next, we combine what we have shown so far - the hardness of mapping between $\Pi[w^*]$ and $\Pi[\neg w^*]$ in $\QObf_5$ and the indistinguishability of $\QObf_3$ and $\QObf_5$ - to show the indistinguishability of $\QObf_2$ and $\QObf_3$ (in the case where the $W$ wires are collapsed to some $w^*$).

\begin{claim}\label{claim:prelim2-3}
     For any QPQ distinguisher $D$ and any $w^* = \{w^*_i\}_{i \in [t]}$, it holds that 
    \begin{align*}&\bigg|\E\left[\big\|D^{\sF_1,\dots,\GSim}\left(k,\Pi[w^*]\ket*{\widetilde{\psi}}\right)\big\|^2 : \ket*{\widetilde{\psi}},(\sF_1,\dots,\sF_t,\GSim) \gets \QObf_2(1^\secp,\ket{\psi})\right]\\ &- \E\left[\big\|D^{\sF_1[w^*],\dots,\GSim}\left(k,\Pi[w^*]\ket*{\widetilde{\psi}}\right)\big\|^2 : \ket*{\widetilde{\psi}},(\sF_1[\cdot],\dots,\sF_t[\cdot],\GSim) \gets \QObf_3(1^\secp,\ket{\psi})\right]\bigg| = 2^{-\Omega(\kappa)},\end{align*}
    where $D$'s input includes the key $k$ sampled by $\QObf_2,\QObf_3$.
\end{claim}

\begin{proof}
    Observe that the oracles $\sF_i, \sF_i[w^*]$ in these experiments are identical except for on inputs \[(x,\sigma_x,\widetilde{v}_1,\dots,\widetilde{v}_i,\widetilde{w}_i,\ell_1,\dots,\ell_{i-1})\]such that $\widetilde{w}_i \in P[\neg w^*_i]$. By combining the previous three claims, we see that no QPQ adversary can find such a $\widetilde{w}_i$ in $\QObf_3$ except with probability $2^{-\Omega(\kappa)}$. The claim follows by applying \cref{lemma:puncture} (a standard oracle hybrid argument).
\end{proof}

Next, we state a direct corollary of these four claims, which is the hardness of mapping between $\Pi[w^*]$ and $\Pi[\neg w^*]$ even in $\QObf_2$.

\begin{corollary}\label{cor:wmapping}
    For any QPQ unitary $U$, and any $w^* = \{w^*_i\}_{i \in [t]}$, it holds that
    \[\E\left[\big\| \Pi[\neg w^*_i]U^O\Pi[w^*]\ket*{\widetilde{\psi}}\big\|^2 : \ket*{\widetilde{\psi}},O \gets \QObf_2(1^\secp,\ket{\psi})\right]= 2^{-\Omega(\kappa).}\]
\end{corollary}

Finally, we consider a sequence of hybrids beginning with the real obfuscation experiment as described at the beginning of the proof, and ending with the simulated experiment using a simulator that we define below.

\begin{itemize}
        \item The real experiment:
        \[\cH_0 = \E\left[\big\|D[A]^{\sF_1,\dots,\sG}\ket*{\widetilde{\psi}}\big\|^2 : \ket*{\widetilde{\psi}},(\sF_1,\dots,\sF_t,\sG) \gets \QObf_0(1^\secp,\ket{\psi})\right]\] 
        \item Replace PRF with random oracle:
        \[\cH_1 = \E\left[\big\|D[A]^{\sF_1,\dots,\sG}\ket*{\widetilde{\psi}}\big\|^2 : \ket*{\widetilde{\psi}},(\sF_1,\dots,\sF_t,\sG) \gets \QObf_1(1^\secp,\ket{\psi})\right]\] 
        \item Simulate the $\sG$ oracle:
        \[\cH_2 = \E\left[\big\|D[A]^{\sF_1,\dots,\GSim}\ket*{\widetilde{\psi}}\big\|^2 : \ket*{\widetilde{\psi}},(\sF_1,\dots,\sF_t,\GSim) \gets \QObf_2(1^\secp,\ket{\psi})\right]\] 
        \item Split $\ket*{\widetilde{\psi}}$ into orthogonal components:
        \[\cH_3=\E\left[\sum_{w^*}\big\|D[A]^{\sF_1,\dots,\GSim} \Pi[w^*]\ket*{\widetilde{\psi}}\big\|^2 : \ket*{\widetilde{\psi}},(\sF_1,\dots,\sF_t,\GSim) \gets \QObf_2(1^\secp,\ket{\psi})\right].\]
        \item Hard-code the $w^*$ measurement results:
        \[\cH_4=\E\left[\sum_{w^*}\big\|D[A]^{\sF_1[w^*],\dots,\GSim}\Pi[w^*]\ket*{\widetilde{\psi}}\big\|^2 : \ket*{\widetilde{\psi}},(\sF_1[\cdot],\dots,\sF_t[\cdot],\GSim) \gets \QObf_3(1^\secp,\ket{\psi})\right].\]

        \item Simulate the $\sF_1,\dots,\sF_t$ oracles:

        \[\cH_5 = \E\left[\sum_{w^*}\big\|D[A]^{\FSim_1,\dots,\GSim}\Pi[w^*]\ket*{\widetilde{\psi}}\big\|^2 : \ket*{\widetilde{\psi}},(\FSim_1,\dots,\FSim_t,\GSim) \gets \QObf_5(1^\secp,\ket{\psi})\right].\] 
        
        \item Put $\ket*{\widetilde{\psi}}$ back together: \[\cH_6=\E\left[\big\|D[A]^{\FSim_1,\dots,\GSim}\ket*{\widetilde{\psi}}\big\|^2 : \ket*{\widetilde{\psi}},(\FSim_1,\dots,\FSim_t,\GSim) \gets \QObf_5(1^\secp,\ket{\psi})\right].\]
        \item Simulate the state: \[\cH_7=\E\left[\big\|D[A]^{\FSim_1,\dots,\GSim}\ket*{\widetilde{\psi}}\big\|^2 : \ket*{\widetilde{\psi}},(\FSim_1,\dots,\FSim_t,\GSim) \gets \QObf_6(1^\secp)\right].\] 
    \end{itemize}

\noindent Now, we are ready to define the simulator $\Sim^Q(1^\secp,n,m,m')$ for our obfuscation scheme:
\begin{itemize}
    \item Sample $\ket*{\widetilde{\psi}},(\FSim_1,\dots,\FSim_t,\GSim) \gets \QObf_6(1^\secp)$.
    \item Output $A^{\FSim_1,\dots,\FSim_t,\GSim}\ket*{\widetilde{\psi}}$.
\end{itemize}

\noindent Thus, the simulated experiment is exactly

\begin{align*}
&\Pr\left[1 \gets D\left(\Sim^Q\left(1^\secp,n,m,m'\right)\right)\right] = \cH_7.
\end{align*}

\noindent The following set of claims then completes the proof.

\begin{claim}\label{claim:0-1}
    $|\cH_0-\cH_1| = \negl(\secp)$.
\end{claim}

\begin{proof}
    This follows directly from the security of the PRF against quantum superposition-query attacks.
\end{proof}

\begin{claim}\label{claim:1-2}
    $|\cH_1-\cH_2| = \negl(\secp)$.
\end{claim}

\begin{proof}
    Suppose that we sample $\ket*{\widetilde{\psi}},(\sF_1,\dots,\sF_t,\sG) \gets \QObf_1(1^\secp,\ket{\psi})$, and then define the set
 
    \[B = \left\{(x,\sigma_x,\widetilde{v}_1,\dots,\widetilde{v}_{t+1},\ell_1,\dots,\ell_t) : \sG(x,\sigma_x,\widetilde{v}_1,\dots,\widetilde{v}_{t+1},\ell_1,\dots,\ell_t) \notin \{Q(x),\bot\}\right\}.\] Observe that the only difference between $\QObf_1$ and $\QObf_2$ is the definition of the oracles $\sG,\GSim$, and that these oracles are identical outside of the set $B$. Suppose for contradiction that the claim is false. Then by \cref{lemma:puncture} (which is a standard oracle hybrid argument), there must exist an adversary that can find an input on which $\sG$ and $\GSim$ differ with non-negligible probability. That is, there exists a QPQ unitary $U$ such that

    \[\E\left[\big\| \Pi[B]U^{\sF_1,\dots,\GSim}\ket*{\widetilde{\psi}}\big\|^2 : \ket*{\widetilde{\psi}},(\sF_1,\dots,\sF_t,\GSim) \gets \QObf_2(1^\secp,\ket{\psi})\right] = \delta(\secp)\] for some $\delta(\secp) = \nonnegl(\secp)$. Now, for each input $x$,
    define \[B[x] \coloneqq \{(x,\cdot) : (x,\cdot) \in B\}.\] Then by a union bound, there must exist \emph{some} $x^* \in \{0,1\}^m$ such that \[\E\left[\big\| \Pi\left[B[x^*]\right]U^{\sF_1,\dots,\sF_t,\GSim}\ket*{\widetilde{\psi}}\big\|^2 : \ket*{\widetilde{\psi}},(\sF_1,\dots,\sF_t,\GSim) \gets \QObf_2(1^\secp,\ket{\psi})\right] \geq \delta(\secp) \cdot 2^{-m}.\]

    Next, define $\mathsf{Coh}\text{-}\LMEval[x^*]$ to be the unitary that coherently applies the evaluation procedure $\LMEval(x^*,\cdot,\{L_i\}_{i \in [t+1]},\{\theta_i\}_{i \in [t+1]},\{f_i\}_{i \in [t]},g)$ for the $\LM$ quantum program that we are obfuscating, and define

    \[\ket{\psi_{x^*}'} \coloneqq \mathsf{Coh}\text{-}\LMEval[x^*]^\dagger \ketbra{Q(x^*)}{Q(x^*)} \mathsf{Coh}\text{-}\LMEval[x^*]\ket{\psi}, ~~ \ket*{\psi_{x^*}} \coloneqq \frac{\ket*{\psi_{x^*}'}}{\| \ket{\psi_{x^*}'}\|}.\] That is, $\ket*{\psi_{x^*}}$ is the result of running the $\LM$ quantum program coherently on input $x^*$ and post-selecting on obtaining the ``correct'' output $Q(x^*)$. By the $\epsilon$-pseudo-determinism of $Q$ and Gentle Measurement (\cref{lemma:gentle-measurement}), there is some $\delta'(\secp) = \nonnegl(\secp)$ such that 

    \begin{align*}
    &\E\left[\big\| \Pi[B[x^*]]U^{\sF_1,\dots,\GSim}\ket*{\widetilde{\psi}_{x^*}}\big\|^2 : \ket*{\widetilde{\psi}_{x^*}},(\sF_1,\dots,\sF_t,\GSim) \gets \QObf_2(1^\secp,\ket{\psi_{x^*}})\right]\\ &\geq \delta'(\secp) \cdot 2^{-m} \geq \delta'(\secp) \cdot 2^{-n}.
    \end{align*} However, this violates \cref{lemma:induction-1} with $\tau = 0$, which is proven in the following section. Indeed, plugging in $\kappa = n^4$, the lemma states that, for some constant $c$,
    \begin{align*}
    &\E\left[\big\| \Pi[B[x^*]]U^{\sF_1,\dots,\GSim}\ket*{\widetilde{\psi}_{x^*}}\big\|^2 : \ket*{\widetilde{\psi}_{x^*}},(\sF_1,\dots,\sF_t,\GSim) \gets \QObf_2(1^\secp,\ket{\psi_{x^*}})\right]\\ &= 2^{3n(t+1)-c n^4} = 2^{-\Omega(n^2)} < \delta'(\secp) \cdot 2^{-n},
    \end{align*}
    which gives us the contradiction.
\end{proof}

\begin{claim}
    $|\cH_2-\cH_3| \leq 2^{-\Omega(\kappa)}$.
\end{claim}

\begin{proof}

    \begin{align*}
        &|\cH_2 - \cH_3| \\
        &\leq 2^n \cdot \left[\sum_{w^* \neq {w'}^*}\E\left[\big\| \Pi[{w'}^*]D_1^{\sF_1,\dots,\GSim}\Pi[{w}^*]\ket*{\widetilde{\psi}}\big\|^2 : \ket*{\widetilde{\psi}},(\sF_1,\dots,\sF_t,\GSim) \gets \QObf_2(1^\secp,\ket{\psi})\right]\right]^{1/2} \\
        &\leq 2^n \cdot \left[2^n \cdot \sum_{w^*}\E\left[\big\| \Pi[\neg w^*]D_1^{\sF_1,\dots,\GSim}\Pi[w^*]\ket*{\widetilde{\psi}}\big\|^2 : \ket*{\widetilde{\psi}},(\sF_1,\dots,\sF_t,\GSim) \gets \QObf_2(1^\secp,\ket{\psi})\right]\right]^{1/2} \\
        &\leq 2^{2n} \cdot 2^{-\Omega(\kappa)} = 2^{-\Omega(\kappa)},
    \end{align*}

    where the first inequality follows from \cref{lemma:map-distinguish} (which is an application of Cauchy-Schwarz) and the third follows from \cref{cor:wmapping} proven above.
    
\end{proof}

\begin{claim}
    $|\cH_3-\cH_4|=2^{-\Omega(\kappa)}$.
\end{claim}

\begin{proof}
    This follows from \cref{cor:wmapping} proven above and \cref{lemma:puncture} (which is a standard oracle hybrid argument).
\end{proof}

\begin{claim}
    $|\cH_4-\cH_5|=2^{-\Omega(\kappa)}$.
\end{claim}

\begin{proof}
    This follows by combining \cref{claim:prelim3-4} and \cref{claim:prelim4-5} proven above.
\end{proof}

\begin{claim}
    $|\cH_5-\cH_6| \leq 2^{-\Omega(\kappa)}$.
\end{claim}

\begin{proof}

    \begin{align*}
        &|\cH_5 - \cH_6| \\
        &\leq 2^n \cdot \left[\sum_{w^* \neq {w'}^*}\E\left[\big\| \Pi[{w'}^*]D_1^{\FSim_1,\dots,\GSim}\Pi[w^*]\ket*{\widetilde{\psi}}\big\|^2 : \ket*{\widetilde{\psi}},\sF_1,\dots,\sF_t,\GSim \gets \QObf_5(1^\secp,\ket{\psi})\right]\right]^{1/2} \\
        &\leq 2^n \cdot \left[2^n \cdot \sum_{w^*}\E\left[\big\| \Pi[\neg w^*]D_1^{\FSim_1,\dots,\GSim}\Pi[w^*]\ket*{\widetilde{\psi}}\big\|^2 : \ket*{\widetilde{\psi}},\sF_1,\dots,\sF_t,\GSim \gets \QObf_5(1^\secp,\ket{\psi})\right]\right]^{1/2} \\
        &\leq 2^{2n} \cdot 2^{-\Omega(\kappa)} = 2^{-\Omega(\kappa)},
    \end{align*}

    where the first inequality follows from \cref{lemma:map-distinguish} (which is an application of Cauchy-Schwarz) and the third follows from \cref{claim:prelim-wmapping} proven above.

\end{proof}

\begin{claim}
    $|\cH_6-\cH_7| \leq 2^{-\Omega(\kappa)}$.
\end{claim}

\begin{proof}
    This follows from privacy of the authentication scheme (\cref{thm:auth-privacy}), since the oracles in $\cH_6$ can be implemented with oracle access to $\Ver_{k,\cdot,\cdot}(\cdot)$ rather than $\Dec_{k,\cdot,\cdot}(\cdot)$.
\end{proof}

\end{proof}

\subsection{Inductive Argument}\label{subsec:induction}

In this section, we give an inductive proof of \cref{lemma:induction-1}, which was required by \cref{claim:1-2} above. First, we describe a variant of $\QObf$ that we call $\ParMeas$, which supports hard-coding an input $x^*$ and the first $\tau$ partial measurement results $\{v^*_\iota,r^*_\iota\}_{\iota \in [\tau]}$ into the oracles $\sF_1,\dots,\sF_t$. When these results are hard-coded, the inputs $\widetilde{v}_1,\dots,\widetilde{v}_\tau,\ell_1,\dots,\ell_\tau$ are merely verified rather than decoded by the oracles $\sF_1,\dots,\sF_t$, and the hard-coded results $\{v^*_\iota,r^*_\iota\}_{\iota \in [\tau]}$ are used in place of the decoded results.

We define the distribution to include two additional outputs:
\begin{itemize}
    \item The set $B[x^*,\{v^*_\iota,r^*_\iota\}_{\iota \in [\tau]}]$ contains the "bad" set of inputs that verify properly but decode to an incorrect output $Q(x^*)$, when using the hard-coded values $\{v^*_\iota,r^*_\iota\}_{\iota \in [\tau]}$.
    \item The set $C[x^*,\{v^*_\iota,r^*_\iota\}_{\iota \in [\tau]}]$ contains the set of inputs on which the oracles would differ had the latest measurement $(v^*_\tau,r^*_\tau)$ not been hard-coded.
\end{itemize}

\begin{tcolorbox}[breakable, enhanced]\footnotesize
$\ParMeas(1^\secp,\ket{\psi})$:
\begin{itemize}
        \item Sample $k \gets \Gen(1^\kappa,n)$, and compute $\ket*{\psi_k} = \Enc_k(\ket{\psi})$.
        \item Sample a signature token $(\vk,\ket{\sk}) \gets \TokGen(1^\kappa)$.
        \item Let $H : \{0,1\}^* \to \{0,1\}^\kappa$ be a random oracle. 
        \item For each $i \in [t]$, define the function $\sF_i[x^*,\{v^*_\iota,r^*_\iota\}_{\iota \in [\tau]}](x,\sigma_{x},\widetilde{v}_1,\dots,\widetilde{v}_i,\widetilde{w}_i,\ell_1,\dots,\ell_{i-1})$: 
        \begin{itemize}
            \item If $x \neq x^*$, output $\sF_i\left(x,\sigma_{x},\widetilde{v}_1,\dots,\widetilde{v}_i,\widetilde{w}_i,\ell_1,\dots,\ell_{i-1}\right)$, where $\sF_i$ is defined as in $\QObf$.
            \item Output $\bot$ if $\TokVer(\vk,x,\sigma_x) = \bot$.
    
            \item For each $\iota \in [\min\{\tau,i-1\}]$, let \[\ell_{\iota,0} = H(x,\sigma_x,\widetilde{v}_1,\dots,\widetilde{v}_\iota,\ell_1,\dots,\ell_{\iota-1},0), ~~ \ell_{\iota,1} = H(x,\sigma_x,\widetilde{v}_1,\dots,\widetilde{v}_\iota,\ell_1,\dots,\ell_{\iota-1},1),\] and output $\bot$ if $\ell_{\iota,0} = \ell_{\iota,1}$ or $\ell_\iota \notin \{\ell_{\iota,0},\ell_{\iota,1}\}$. 
            
            \item If $i \leq \tau$:
            \begin{itemize}
                
                \item Output $\bot$ if $\Ver_{k,L_i \dots L_1,\theta_i}(\widetilde{v}_1,\dots,\widetilde{v}_i,\widetilde{w}_i) = \bot$. 
                 \item Set $\ell_i \coloneqq H(x,\sigma_x,\widetilde{v}_1,\dots,\widetilde{v}_i,\ell_1,\dots,\ell_{i-1},r^*_i)$, and output $(\ell_i,\widetilde{v}_i)$. 
            \end{itemize}

            \item If $i > \tau$:
            \begin{itemize}
                
                \item For each $\iota \in [\tau+1,i-1]$, let \[\ell_{\iota,0} = H(x,\sigma_x,\widetilde{v}_1,\dots,\widetilde{v}_\iota,\ell_1,\dots,\ell_{\iota-1},0), ~~ \ell_{\iota,1} = H(x,\sigma_x,\widetilde{v}_1,\dots,\widetilde{v}_\iota,\ell_1,\dots,\ell_{\iota-1},1),\] and output $\bot$ if $\ell_{\iota,0} = \ell_{\iota,1}$ or $\ell_\iota \notin \{\ell_{\iota,0},\ell_{\iota,1}\}$. Otherwise, let $r_\iota$ be such that $\ell_\iota = \ell_{\iota,r_\iota}$.
                \item Output $\bot$ if $\Ver_{k,L_i \dots L_1,\theta_i[V_1,\dots,V_\tau]}(\widetilde{v}_1,\dots,\widetilde{v}_\tau) = \bot$. 
                \item Compute $(v_{\tau+1},\dots,v_i,w_i) = \Dec_{k,L_i\dots L_1,\theta_i[V_{\tau+1},\dots,V_i,W_i]}(\widetilde{v}_{\tau+1},\dots,\widetilde{v}_i,\widetilde{w}_i)$, and output $\bot$ if the result is $\bot$.
                \item Compute $(\cdot,r_i) = f_i^{x,r^*_1,\dots,r^*_\tau,r_{\tau+1},\dots,r_{i-1}}(v^*_1,\dots,v^*_\tau,v_{\tau+1},\dots,v_i,w_i)$.
                 \item Set $\ell_i \coloneqq H(x,\sigma_x,\widetilde{v}_1,\dots,\widetilde{v}_i,\ell_1,\dots,\ell_{i-1},r_i)$, and output $(\ell_i,\widetilde{v}_i)$. 
            \end{itemize}
        \end{itemize}
        \item Define $\GSim$ as in $\QObf_2$.

        \item Let $B[x^*,\{v^*_\iota,r^*_\iota\}_{\iota \in [\tau]}]$ be the set of $(\sigma_{x^*},\widetilde{v}_1,\dots,\widetilde{v}_{t+1},\ell_1,\dots,\ell_t)$ such that the output of the following procedure is $\notin \{Q(x),\bot\}$:
            \begin{itemize}
                \item Output $\bot$ if $\TokVer(\vk,x,\sigma_x) = \bot$.
                \item For each $\iota \in [t]$, let \[\ell_{\iota,0} = H(x,\sigma_x,\widetilde{v}_1,\dots,\widetilde{v}_\iota,\ell_1,\dots,\ell_{\iota-1},0), ~~ \ell_{\iota,1} = H(x,\sigma_x,\widetilde{v}_1,\dots,\widetilde{v}_\iota,\ell_1,\dots,\ell_{\iota-1},1),\] and output $\bot$ if $\ell_{\iota,0} = \ell_{\iota,1}$ or $\ell_\iota \notin \{\ell_{\iota,0},\ell_{\iota,1}\}$. If $\iota > \tau$, let $r_\iota$ be such that $\ell_\iota = \ell_{\iota,r_\iota}$.
                \item Output $\bot$ if $\Ver_{k,L_{t+1} \dots L_1,\theta_{t+1}[V_1,\dots,V_\tau]}(\widetilde{v}_1,\dots,\widetilde{v}_\tau) = \bot$. 
                \item Compute $(v_{\tau+1},\dots,v_{t+1}) = \Dec_{k,L_{t+1}\dots L_1,\theta_{t+1}[V_{\tau+1},\dots,V_{t+1}]}(\widetilde{v}_{\tau+1},\dots,\widetilde{v}_{t+1})$, and output $\bot$ if the result is $\bot$.
                \item Output $g^{x,r^*_1,\dots,r^*_\tau,r_{\tau+1},\dots,r_t}(v_1^*,\dots,v_\tau^*,v_{\tau+1},\dots,v_{t+1}).$
            \end{itemize}
        \item Let $C[x^*,\{v^*_\iota,r^*_\iota\}_{\iota \in [\tau]}]$ be the set that includes, for any $i \in [\tau]$, all $(\sigma_{x^*},\widetilde{v}_1,\dots,\widetilde{v}_{i},\widetilde{w}_i,\ell_1,\dots,\ell_{i-1})$ such that 
        \begin{align*}
        &\sF_i[x^*,\{v^*_\iota,r^*_\iota\}_{\iota \in [\tau-1]}](x^*,\sigma_{x^*},\widetilde{v}_1,\dots,\widetilde{v}_{i},\widetilde{w}_i,\ell_1,\dots,\ell_{i-1})\\ &\neq \sF_i[x^*,\{v^*_\iota,r^*_\iota\}_{\iota \in [\tau]}](x^*,\sigma_{x^*},\widetilde{v}_1,\dots,\widetilde{v}_{i},\widetilde{w}_i,\ell_1,\dots,\ell_{i-1}),
        \end{align*} and all $(\sigma_{x^*},\widetilde{v}_1,\dots,\widetilde{v}_{t+1},\ell_1,\dots,\ell_{t})$ such that 
        \[(\sigma_{x^*},\widetilde{v}_1,\dots,\widetilde{v}_{t+1},\ell_1,\dots,\ell_{t}) \in B[x^*,\{v^*_\iota,r^*_\iota\}_{\iota \in [\tau-1]}] \setminus B[x^*,\{v^*_\iota,r^*_\iota\}_{\iota \in [\tau]}].\]
        
        \item Output $\ket*{\widetilde{\psi}} = \ket{\psi_k}\ket{\sk}, O = \left(\sF_1[\cdot],\dots,\sF_t[\cdot],\GSim\right), B[\cdot], C[\cdot].$
    \end{itemize}

\end{tcolorbox}

We now make a few remarks.

\begin{itemize}
    \item Throughout the remainder of the proof, we will be working with \emph{fully-deterministic} $\LM$ quantum programs for a given input $x^*$, i.e.
    \[\Pr[\LMEval(x^*,\ket{\psi},\{L_i\}_{i \in [t+1]},\{\theta_i\}_{i \in [t+1]}, \{f_i\}_{i \in [t+1]},g) \to Q(x^*)] = 1.\] Indeed, recall that we performed a post-selection on the correct outcome $Q(x^*)$ during the proof of \cref{claim:1-2} above.
    
    \item Whenever we reference an input $x^*$ and partial measurement results $\{v^*_\iota,r^*_\iota\}_{\iota \in [\tau]}$, we always mean measurement results that occur with \emph{non-zero} probability, i.e.\ they are in the support of the partial evaluation of $\ket{\psi}$ on input $x^*$.

    \item For $\tau = 0$ (i.e.\ no partial measurements), the above distribution $\ParMeas$ is identical to $\QObf_2$ augmented with the set $B[x^*]$ as defined in the proof of \cref{claim:1-2} (and $C[x^*]$ is undefined in this case since it requires $\tau \geq 1$).
    
    \item For any input $x^*$ and full set of measurement results $\{v^*_\iota,r^*_\iota\}_{\iota \in [t+1]}$, the set $B[x^*,\{v^*_\iota,r^*_\iota\}_{\iota \in [t+1]}]$ is empty, by virtue of the fact that these measurement results occur with non-zero probability, and the program outputs $Q(x^*)$ with probability 1.

    \item In the remainder of the proof, we will make use of the notation  $\ket{\psi[x^*,\{v^*_\iota,r^*_\iota\}_{\iota \in [\tau]}]}$ as defined in \cref{subsec:notation}.
\end{itemize}

Before proving the main inductive lemma of this section, we show the following statement, which essentially says that it is hard to find an element of

\[B[x^*,\{v^*_\iota,r^*_\iota\}_{\iota \in [\tau-1]}] \setminus B[x^*,\{v^*_\iota,r^*_\iota\}_{\iota \in [\tau]}]\] given the authenticated version of $\ket{\psi[x^*,\{v^*_\iota,r^*_\iota\}_{\iota \in [\tau]}]}$ and the oracles with $x^*$ and $\{v^*_\iota,r^*_\iota\}_{\iota \in [\tau]}$ hard-coded. The meat of this proof is actually deferred to the following section, in which we prove the "hardness of mapping" lemma, \cref{lemma:mapping-hardness}.

\begin{lemma}\label{lemma:parmeas-difference}
    For any input $x^*$, $\tau \in [1,\dots,t+1]$, measurement results $\{v^*_\iota,r^*_\iota\}_{\iota \in [\tau]}$, and QPQ unitary $U$, it holds that 

    \begin{align*}
        \bigg|&\E\left[\big\|\Pi\left[B\left[x^*,\{v^*_\iota,r^*_\iota\}_{\iota \in [\tau-1]}\right]\right]U^{\sF_1[x^*,\{v^*_\iota,r^*_\iota\}_{\iota \in [\tau-1]}],\dots,\sF_t[x^*,\{v^*_\iota,r^*_\iota\}_{\iota \in [\tau-1]}],\GSim}\ket*{\widetilde{\psi}}\big\|^2\right] \\
        &- \E\left[\big\|\Pi\left[B\left[x^*,\{v^*_\iota,r^*_\iota\}_{\iota \in [\tau]}\right]\right]U^{\sF_1[x^*,\{v^*_\iota,r^*_\iota\}_{\iota \in [\tau]}],\dots,\sF_t[x^*,\{v^*_\iota,r^*_\iota\}_{\iota \in [\tau]}],\GSim}\ket*{\widetilde{\psi}}\big\|^2\right]\bigg| \leq 2^{-\Omega(\kappa)},
    \end{align*}

    where both expectations are over 

    \[\ket*{\widetilde{\psi}},(\sF_1[\cdot],\sF_t[\cdot],\GSim),B[\cdot],C[\cdot] \gets \ParMeas(1^\secp,\ket{\psi[x^*,\{v^*_\iota,r^*_\iota\}_{\iota \in [\tau]}]}).\]

\end{lemma}

\begin{proof}

Note that the set $C[x^*,\{v^*_\iota,r^*_\iota\}_{\iota \in [\tau]}]$ includes all elements of \[B[x^*,\{v^*_\iota,r^*_\iota\}_{\iota \in [\tau-1]}] \setminus B[x^*,\{v^*_\iota,r^*_\iota\}_{\iota \in [\tau]}]\] and all inputs on which $\sF_i[x^*,\{v^*_\iota,r^*_\iota\}_{\iota \in [\tau-1]}]$ and $\sF_i[x^*,\{v^*_\iota,r^*_\iota\}_{\iota \in [\tau]}]$ differ for any $i \in [t]$. Thus, by \cref{lemma:puncture} (a standard oracle hybrid argument) it suffices to show that
    \[\E\left[\big\|\Pi\left[C\left[x^*,\{v^*_\iota,r^*_\iota\}_{\iota \in [\tau]}\right]\right]U^{\sF_1[x^*,\{v^*_\iota,r^*_\iota\}_{\iota \in [\tau]}],\dots,\sF_t[x^*,\{v^*_\iota,r^*_\iota\}_{\iota \in [\tau]}],\GSim}\ket*{\widetilde{\psi}}\big\|^2\right] \leq 2^{-\Omega(\kappa)},\] where the expectation is over \[\ket*{\widetilde{\psi}},(\sF_1[\cdot],\dots,\sF_t[\cdot],\GSim),B[\cdot],C[\cdot] \gets \ParMeas(1^\secp,\ket*{\psi[x^*,\{v^*_\iota,r^*_\iota\}_{\iota \in [\tau]}]}).\]

Now we consider all possible elements of the set $C[x^*,\{v^*_\iota,r^*_\iota\}_{\iota \in [\tau]}]$. By inspecting the definition, we see that any element of $C[x^*,\{v^*_\iota,r^*_\iota\}_{\iota \in [\tau]}]$ must fall into one of the following categories:

\begin{itemize}
    \item An input to $\sF_\tau[x^*,\{v_\iota^*,r_\iota^*\}_{\iota \in [\tau]}]$ that contains a sub-string $(\widetilde{v}_\tau,\widetilde{w}_\tau)$ such that \[f_\tau^{x^*,r_1^*,\dots,r^*_{\tau-1}}(v_1^*,\dots,v_{\tau-1}^*,\Dec_{k,L_\tau \dots L_1, \theta_\tau[V_\tau,W_\tau]}(\widetilde{v}_\tau,\widetilde{w}_\tau)) = (\cdot, 1-r_\tau^*),\] where $\Dec_{k,L_\tau \dots L_1, \theta_\tau[V_\tau,W_\tau]}(\widetilde{v}_\tau,\widetilde{w}_\tau) = (v_\tau,w_\tau) \neq \bot$.

    \item An input to $\sF_i[x^*,\{v_\iota^*,r_\iota^*\}_{\iota \in [\tau]}]$ for $i > \tau$ or an element of $B[x^*,\{v^*_\iota,r^*_\iota\}_{\iota \in [\tau-1]}] \setminus B[x^*,\{v^*_\iota,r^*_\iota\}_{\iota \in [\tau]}]$ that contains either: 
    
    \begin{itemize}
        \item $\widetilde{v}_\tau$ such that $\Dec_{k,L_\tau \dots L_1, \theta_\tau[V_\tau]}(\widetilde{v}_\tau) \notin \{v^*_\tau,\bot\}$, or
        \item $\ell_\tau$ such that $\ell_\tau = H(\dots, 1-r_\tau^*)$.
    \end{itemize} 
    
\end{itemize}

First, notice that by definition, the oracles $\sF_1[x^*,\{v^*_\iota,r^*_\iota\}_{\iota \in [\tau]}],\dots,\sF_t[x^*,\{v^*_\iota,r^*_\iota\}_{\iota \in [\tau]}],\GSim$ never output a label $H(\dots,1-r^*_\tau)$. Thus, $U$ can only successfully guess an $\ell_\tau$ such that $\ell_\tau = H(\dots,1-r^*_\tau)$ with probability $2^{-\kappa}$ over the randomness of the random oracle. 

Now, define $\Pi[\neg r^*_\tau]$ to be the projection onto strings $(\widetilde{v}_\tau,\widetilde{w}_\tau)$ such that \[f_\tau^{x^*,r_1^*,\dots,r^*_{\tau-1}}(v_1^*,\dots,v_{\tau-1}^*,\Dec_{k,L_\tau \dots L_1, \theta_\tau[V_\tau,W_\tau]}(\widetilde{v}_\tau,\widetilde{w}_\tau)) = (\cdot, 1-r_\tau^*),\] where $\Dec_{k,L_\tau \dots L_1, \theta_\tau[V_\tau,W_\tau]}(\widetilde{v}_\tau,\widetilde{w}_\tau) = (v_\tau,w_\tau) \neq \bot,$ and define $\Pi[\neg v^*_\tau]$ to be the projection onto strings $\widetilde{v}_\tau$ such that 
\[ \Dec_{k,L_\tau\dots L_1,\theta_\tau[V_\tau]}(\widetilde{v}_\tau) \notin \{v^*_\tau,\bot\}.\] Then, define \[\Pi[\neg (v^*_\tau,r^*_\tau)] \coloneqq \Pi[\neg r^*_\tau] + \Pi[\neg v^*_\tau].\]

Finally, given the sampled signature token verification key $\vk$, define the projector onto valid signatures of $x^*$: \[\Pi[x^*,\vk] \coloneqq \sum_{\sigma : \TokVer(\vk,x^*,\sigma) = \top}\ketbra{\sigma}{\sigma},\] and note that any element of $C\left[x^*,\{v^*_\iota,r^*_\iota\}_{\iota \in [\tau]}\right]$ must include a valid signature of $x^*$.
 
Now, by the preceding observations, we have that 
\begin{align*}
&\E\left[\big\|\Pi\left[C\left[x^*,\{v^*_\iota,r^*_\iota\}_{\iota \in [\tau]}\right]\right]U^{\sF_1[x^*,\{v^*_\iota,r^*_\iota\}_{\iota \in [\tau]}],\dots,\sF_t[x^*,\{v^*_\iota,r^*_\iota\}_{\iota \in [\tau]}],\GSim}\ket*{\widetilde{\psi}}\big\|^2\right]\\
&\leq \E\left[\big\|\left(\Pi[x^*,\vk] \otimes \Pi[\neg (v^*_\tau,r^*_\tau)]\right)U^{\sF_1[x^*,\{v^*_\iota,r^*_\iota\}_{\iota \in [\tau]}],\dots,\sF_t[x^*,\{v^*_\iota,r^*_\iota\}_{\iota \in [\tau]}],\GSim}\ket*{\widetilde{\psi}}\big\|^2\right] + 2^{-\Omega(\kappa)}\\ &\leq  2^{-\Omega(\kappa)},
\end{align*}
where the first inequality is due to the observation that $\ell_\tau$ such that $\ell_\tau = H(\dots,1-r^*_\tau)$ can only be guessed with probability $2^{-\kappa}$, and the second inequality is \cref{lemma:mapping-hardness} proven in the next section. This completes the proof.
\end{proof}

Now, we show the main lemma of the section.

\begin{lemma}\label{lemma:induction-1}

There exist a constant $c > 0$ such that for any input $x^*$, $\tau \in [0,\dots,t+1]$, measurement results $\{v^*_\iota,r^*_\iota\}_{\iota \in [\tau]}$, and QPQ unitary $U$, it holds that

\[\E\left[\big\|\Pi\left[B\left[x^*,\{v^*_\iota,r^*_\iota\}_{\iota \in [\tau]}\right]\right]U^{\sF_1[x^*,\{v^*_\iota,r^*_\iota\}_{\iota \in [\tau]}],\dots,\sF_t[x^*,\{v^*_\iota,r^*_\iota\}_{\iota \in [\tau]}],\GSim}\ket*{\widetilde{\psi}}\big\|^2\right] \leq 2^{3n(t+1-\tau)- c\kappa},\] where the expectation is over \[\ket*{\widetilde{\psi}},(\sF_1[\cdot],\dots,\sF_t[\cdot],\GSim),B[\cdot],C[\cdot] \gets \ParMeas(1^\secp,\ket*{\psi[x^*,\{v^*_\iota,r^*_\iota\}_{\iota \in [\tau]}]}).\]

\end{lemma}

\begin{proof}

We will show this by induction on $\tau$, starting with $\tau = t+1$ and ending with $\tau = 0$. The base case $(\tau = t+1)$ is trivial because the set $B[x^*,\{v^*_\iota,r^*_\iota\}_{\iota \in [t+1]}]$ is empty, as noted above. 

Now, we will let $c$ be a constant such that $2^{-c\kappa}$ is an upper bound on the expression in the statement of \cref{lemma:parmeas-difference}. For the inductive step, suppose that \cref{lemma:induction-1} holds for some $\tau \in [t+1]$. Consider any $x^*$ and measurement results $\{v^*_\iota,r^*_\iota\}_{\iota \in [\tau]}$, and define the following two distributions over $\ket*{\widetilde{\psi}},\allowbreak(\sF_1[\cdot],\dots,\sF_t[\cdot],\allowbreak\GSim),\allowbreak B[\cdot]$ (dropping the set $C[\cdot]$ since we don't need it for this proof):

\begin{itemize}
    \item $\cD: \ParMeas(1^\secp,\ket*{\psi[x^*,\{v^*_\iota,r^*_\iota\}_{\iota \in [\tau-1]}]})$.
    \item $\cD[v^*_\tau,r^*_\tau]: \ParMeas(1^\secp,\ket*{\psi[x^*,\{v^*_\iota,r^*_\iota\}_{\iota \in [\tau]}]})$.
\end{itemize}

To show that \cref{lemma:induction-1} holds for $\tau-1$, we have that  

\begin{align*}
    &\E_{\ket*{\widetilde{\psi}},(\sF_1[\cdot],\dots,\sF_t[\cdot],\GSim),B[\cdot] \gets \cD}\left[\big\|\Pi\left[B\left[x^*,\{v^*_\iota,r^*_\iota\}_{\iota \in [\tau-1]}\right]\right]U^{\sF_1[x^*,\{v^*_\iota,r^*_\iota\}_{\iota \in [\tau-1]}],\dots,\GSim}\ket*{\widetilde{\psi}}\big\|^2\right] \\
    &\leq 2^n \cdot \sum_{v^*_\tau,r^*_\tau}\E_{\substack{\ket*{\widetilde{\psi}},(\sF_1[\cdot],\dots,\sF_t[\cdot],\GSim),\\B[\cdot] \gets \cD[v^*_\tau,r^*_\tau]}}\left[\big\|\Pi\left[B\left[x^*,\{v^*_\iota,r^*_\iota\}_{\iota \in [\tau-1]}\right]\right]U^{\sF_1[x^*,\{v^*_\iota,r^*_\iota\}_{\iota \in [\tau-1]}],\dots,\GSim}\ket*{\widetilde{\psi}}\big\|^2\right]\\
    &\leq 2^n \cdot \sum_{v^*_\tau,r^*_\tau}\E_{\substack{\ket*{\widetilde{\psi}},(\sF_1[\cdot],\dots,\sF_t[\cdot],\GSim),\\B[\cdot] \gets \cD[v^*_\tau,r^*_\tau]}}\left[\big\|\Pi\left[B\left[x^*,\{v^*_\iota,r^*_\iota\}_{\iota \in [\tau]}\right]\right]U^{\sF_1[x^*,\{v^*_\iota,r^*_\iota\}_{\iota \in [\tau]}],\dots,\GSim}\ket*{\widetilde{\psi}}\big\|^2\right] + 2^{2n-c\kappa}\\
    &\leq 2^{2n + 3n(t+1-\tau)-c\kappa} + 2^{2n-c\kappa} \\
    &\leq 2^{3n(t+1-(\tau-1))-c\kappa-n} + 2^{2n-c\kappa} \\
    &\leq 2^{3n(t+1-(\tau-1))-c\kappa},
\end{align*}

where 

\begin{itemize}
    \item The first inequality follows from \cref{lemma:project} (an application of Cauchy-Schwarz).
    \item The second inequality follows from \cref{lemma:parmeas-difference} proven above.

    \item The third inequality follows from \cref{lemma:induction-1} for $\tau$ (the induction hypothesis).
    \item The final inequality follows because the two summands can each be bounded by the quantity $2^{3n(t+1-(\tau-1))-c\kappa -1}$.
\end{itemize}

\end{proof}

\subsection{Hardness of Mapping}\label{subsec:mapping-hardness}

Our final step is to prove the "hardness of mapping" lemma, \cref{lemma:mapping-hardness}, that was required for \cref{lemma:parmeas-difference} above. But first, we will need to introduce some "partial simulation" hybrid distributions $\ParSim_i$. We let $\ParSim_0 = \ParMeas$ as defined in the preceding section. We will change the distribution gradually until it corresponds to a simulated distribution, equivalent to $\QObf_5$ from \cref{subsec:main-theorem}.

\begin{tcolorbox}[breakable, enhanced]\footnotesize

$\ParSim_1(1^\secp,\ket{\psi})$:
\begin{itemize}
\item Sample $k \gets \Gen(1^\kappa,n)$, and compute $\ket*{\psi_k} = \Enc_k(\ket{\psi})$.
        \item Sample a signature token $(\vk,\ket{\sk}) \gets \TokGen(1^\kappa)$.
        \item Let $H : \{0,1\}^* \to \{0,1\}^\kappa$ be a random oracle. 
        \item For each $i \in [t]$, define the function $\sF_i[x^*,\{v^*_\iota,r^*_\iota\}_{\iota \in [\tau]},\textcolor{red}{\{w^*_\iota\}_{\iota \in [\tau+1,t]}}](x,\sigma_{x},\widetilde{v}_1,\dots,\widetilde{v}_i,\widetilde{w}_i,\ell_1,\dots,\ell_{i-1})$: 
        \begin{itemize}
            \item If $x \neq x^*$, output $\sF_i\left(x,\sigma_{x},\widetilde{v}_1,\dots,\widetilde{v}_i,\widetilde{w}_i,\ell_1,\dots,\ell_{i-1}\right)$, where $\sF_i$ is defined as in $\QObf$.
            \item Output $\bot$ if $\TokVer(\vk,x,\sigma_x) = \bot$.
    
            \item For each $\iota \in [\min\{\tau,i-1\}]$, let \[\ell_{\iota,0} = H(x,\sigma_x,\widetilde{v}_1,\dots,\widetilde{v}_\iota,\ell_1,\dots,\ell_{\iota-1},0), ~~ \ell_{\iota,1} = H(x,\sigma_x,\widetilde{v}_1,\dots,\widetilde{v}_\iota,\ell_1,\dots,\ell_{\iota-1},1),\] and output $\bot$ if $\ell_{\iota,0} = \ell_{\iota,1}$ or $\ell_\iota \notin \{\ell_{\iota,0},\ell_{\iota,1}\}$. 
            
            \item If $i \leq \tau$:
            \begin{itemize}
                
                \item Output $\bot$ if $\Ver_{k,L_i \dots L_1,\theta_i}(\widetilde{v}_1,\dots,\widetilde{v}_i,\widetilde{w}_i) = \bot$. 
                 \item Set $\ell_i \coloneqq H(x,\sigma_x,\widetilde{v}_1,\dots,\widetilde{v}_i,\ell_1,\dots,\ell_{i-1},r^*_i)$, and output $(\ell_i,\widetilde{v}_i)$. 
            \end{itemize}

            \item If $i > \tau$:
            \begin{itemize}
                \item For each $\iota \in [\tau+1,i-1]$, let \[\ell_{\iota,0} = H(x,\sigma_x,\widetilde{v}_1,\dots,\widetilde{v}_\iota,\ell_1,\dots,\ell_{\iota-1},0), ~~ \ell_{\iota,1} = H(x,\sigma_x,\widetilde{v}_1,\dots,\widetilde{v}_\iota,\ell_1,\dots,\ell_{\iota-1},1),\] and output $\bot$ if $\ell_{\iota,0} = \ell_{\iota,1}$ or $\ell_\iota \notin \{\ell_{\iota,0},\ell_{\iota,1}\}$. Otherwise, let $r_\iota$ be such that $\ell_\iota = \ell_{\iota,r_\iota}$.
                \item Output $\bot$ if $\Ver_{k,L_i \dots L_1,\theta_i[V_1,\dots,V_\tau,\textcolor{red}{W_i}]}(\widetilde{v}_1,\dots,\widetilde{v}_\tau,\textcolor{red}{\widetilde{w}_i}) = \bot$. 
                \item Compute $(v_{\tau+1},\dots,v_i) = \Dec_{k,L_i\dots L_1,\theta_i[V_{\tau+1},\dots,V_i]}(\widetilde{v}_{\tau+1},\dots,\widetilde{v}_i)$, and output $\bot$ if the result is $\bot$.
                \item Compute $(\cdot,r_i) = f_i^{x,r^*_1,\dots,r^*_\tau,r_{\tau+1},\dots,r_{i-1}}(v^*_1,\dots,v^*_\tau,v_{\tau+1},\dots,v_i,\textcolor{red}{w^*_i})$.
                 \item Set $\ell_i \coloneqq H(x,\sigma_x,\widetilde{v}_1,\dots,\widetilde{v}_i,\ell_1,\dots,\ell_{i-1},r_i)$, and output $(\ell_i,\widetilde{v}_i)$. 
            \end{itemize}
        \end{itemize}
        \item Define $\GSim$ as in $\QObf_2$.

        \item Output $\ket*{\widetilde{\psi}} = \ket{\psi_k}\ket{\sk}, O = \left(\sF_1[\cdot],\dots,\sF_t[\cdot],\GSim\right).$
    \end{itemize}

    \vspace{0.2cm}\hrule\vspace{0.3cm}

    $\ParSim_2(1^\secp,\ket{\psi})$:
    \begin{itemize}
        \item Sample $k \gets \Gen(1^\kappa,n)$, and compute $\ket*{\psi_k} = \Enc_k(\ket{\psi})$.
        \item Sample a signature token $(\vk,\ket{\sk}) \gets \TokGen(1^\kappa)$.
        \item Let $H : \{0,1\}^* \to \{0,1\}^\kappa$ be a random oracle. 
        \item For each $i \in [t]$, define the function $\sF_i[x^*,\{v^*_\iota,r^*_\iota\}_{\iota \in [\tau]},\{w^*_\iota\}_{\iota \in [\tau+1,t]}](x,\sigma_{x},\widetilde{v}_1,\dots,\widetilde{v}_i,\widetilde{w}_i,\ell_1,\dots,\ell_{i-1})$: 
        \begin{itemize}
            \item If $x \neq x^*$, output $\sF_i\left(x,\sigma_{x},\widetilde{v}_1,\dots,\widetilde{v}_i,\widetilde{w}_i,\ell_1,\dots,\ell_{i-1}\right)$, where $\sF_i$ is defined as in $\QObf$.
            \item Output $\bot$ if $\TokVer(\vk,x,\sigma_x) = \bot$.
    
            \item For each $\iota \in [\textcolor{red}{i-1}]$, let \[\ell_{\iota,0} = H(x,\sigma_x,\widetilde{v}_1,\dots,\widetilde{v}_\iota,\ell_1,\dots,\ell_{\iota-1},0), ~~ \ell_{\iota,1} = H(x,\sigma_x,\widetilde{v}_1,\dots,\widetilde{v}_\iota,\ell_1,\dots,\ell_{\iota-1},1),\] and output $\bot$ if $\ell_{\iota,0} = \ell_{\iota,1}$ or $\ell_\iota \notin \{\ell_{\iota,0},\ell_{\iota,1}\}$. 
            
            \item If $i \leq \tau$:
            \begin{itemize}
                
                \item Output $\bot$ if $\Ver_{k,L_i \dots L_1,\theta_i}(\widetilde{v}_1,\dots,\widetilde{v}_i,\widetilde{w}_i) = \bot$. 
                 \item Set $\ell_i \coloneqq H(x,\sigma_x,\widetilde{v}_1,\dots,\widetilde{v}_i,\ell_1,\dots,\ell_{i-1},r^*_i)$, and output $(\ell_i,\widetilde{v}_i)$. 
            \end{itemize}

            \item If $i > \tau$:
            \begin{itemize}
                \item Output $\bot$ if $\Ver_{k,L_i \dots L_1,\theta_i[V_1,\dots,V_\tau,W_i]}(\widetilde{v}_1,\dots,\widetilde{v}_\tau,\widetilde{w}_i) = \bot$. 
                \item Compute $(v_{\tau+1},\dots,v_i) = \Dec_{k,L_i\dots L_1,\theta_i[V_{\tau+1},\dots,V_i]}(\widetilde{v}_{\tau+1},\dots,\widetilde{v}_i)$, and output $\bot$ if the result is $\bot$.
                \item \textcolor{red}{For $\iota \in [\tau+1,i]$, compute $(\cdot,r_\iota) = f_\iota^{x,r^*_1,\dots,r^*_\tau,r_{\tau+1},\dots,r_{\iota-1}}(v^*_1,\dots,v^*_\tau,v_{\tau+1},\dots,v_\iota,w^*_\iota)$.}
                 \item Set $\ell_i \coloneqq H(x,\sigma_x,\widetilde{v}_1,\dots,\widetilde{v}_i,\ell_1,\dots,\ell_{i-1},r_i)$, and output $(\ell_i,\widetilde{v}_i)$. 
            \end{itemize}
        \end{itemize}
        \item Define $\GSim$ as in $\QObf_2$.

        \item Output $\ket*{\widetilde{\psi}} = \ket{\psi_k}\ket{\sk}, O = \left(\sF_1[\cdot],\dots,\sF_t[\cdot],\GSim\right).$
    \end{itemize}

    \vspace{0.2cm}\hrule\vspace{0.3cm}

    $\ParSim_3(1^\secp,\ket{\psi})$:
    \begin{itemize}
        \item Sample $k \gets \Gen(1^\kappa,n)$, and compute $\ket*{\psi_k} = \Enc_k(\ket{\psi})$.
        \item Sample a signature token $(\vk,\ket{\sk}) \gets \TokGen(1^\kappa)$.
        \item Let $H : \{0,1\}^* \to \{0,1\}^\kappa$ be a random oracle. 
        \item For each $i \in [t]$, define the function $\textcolor{red}{\FSim_i[x^*]}(x,\sigma_{x},\widetilde{v}_1,\dots,\widetilde{v}_i,\widetilde{w}_i,\ell_1,\dots,\ell_{i-1})$: 
        \begin{itemize}
            \item If $x \neq x^*$, output $\sF_i\left(x,\sigma_{x},\widetilde{v}_1,\dots,\widetilde{v}_i,\widetilde{w}_i,\ell_1,\dots,\ell_{i-1}\right)$, where $\sF_i$ is defined as in $\QObf$.
            \item Output $\bot$ if $\TokVer(\vk,x,\sigma_x) = \bot$.
            \item For each $\iota \in [i-1]$, let \[\ell_{\iota,0} = H(x,\sigma_x,\widetilde{v}_1,\dots,\widetilde{v}_\iota,\ell_1,\dots,\ell_{\iota-1},0), ~~ \ell_{\iota,1} = H(x,\sigma_x,\widetilde{v}_1,\dots,\widetilde{v}_\iota,\ell_1,\dots,\ell_{\iota-1},1),\] and output $\bot$ if $\ell_{\iota,0} = \ell_{\iota,1}$ or $\ell_\iota \notin \{\ell_{\iota,0},\ell_{\iota,1}\}$. 
            \item \textcolor{red}{Output $\bot$ if $\Ver_{k,L_i \dots L_1,\theta_i}(\widetilde{v}_1,\dots,\widetilde{v}_i,\widetilde{w}_i) = \bot$.}
            \item Set $\ell_i \coloneqq H(x,\sigma_x,\widetilde{v}_1,\dots,\widetilde{v}_i,\ell_1,\dots,\ell_{i-1},\textcolor{red}{0}),$ and output $(\widetilde{v}_i,\ell_i)$.
        \end{itemize}
        \item Define $\GSim$ as in $\QObf_2$.

        \item Output $\ket*{\widetilde{\psi}} = \ket{\psi_k}\ket{\sk}, O = \left(\FSim_1[\cdot],\dots,\FSim_t[\cdot],\GSim\right).$
    \end{itemize}

    \vspace{0.2cm}\hrule\vspace{0.3cm}

    $\ParSim_4(1^\secp,\ket{\psi})$:
    \begin{itemize}
        \item Sample $k \gets \Gen(1^\kappa,n)$, and compute $\ket*{\psi_k} = \Enc_k(\ket{\psi})$.
        \item Sample a signature token $(\vk,\ket{\sk}) \gets \TokGen(1^\kappa)$.
        \item Let $H : \{0,1\}^* \to \{0,1\}^\kappa$ be a random oracle. 
        \item For each $i \in [t]$, define the function $\FSim_i[x^*, \textcolor{red}{w^*}](x,\sigma_{x},\widetilde{v}_1,\dots,\widetilde{v}_i,\widetilde{w}_i,\ell_1,\dots,\ell_{i-1})$: 
        \begin{itemize}
            \item If $x = x^*$, output $\FSim_i(x^*,\sigma_{x^*},\widetilde{v}_1,\dots,\widetilde{v}_i,\widetilde{w}_i,\ell_1,\dots,\ell_{i-1})$, where $\FSim_i$ is defined as in $\QObf_5$.
            \item \textcolor{red}{If $x \neq x^*$, output $\FSim_i[w^*](x,\sigma_{x},\widetilde{v}_1,\dots,\widetilde{v}_i,\widetilde{w}_i,\ell_1,\dots,\ell_{i-1})$, where $\FSim_i[w^*]$ is defined as in $\QObf_4$.}
        \end{itemize}
        \item Define $\GSim$ as in $\QObf_2$.

        \item Output $\ket*{\widetilde{\psi}} = \ket{\psi_k}\ket{\sk}, O = \left(\FSim_1[\cdot],\dots,\FSim_t[\cdot],\GSim\right).$
    \end{itemize}

        \vspace{0.2cm}\hrule\vspace{0.3cm}

    $\ParSim_5(1^\secp,\ket{\psi})$:
    \begin{itemize}
        \item Sample $k \gets \Gen(1^\kappa,n)$, and compute $\ket*{\psi_k} = \Enc_k(\ket{\psi})$.
        \item Sample a signature token $(\vk,\ket{\sk}) \gets \TokGen(1^\kappa)$.
        \item Let $H : \{0,1\}^* \to \{0,1\}^\kappa$ be a random oracle. 
        \item For each $i \in [t]$, define the function \textcolor{red}{$\FSim_i$ as in $\QObf_5$.}
        \item Define $\GSim$ as in $\QObf_2$.
        \item Output $\ket*{\widetilde{\psi}} = \ket{\psi_k}\ket{\sk}, O = \left(\FSim_1,\dots,\FSim_t,\GSim\right).$
    \end{itemize}

\end{tcolorbox}

Next, before proving the main lemma (\cref{lemma:mapping-hardness}) of this section, which involves $\ParMeas_0$, we prove several claims about these hybrid distributions, which will allow us to use the properties of simulated distributions when proving \cref{lemma:mapping-hardness}. We will essentially "work backwards" from $\ParSim_5$ to $\ParSim_0$ in order to show the sequence of indistinguishability claims we will need. Whenever we write $\{w^*_i\}_{i \in [S]}$ for some set $S \subseteq [t]$, we parse each $w^*_i \in \{0,1\}^{|W_i|}$.

First, we confirm that the mapping hardness claims we'll need hold in the fully simulated case of $\ParSim_5$.

\begin{claim}\label{claim:parsim-auth}
    For any QPQ unitary $U$, input $x^*$, partial measurement results $\{v^*_\iota,r^*_\iota\}_{\iota \in [\tau]}$, and $w^* = \{w^*_i\}_{i \in [t]}$, it holds that

    \[\E\left[\big\| \Pi[\neg w^*]U^O\Pi[w^*]\ket*{\widetilde{\psi}}\big\|^2 : \ket*{\widetilde{\psi}},O \gets \ParSim_5(1^\secp,\ket*{\psi[x^*,\{v^*_\iota,r^*_\iota\}_{\iota \in [\tau]}]})\right] = 2^{-\Omega(\kappa)}\]

    and 
    \[\E\left[\big\| \Pi[\neg (v^*_\tau,r^*_\tau)]U^O\Pi[\{w^*_\iota\}_{\iota \in [\tau+1,t]}]\ket*{\widetilde{\psi}} \big\|^2 : \ket*{\widetilde{\psi}},O \gets \ParSim_5(1^\secp,\ket*{\psi[x^*,\{v^*_\iota,r^*_\iota\}_{\iota \in [\tau]}]})\right] = 2^{-\Omega(\kappa)},\]

    where $\Pi[\neg (v^*_\tau,r^*_\tau)]$ is defined in the proof of \cref{lemma:parmeas-difference}.

\end{claim}

\begin{proof}
    The key point is that in $\QObf_5$, none of the oracles $\FSim_1,\dots,\FSim_t,\GSim$ require access to the decryption oracle $\Dec_{k,\cdot,\cdot}(\cdot)$ for the authentication scheme. Rather, they can be implemented just given access to the verification oracle $\Ver_{k,\cdot,\cdot}(\cdot)$. Thus, these claims follow directly from the security of the authentication scheme (\cref{thm:auth-sec}), and in particular that it satisfies \cref{def:auth-mapping-sec} (mapping security). Note that in the second claim, we only collapse the wires $\{W_i\}_{i \in [\tau+1,t]}$, that is, the $W$ wires \emph{after} level $\tau$, which are disjoint from $(V_\tau,W_\tau)$. The claim could have been trivially false if we had collapsed the wires in $(V_\tau,W_\tau)$, in particular to an outcome different from $(v_\tau^*,r^*_\tau)$.
\end{proof}

The next two claims establish the indistinguishability of $\ParSim_3,\ParSim_4,$ and $\ParSim_5$ in the case when all of the $W$ wires have been collapsed to some value $w^*$.

\begin{claim}\label{claim:parsim4-5}
    For any (unbounded) distinguisher $D$, input $x^*$, partial measurement results $\{v^*_\iota,r^*_\iota\}_{\iota \in [\tau]}$, and $w^* = \{w^*_i\}_{i \in [t]}$, it holds that 
    \begin{align*}&\E\left[\big\|D^{\FSim_1[x^*,w^*],\dots,\GSim}\left(k,\Pi[w^*]\ket*{\widetilde{\psi}}\right)\big\|^2 : \ket*{\widetilde{\psi}},(\FSim_1[\cdot],\dots,\GSim) \gets \ParSim_4(1^\secp,\ket*{\psi[x^*,\{v^*_\iota,r^*_\iota\}_{\iota \in [\tau]}]})\right]\\ &= \E\left[\big\|D^{\FSim_1,\dots,\GSim}\left(k,\Pi[w^*]\ket*{\widetilde{\psi}}\right)\big\|^2 : \ket*{\widetilde{\psi}},(\FSim_1,\dots,\GSim) \gets \ParSim_5(1^\secp,\ket*{\psi[x^*,\{v^*_\iota,r^*_\iota\}_{\iota \in [\tau]}]})\right].\end{align*}
    where $D$'s input includes the key $k$ sampled by $\ParSim_4,\ParSim_5$.
\end{claim}

\begin{proof}
    The proof is exactly the same as the proof of \cref{claim:prelim4-5}, except that here we are only switching oracle inputs on $x \neq x^*$ rather than all $x$.
\end{proof}

\begin{claim}\label{claim:parsim3-4}
    For any QPQ distinguisher $D$, input $x^*$, partial measurement results $\{v^*_\iota,r^*_\iota\}_{\iota \in [\tau]}$, and $w^* = \{w^*_i\}_{i \in [t]}$, it holds that 
    \begin{align*}&\bigg|\E\left[\big\|D^{\FSim_1[x^*],\dots,\GSim}\left(k,\Pi[w^*]\ket*{\widetilde{\psi}}\right)\big\|^2 : \ket*{\widetilde{\psi}},(\FSim_1[\cdot],\dots,\GSim) \gets \ParSim_3(1^\secp,\ket*{\psi[x^*,\{v^*_\iota,r^*_\iota\}_{\iota \in [\tau]}]})\right] \\ &-\E\left[\big\|D^{\FSim_1[x^*,w^*],\dots,\GSim}\left(k,\Pi[w^*]\ket*{\widetilde{\psi}}\right)\big\|^2 : \ket*{\widetilde{\psi}},(\FSim_1[\cdot],\dots,\GSim) \gets \ParSim_4(1^\secp,\ket*{\psi[x^*,\{v^*_\iota,r^*_\iota\}_{\iota \in [\tau]}]})\right]\bigg|\\ &= 2^{-\Omega(\kappa)},\end{align*}
    where $D$'s input includes the key $k$ sampled by $\ParSim_3,\ParSim_4$.
\end{claim}

\begin{proof}
    The proof is exacly the same as the proof of \cref{claim:prelim3-4}, except that here we are only switching oracle inputs on $x \neq x^*$ rather than all $x$.
\end{proof}

This next claim shows that in $\ParSim_3$, it is hard to map wires $W_{\tau+1},\dots,W_t$ collapsed to $w^*_{\tau+1},\dots,w^*_t$ to an outcome different from $w^*_{\tau+1},\dots,w^*_t$.

\begin{claim}\label{claim:parsim-partialmap}
    For any QPQ unitary $U$, input $x^*$, partial measurement results $\{v^*_\iota,r^*_\iota\}_{\iota \in [\tau]}$, and $\{w^*_i\}_{i \in [\tau+1,t]}$, it holds that

    \[\E\left[\big\| \Pi[\neg \{w^*_i\}_{i \in [\tau+1,t]}]U^{\FSim_1[x^*],\dots,\GSim}\Pi[\{w^*_i\}_{i \in [\tau+1,t]}]\ket*{\widetilde{\psi}}\big\|^2\right] = 2^{-\Omega(\kappa)},\] where the expectation is over \[\ket*{\widetilde{\psi}},(\FSim_1[\cdot],\dots,\GSim) \gets \ParSim_3(1^\secp,\ket*{\psi[x^*,\{v^*_\iota,r^*_\iota\}_{\iota \in [\tau]}]}).\]
\end{claim}

\begin{proof}
    We write $\E_{\ParSim_3}$ as shorthand for the expectation over \[\ket*{\widetilde{\psi}},(\FSim_1[\cdot],\dots,\GSim) \gets \ParSim_3(1^\secp,\ket*{\psi[x^*,\{v^*_\iota,r^*_\iota\}_{\iota \in [\tau]}]}),\] we write $\E_{\ParSim_5}$ as shorthand for the expectation over \[\ket*{\widetilde{\psi}},(\FSim_1,\dots,\GSim) \gets \ParSim_5(1^\secp,\ket*{\psi[x^*,\{v^*_\iota,r^*_\iota\}_{\iota \in [\tau]}]}),\] and, given $\{w^*_\iota\}_{\iota \in [\tau]}, \{w^*_{\iota}\}_{\iota \in [\tau+1,t]}$, we write $w^* = \{w^*_\iota\}_{\iota \in [\tau]} \cup \{w^*_{\iota}\}_{\iota \in [\tau+1,t]}$. Then, 
    
    \begin{align*}
    &\E_{\ParSim_3}\left[\big\| \Pi[\neg \{w^*_i\}_{i \in [\tau+1,t]}]U^{\FSim_1[x^*],\dots,\GSim}\Pi[\{w^*_i\}_{i \in [\tau+1,t]}]\ket*{\widetilde{\psi}}\big\|^2\right] \\
    &\leq 2^n \cdot \sum_{\{w^*_\iota\}_{\iota \in [\tau]}}\E_{\ParSim_3}\left[\big\| \Pi[\neg \{w^*_i\}_{i \in [\tau+1,t]}]U^{\FSim_1[x^*],\dots,\GSim}\Pi[w^*]\ket*{\widetilde{\psi}}\big\|^2\right] \\
    &\leq 2^n \cdot \sum_{\{w^*_\iota\}_{\iota \in [\tau]}}\E_{\ParSim_5}\left[\big\| \Pi[\neg \{w^*_i\}_{i \in [\tau+1,t]}]U^{\FSim_1,\dots,\GSim}\Pi[w^*]\ket*{\widetilde{\psi}}\big\|^2\right] + 2^n \cdot 2^{-\Omega(\kappa)} \\ 
    &= 2^{-\Omega(\kappa)} 
    \end{align*}

    where 
    \begin{itemize}
        \item The first inequality follows from \cref{lemma:project} (an application of Cauchy-Schwarz).
        \item The second inequality follows by combining \cref{claim:parsim3-4} and \cref{claim:parsim4-5} proven above.
        \item The third inequality follows from \cref{claim:parsim-auth} proven above.
    \end{itemize}
\end{proof}

In the proof of the next two claims, we will use the following notation. Fix a key $k$, input $x^*$, partial measurement results $\{v^*_\iota,r^*_\iota\}_{\iota \in [\tau]}$, and $\{w^*_\iota\}_{i \in [\tau+1,t]}$, and define the following function: \\

\noindent $R[k,x^*,\{v^*_\iota,r^*_\iota\}_{\iota \in [\tau]},\{w^*_\iota\}_{\iota \in [\tau+1,t]}] : (\widetilde{v}_1,\dots,\widetilde{v}_i) \to r_i$
\begin{itemize}
    \item If $i \leq \tau$, output $r^*_i$.
    \item Otherwise, compute $(v_{\tau+1},\dots,v_i) = \Dec_{k,L_i\dots L_1,\theta_i[V_{\tau+1},\dots,V_i]}(\widetilde{v}_{\tau+1},\dots,\widetilde{v}_i)$, and output $\bot$ if the result is $\bot$.
    \item For $\iota \in [\tau+1,i],$ compute $(\cdot,r_\iota) = f_\iota^{x^*,r^*_1,\dots,r^*_\tau,r_{\tau+1},\dots,r_{\iota-1}}(v_1^*,\dots,v_\tau^*,v_{\tau+1},\dots,v_\iota,w^*_\iota)$.
    \item Output $r_i$.
\end{itemize}

This function determines the bit $r_i$ when $x^*$, $\{v^*_\iota,r^*_\iota\}_{\iota \in [\tau]}$, and $\{w^*_\iota\}_{\iota \in [\tau+1,t]}$ have been hard-coded into the oracles in $\ParSim_2$. We will use it when showing indistinguishability between $\ParSim_1,\ParSim_2$, and $\ParSim_3$ in the case when the $W_{\tau+1},\dots,W_t$ wires of the input state have been collapsed to outcome $w^*_{\tau+1},\dots,w^*_t$.

\begin{claim}\label{claim:parsim2-3}
    For any (unbounded) distinguisher $D$, input $x^*$, partial measurement results $\{v^*_\iota,r^*_\iota\}_{\iota \in [\tau]}$, and $\{w^*_\iota\}_{i \in [\tau+1,t]}$, it holds that 
    \begin{align*}&\E\left[\big\|D^{\sF_1[x^*,\{v^*_\iota,r^*_\iota\}_{\iota \in [\tau]},\{w^*_\iota\}_{\iota \in [\tau+1,t]}],\dots,\GSim}\left(k,\vk,\Pi[\{w^*_\iota\}_{\iota \in [\tau+1,t]}]\ket*{\widetilde{\psi}}\right)\big\|^2\right]\\ &= \E\left[\big\|D^{\FSim_1[x^*],\dots,\GSim}\left(k,\vk,\Pi[\{w^*_\iota\}_{\iota \in [\tau+1,t]}]\ket*{\widetilde{\psi}}\right)\big\|^2\right],\end{align*}

    where the first expectation is over 
    \[\ket*{\widetilde{\psi}},(\sF_1[\cdot],\dots,\sF_t[\cdot],\GSim) \gets  \ParSim_2(1^\secp,\ket{\psi[x^*,\{v^*_\iota,r^*_\iota\}_{\iota \in [\tau]}]}),\]
    
    the second expectation is over 

    \[\ket*{\widetilde{\psi}},(\FSim_1[\cdot],\dots,\FSim_t[\cdot],\GSim) \gets  \ParSim_3(1^\secp,\ket{\psi[x^*,\{v^*_\iota,r^*_\iota\}_{\iota \in [\tau]}]}),\]
    
    and $D$'s input includes the keys $k,\vk$ sampled by $\ParSim_2,\ParSim_3$.
\end{claim}

\begin{proof}
    In $\ParSim_2$, we have that for inputs that begin with $x^*$, 
    \begin{align*}
        &\sF_i[x^*,\{v^*_\iota,r^*_\iota\}_{\iota \in [\tau]},\{w^*_\iota\}_{\iota \in [\tau+1,t]}](x^*,\sigma_{x^*},\widetilde{v}_1,\dots,\widetilde{v}_i,\widetilde{w}_i,\ell_1,\dots,\ell_{i-1})\\
        & \in \{H(x^*,\sigma_{x^*},\widetilde{v}_1,\dots,\widetilde{v}_i,\ell_1,\dots,\ell_{i-1},R[k,x^*,\{v^*_\iota,r^*_\iota\}_{\iota \in [\tau]},\{w^*_\iota\}_{\iota \in [\tau+1,t]}](\widetilde{v}_1,\dots,\widetilde{v}_i)), \bot\},
    \end{align*}

    while in $\ParSim_3$, we have that for inputs that begin with $x^*$, 

    \begin{align*}
        &\FSim_i[x^*,\{v^*_\iota,r^*_\iota\}_{\iota \in [\tau]},\{w^*_\iota\}_{\iota \in [\tau+1,t]}](x^*,\sigma_{x^*},\widetilde{v}_1,\dots,\widetilde{v}_i,\widetilde{w}_i,\ell_1,\dots,\ell_{i-1})\\
        & \in \{H(x^*,\sigma_{x^*},\widetilde{v}_1,\dots,\widetilde{v}_i,\ell_1,\dots,\ell_{i-1},0), \bot\}.
    \end{align*}

    Both implementations of the oracles are identical on inputs $x \neq x^*$, and both will always output $\bot$ on the same set of inputs, since this is true of $\Dec_{k,\cdot,\cdot}(\cdot)$ and $\Ver_{k,\cdot,\cdot}(\cdot)$ by definition. Finally, their non-$\bot$ answers on inputs that begin with $x^*$ are identically distributed over the  randomness of the random oracle $H$, since each $(x^*,\sigma_{x^*},\widetilde{v}_1,\dots,\widetilde{v}_i,\ell_1,\dots,\ell_{i-1})$ fixes a single choice of bit \[R[k,x^*,\{v^*_\iota,r^*_\iota\}_{\iota \in [\tau]},\{w^*_\iota\}_{\iota \in [\tau+1,t]}](\widetilde{v}_1,\dots,\widetilde{v}_i).\] Indeed, for any $(x^*,\sigma_{x^*},\widetilde{v}_1,\dots,\widetilde{v}_i,\ell_1,\dots,\ell_{i-1})$, \[H(x^*,\sigma_{x^*},\widetilde{v}_1,\dots,\widetilde{v}_i,\widetilde{w}_i,\ell_1,\dots,\ell_{i-1},0) \ \ \ \text{and} \ \ \  H(x^*,\sigma_{x^*},\widetilde{v}_1,\dots,\widetilde{v}_i,\widetilde{w}_i,\ell_1,\dots,\ell_{i-1},1)\] are identically distributed (each is a uniformly random string).

\end{proof}

\begin{claim}\label{claim:parsim1-2}
    For any QPQ distinguisher $D$, input $x^*$, partial measurement results $\{v^*_\iota,r^*_\iota\}_{\iota \in [\tau]}$, and $\{w^*_\iota\}_{i \in [\tau+1,t]}$, it holds that  
    \begin{align*}&\bigg|\E\left[\big\|D^{\sF_1[x^*,\{v^*_\iota,r^*_\iota\}_{\iota \in [\tau]},\{w^*_\iota\}_{\iota \in [\tau+1,t]}],\dots,\GSim}\left(k,\vk,\Pi[\{w^*_\iota\}_{\iota \in [\tau+1,t]}]\ket*{\widetilde{\psi}}\right)\big\|^2\right]\\ &- \E\left[\big\|D^{\sF_1[x^*,\{v^*_\iota,r^*_\iota\}_{\iota \in [\tau]},\{w^*_\iota\}_{\iota \in [\tau+1,t]}],\dots,\GSim}\left(k,\vk,\Pi[\{w^*_\iota\}_{\iota \in [\tau+1,t]}]\ket*{\widetilde{\psi}}\right)\big\|^2\right]\bigg| = 2^{-\Omega(\kappa)},\end{align*}
    where the first expectation is over 
    \[\ket*{\widetilde{\psi}},(\sF_1[\cdot],\dots,\sF_t[\cdot],\GSim) \gets  \ParSim_1(1^\secp,\ket*{\psi[x^*,\{v^*_\iota,r^*_\iota\}_{\iota \in [\tau]}]}),\]
    and the second expectation is over 
    \[\ket*{\widetilde{\psi}},(\sF_1[\cdot],\dots,\sF_t[\cdot],\GSim) \gets  \ParSim_2(1^\secp,\ket*{\psi[x^*,\{v^*_\iota,r^*_\iota\}_{\iota \in [\tau]}]}),\]
    and $D$'s input includes the keys $k,\vk$ sampled by $\ParSim_1,\ParSim_2$.
\end{claim}

\begin{proof}
    Observe that the oracles $\sF_i[x^*,\{v^*_\iota,r^*_\iota\}_{\iota \in [\tau]},\{w^*_\iota\}_{\iota \in [\tau+1,t]}]$ in these experiments are identical except for on inputs \[(x^*,\sigma_{x^*},\widetilde{v}_1,\dots,\widetilde{v}_i,\widetilde{w}_i,\ell_1,\dots,\ell_{i-1})\] such that there exists an $\iota \in [i-1]$ with \[\ell_{\iota} = H(x^*,\sigma_{x^*},\widetilde{v}_1,\dots,\widetilde{v}_\iota,\ell_1,\dots,\ell_{\iota-1},1-R[k,x^*,\{v^*_\iota,r^*_\iota\}_{\iota \in [\tau]},\{w^*_\iota\}_{\iota \in [\tau+1,t]}](\widetilde{v}_1,\dots,\widetilde{v}_\iota)).\] However, the oracles $\sF_i[x^*,\{v^*_\iota,r^*_\iota\}_{\iota \in [\tau]},\{w^*_\iota\}_{\iota \in [\tau+1,t]}]$ in $\ParSim_2$ are defined to never output such an $\ell_\iota$, and thus such an input can only be guessed with probability $2^{-\kappa}$ over the randomness of $H$. The claim follows by applying \cref{lemma:puncture} (a standard oracle hybrid argument).
\end{proof}

Now, in our final claim before the main lemma of this section, we combine what we have shown so far to establish the indistinguishability of $\ParSim_0$ and $\ParSim_1$ in the case when the $W_{\tau+1},\dots,W_t$ wires of the input state have been collapsed to some outcome $w^*_{\tau+1},\dots,w^*_t$.

\begin{claim}\label{claim:parsim0-1}
    For any QPQ distinguisher $D$, input $x^*$, partial measurement results $\{v^*_\iota,r^*_\iota\}_{\iota \in [\tau]}$, and $\{w^*_\iota\}_{i \in [\tau+1,t]}$, it holds that  
    \begin{align*}&\bigg|\E\left[\big\|D^{\sF_1[x^*,\{v^*_\iota,r^*_\iota\}_{\iota \in [\tau]}],\dots,\GSim}\left(k,\vk,\Pi[\{w^*_\iota\}_{\iota \in [\tau+1,t]}]\ket*{\widetilde{\psi}}\right)\big\|^2\right]\\ &- \E\left[\big\|D^{\sF_1[x^*,\{v^*_\iota,r^*_\iota\}_{\iota \in [\tau]},\{w^*_\iota\}_{\iota \in [\tau+1,t]}],\dots,\GSim}\left(k,\vk,\Pi[\{w^*_\iota\}_{\iota \in [\tau+1,t]}]\ket*{\widetilde{\psi}}\right)\big\|^2\right]\bigg| = 2^{-\Omega(\kappa)},\end{align*}
    where the first expectation is over 
    \[\ket*{\widetilde{\psi}},(\sF_1[\cdot],\dots,\sF_t[\cdot],\GSim) \gets  \ParSim_0(1^\secp,\ket*{\psi[x^*,\{v^*_\iota,r^*_\iota\}_{\iota \in [\tau]}]}),\]
    the second expectation is over
    \[\ket*{\widetilde{\psi}},(\sF_1[\cdot],\dots,\sF_t[\cdot],\GSim) \gets  \ParSim_1(1^\secp,\ket*{\psi[x^*,\{v^*_\iota,r^*_\iota\}_{\iota \in [\tau]}]}),\]
    and $D$'s input includes the keys $k,\vk$ sampled by $\ParSim_0,\ParSim_1$.
\end{claim}

\begin{proof}
    Observe that the oracles in these experiments are identical except for on inputs that include a $\widetilde{w}_i \in P[\neg w_i^*]$ for some $i \in [\tau+1,t]$. Thus, by \cref{lemma:puncture} (a standard oracle hybrid argument), it suffices to show that for any QPQ unitary $U$, 
    \begin{align*}
        \E\left[\Pi[\neg \{w^*_\iota\}_{i \in [\tau+1,t]}]U^{\sF_1[x^*,\{v^*_\iota,r^*_\iota\}_{\iota \in [\tau]},\{w^*_\iota\}_{\iota \in [\tau+1,t]}],\dots,\GSim}\Pi[\{w^*_\iota\}_{i \in [\tau+1,t]}]\ket*{\widetilde{\psi}}\right] = 2^{-\Omega(\kappa)},
    \end{align*}
    where the expectation is over

    \[\ket*{\widetilde{\psi}},(\sF_1[\cdot],\dots,\sF_t[\cdot],\GSim) \gets  \ParSim_1(1^\secp,\ket*{\psi[x^*,\{v^*_\iota,r^*_\iota\}_{\iota \in [\tau]}]}).\]

     Write $\E_{\ParSim_1}$ as shorthand for the expectation over \[\ket*{\widetilde{\psi}},(\sF_1[\cdot],\dots,\FSim_t[\cdot],\GSim) \gets \ParSim_1(1^\secp,\ket*{\psi[x^*,\{v^*_\iota,r^*_\iota\}_{\iota \in [\tau]}]}),\] and write $\E_{\ParSim_3}$ as shorthand for the expectation over \[\ket*{\widetilde{\psi}},(\FSim_1[\cdot],\dots,\FSim_t[\cdot],\GSim) \gets \ParSim_3(1^\secp,\ket*{\psi[x^*,\{v^*_\iota,r^*_\iota\}_{\iota \in [\tau]}]}).\]
        
    Then, 
    \begin{align*}
        &\E_{\ParSim_1}\left[\Pi[\neg \{w^*_\iota\}_{i \in [\tau+1,t]}]U^{\sF_1[x^*,\{v^*_\iota,r^*_\iota\}_{\iota \in [\tau]},\{w^*_\iota\}_{\iota \in [\tau+1,t]}],\dots,\GSim}\Pi[\{w^*_\iota\}_{i \in [\tau+1,t]}]\ket*{\widetilde{\psi}}\right] \\
        &\leq \E_{\ParSim_3}\left[\Pi[\neg \{w^*_\iota\}_{i \in [\tau+1,t]}]U^{\FSim_1[x^*],\dots,\GSim}\Pi[\{w^*_\iota\}_{i \in [\tau+1,t]}]\ket*{\widetilde{\psi}}\right] + 2^{-\Omega(\kappa)} \\
        &\leq 2^{-\Omega(\kappa)},
    \end{align*}

    where the first inequality follows by combining \cref{claim:parsim1-2} and \cref{claim:parsim2-3}, and the second inequality follows from \cref{claim:parsim-partialmap}, all proven above.
\end{proof}

Now, we prove the main lemma of this section. The proof of \cref{lemma:mapping-hardness} combines some claims proven above with \cref{lemma:mapping-hardness-3} that follows. \cref{lemma:mapping-hardness-3} is a similar claim but with respect to $\ParSim_3$ instead of $\ParSim_0$.

\begin{lemma}\label{lemma:mapping-hardness}
    For any QPQ unitary $U$, input $x^*$, and partial measurement results $\{v^*_\iota,r^*_\iota\}_{\iota \in [\tau]}$, it holds that 
    \[\E\left[\big\|\left(\Pi[x^*,\vk] \otimes \Pi[\neg (v^*_\tau,r^*_\tau)]\right)U^{\sF_1[x^*,\{v^*_\iota,r^*_\iota\}_{\iota \in [\tau]}],\dots,\GSim}\ket*{\widetilde{\psi}}\big\|^2\right] = 2^{-\Omega(\kappa)},\]

    where the projectors are defined as in the proof of \cref{lemma:parmeas-difference}, and the expectation is over

    \[\ket*{\widetilde{\psi}},(\sF_1[\cdot],\dots,\sF_t[\cdot],\GSim) \gets \ParMeas(1^\secp,\ket*{\psi[x^*,\{v^*_\iota,r^*_\iota\}_{\iota \in [\tau]}]}),\] which, recall, is defined to be the same as 
     \[\ket*{\widetilde{\psi}},(\sF_1[\cdot],\dots,\sF_t[\cdot],\GSim) \gets \ParSim_0(1^\secp,\ket*{\psi[x^*,\{v^*_\iota,r^*_\iota\}_{\iota \in [\tau]}]}).\]
    
\end{lemma}

\begin{proof}

    Write $\E_{\ParSim_0}$ as shorthand for the expectation over \[\ket*{\widetilde{\psi}},(\sF_1[\cdot],\dots,\sF_t[\cdot],\GSim) \gets \ParSim_1(1^\secp,\ket*{\psi[x^*,\{v^*_\iota,r^*_\iota\}_{\iota \in [\tau]}]}),\] and write $\E_{\ParSim_3}$ as shorthard for the expectation over \[\ket*{\widetilde{\psi}},(\FSim_1[\cdot],\dots,\FSim_t[\cdot],\GSim) \gets \ParSim_3(1^\secp,\ket*{\psi[x^*,\{v^*_\iota,r^*_\iota\}_{\iota \in [\tau]}]}).\]

    Then,

    \begin{align*}
        &\E_{\ParSim_0}\left[\big\|\left(\Pi[x^*,\vk] \otimes \Pi[\neg (v^*_\tau,r^*_\tau)]\right)U^{\sF_1[x^*,\{v^*_\iota,r^*_\iota\}_{\iota \in [\tau]}],\dots,\GSim}\ket*{\widetilde{\psi}}\big\|^2\right] \\
        &\leq 2^n \cdot \sum_{\{w^*_\iota\}_{\iota \in [\tau+1,t]}} \E_{\ParSim_0}\left[\big\|\left(\Pi[x^*,\vk] \otimes \Pi[\neg (v^*_\tau,r^*_\tau)]\right)U^{\sF_1[x^*,\{v^*_\iota,r^*_\iota\}_{\iota \in [\tau]}],\dots,\GSim}\Pi[\{w^*_\iota\}_{\iota \in [\tau+1,t]}]\ket*{\widetilde{\psi}}\big\|^2\right]  \\
        &\leq 2^n \cdot \sum_{\{w^*_\iota\}_{\iota \in [\tau+1,t]}} \E_{\ParSim_3}\left[\big\|\left(\Pi[x^*,\vk] \otimes \Pi[\neg (v^*_\tau,r^*_\tau)]\right)U^{\FSim_1[x^*],\dots,\GSim}\Pi[\{w^*_\iota\}_{\iota \in [\tau+1,t]}]\ket*{\widetilde{\psi}}\big\|^2\right] + 2^{-\Omega(\kappa)} \\
        &\leq 2^{-\Omega(\kappa)},\\
    \end{align*}

    where 
    
    \begin{itemize}
        \item The first inequality follows from \cref{lemma:project} (an application of Cauchy-Schwarz).
        \item The second inequality follows from combining \cref{claim:parsim0-1}, \cref{claim:parsim1-2}, and \cref{claim:parsim2-3} proven above.
        \item The third inequality follows from \cref{lemma:mapping-hardness-3} proven below.
    \end{itemize}

\end{proof}

We finally complete the proof of security of our obfuscation scheme with the following lemma. An overview of the techniques involved in this proof, which include extraction via a purified random oracle, is given in \cref{subsec:proof-overview}.

\begin{lemma}\label{lemma:mapping-hardness-3}
    For any QPQ unitary $U$, input $x^*$, partial measurement results $\{v^*_\iota,r^*_\iota\}_{\iota \in [\tau]}$, and $\{w^*_\iota\}_{\iota \in [\tau+1,t]}$, it holds that 
    \[\E\left[\big\|\left(\Pi[x^*,\vk] \otimes \Pi[\neg (v^*_\tau,r^*_\tau)]\right)U^{\FSim_1[x^*],\dots,\GSim}\Pi[\{w^*_\iota\}_{\iota \in [\tau+1,t]}]\ket*{\widetilde{\psi}}\big\|^2\right] = 2^{-\Omega(\kappa)},\]
    
    where the expectation is over

    \[\ket*{\widetilde{\psi}},(\FSim_1[\cdot],\dots,\FSim_t[\cdot],\GSim) \gets \ParSim_3(1^\secp,\ket*{\psi[x^*,\{v^*_\iota,r^*_\iota\}_{\iota \in [\tau]}]}).\]
    
\end{lemma}

\begin{proof}
Recall that the oracles $\FSim_1[x^*],\dots,\FSim_t[x^*],\GSim$ output by $\ParSim_3$ internally make use of a random oracle $H : \{0,1\}^* \to \{0,1\}^\kappa$, and let $K = \poly(\secp)$ be an upper bound on the length of strings that $H$ takes as input.

We will purify the random oracle $H$, introducing an oracle register $\regD \coloneqq \{\regD_a\}_{a \in \{0,1\}^K}$, where each $\regD_a$ is a $\kappa$-qubit register. Define $\ket{+_\kappa}$ to be the uniform superposition over all $\kappa$-bit strings, and define $\ket{+_H}^\regD \coloneqq \ket{+_\kappa}^{\otimes \{\regD_a\}_{a \in \{0,1\}^K}}$ to be the uniform superposition over all random oracles $H$. Finally, let $A[\neq x^*] \coloneqq \{a \in \{0,1\}^K : a = (x,\cdot) \text{ for } x \neq x^*\}$ be the set of all random oracle inputs / sub-registers of $\regD$ that do not start with $x^*$.

Now, the purified random oracle begins by initializing $\regD$ to the state $\ket{+_H}$. Each time a query to $H$ is made on input register $\regA$ and output register $\regB$, we apply a unitary defined by the map \[\ket{a}^\regA\ket{b}^\regB\ket{H}^\regD \to \ket{a}\left(\CNOT^{\otimes \kappa}\right)^{\regD_a,\regB}\ket{b}^\regB\ket{H}^\regD.\] For the rest of this proof, we will implement oracle queries to $H$ using this purified procedure, and explicitly introduce the register $\regD$, initialized to $\ket{+_H}$, in our expressions.

The central claim we need is the following, which shows that is is hard to map onto $\Pi[\neg (v^*_\tau,r^*_\tau)]$ \emph{without disturbing} the random oracle registers $\{\regD_a\}_{a \in A[\neq x^*]}$. In other words, at any point at which the adversary's state has some overlap with $\Pi[\neg (v^*_\tau,r^*_\tau)]$, it \emph{must} be the case that the adversary currently holds some information about a random oracle output at $(x,\dots)$ for some $x \neq x^*$.

\begin{claim}\label{claim:extract} 

     For any QPQ unitary $U$, input $x^*$, partial measurement results $\{v^*_\iota,r^*_\iota\}_{\iota \in [\tau]}$, and $\{w^*_\iota\}_{\iota \in [\tau+1,t]}$, it holds that

    \begin{align*}
    \E\left[\big\| \left(\Pi[\neg (v^*_\tau,r^*_\tau)] \otimes \ketbra{+_\kappa}{+_\kappa}^{\otimes \{\regD_a\}_{a \in A[\neq x^*]}}\right)U^{\FSim_1[x^*],\dots,\GSim}\ket*{\widetilde{\psi}^*}\ket{+_H}^\regD\big\|^2\right] = 2^{-\Omega(\kappa)},
    \end{align*}

    where the expectation is over 

    \[\ket*{\widetilde{\psi}},(\FSim_1[\cdot],\dots,\FSim_t[\cdot],\GSim) \gets \ParSim_3(1^\secp,\ket*{\psi[x^*,\{v^*_\iota,r^*_\iota\}_{\iota \in [\tau]}]}),\]

    and \[\ket*{\widetilde{\psi}^*} \coloneqq \Pi[\{w^*_\iota\}_{\iota \in [\tau+1,t]}]\ket*{\widetilde{\psi}}.\]

\end{claim}

\begin{proof}

First, consider the following distribution, which essentially toggles between $\ParSim_4$ and $\ParSim_5$.

\begin{tcolorbox}[breakable, enhanced]\footnotesize
$\Sim[x^*](1^\secp,\ket{\psi})$:
    \begin{itemize}
        \item Sample $k \gets \Gen(1^\kappa,n)$, and compute $\ket*{\psi_k} = \Enc_k(\ket{\psi})$.
        \item Sample a signature token $(\vk,\ket{\sk}) \gets \TokGen(1^\kappa)$.
        \item Let $H : \{0,1\}^* \to \{0,1\}^\kappa$ be a random oracle. 
        \item For each $i \in [t]$, define the function $\FSim_i[x^*][z](x,\sigma_{x},\widetilde{v}_1,\dots,\widetilde{v}_i,\widetilde{w}_i,\ell_1,\dots,\ell_{i-1})$: 
        \begin{itemize}
            \item If $z = \emptyset$, output $\FSim_i(x,\sigma_{x},\widetilde{v}_1,\dots,\widetilde{v}_i,\widetilde{w}_i,\ell_1,\dots,\ell_{i-1})$, where $\FSim_i$ is defined as in $\QObf_5$.
            \item Otherwise, if $z = w^*$:
            \begin{itemize}
                \item If $x = x^*$, output $\FSim_i(x^*,\sigma_{x^*},\widetilde{v}_1,\dots,\widetilde{v}_i,\widetilde{w}_i,\ell_1,\dots,\ell_{i-1})$.
                \item If $x \neq x^*$, output $\FSim_i[w^*](x,\sigma_{x},\widetilde{v}_1,\dots,\widetilde{v}_i,\widetilde{w}_i,\ell_1,\dots,\ell_{i-1})$, where $\FSim_i[w^*]$ is defined as in $\QObf_4$.
            \end{itemize}
        \end{itemize}
        \item Define $\GSim$ as in $\QObf_2$.

        \item Output $\ket*{\widetilde{\psi}} = \ket{\psi_k}\ket{\sk}, O = \left(\FSim_1[x^*][\cdot],\dots,\FSim_t[x^*][\cdot],\GSim\right).$
    \end{itemize}
\end{tcolorbox}

Indeed, observe that  

\[\ket*{\widetilde{\psi}},(\FSim_1[x^*][w^*],\dots,\FSim_t[x^*][w^*],\GSim) \gets \Sim[x^*](1^\secp,\ket*{\psi[x^*,\{v^*_\iota,r^*_\iota\}_{\iota \in [\tau]}]})\] is equivalent to \[\ket*{\widetilde{\psi}},(\FSim_1[x^*,w^*],\dots,\FSim_t[x^*,w^*],\GSim) \gets \ParSim_4(1^\secp,\ket*{\psi[x^*,\{v^*_\iota,r^*_\iota\}_{\iota \in [\tau]}]}),\] while
\[\ket*{\widetilde{\psi}},(\FSim_1[x^*][\emptyset],\dots,\FSim_t[x^*][\emptyset],\GSim) \gets \Sim[x^*](1^\secp,\ket*{\psi[x^*,\{v^*_\iota,r^*_\iota\}_{\iota \in [\tau]}]})\] is equivalent to \[\ket*{\widetilde{\psi}},(\FSim_1,\dots,\FSim_t,\GSim) \gets \ParSim_5(1^\secp,\ket*{\psi[x^*,\{v^*_\iota,r^*_\iota\}_{\iota \in [\tau]}]}).\]

Next, recall the definition of $R[k,w^*]$ from \cref{subsec:main-theorem}: \\

\noindent $R[k,w^*](x,\widetilde{v}_1,\dots,\widetilde{v}_i):$
\begin{itemize}
    \item Compute $(v_1,\dots,v_i) = \Dec_{k,L_i\dots L_1,\theta_i[V_1,\dots,V_i]}(\widetilde{v}_1,\dots,\widetilde{v}_i)$. If the result is $\bot$, then output $\bot$.
    \item For $\iota \in [i],$ compute $(\cdot,r_\iota) = f_\iota^{x,r_1,\dots,r_{\iota-1}}(v_1,\dots,v_\iota,w^*_\iota)$.
    \item Output $r_i$.
\end{itemize}

Now, for any $w^* = \{w^*_\iota\}_{\iota \in [t]}$, define a unitary $\Sigma[w^*]$ that permutes the registers $A[\neq x^*]$ according to the rule that, for $x \neq x^*$, swaps 
\[(x,\sigma_x,\widetilde{v}_1,\dots,\widetilde{v}_i,\widetilde{w}_i,\ell_1,\dots,\ell_{i-1},0) ~~~ \text{and} ~~~ (x,\sigma_x,\widetilde{v}_1,\dots,\widetilde{v}_i,\widetilde{w}_i,\ell_1,\dots,\ell_{i-1},1)\] whenever $R[k,w^*](x,\widetilde{v}_1,\dots,\widetilde{v}_i)=1$. By definition of $\Sim[x^*]$, we have the following fact.

\begin{fact}\label{fact:1}
For any unitary $U$ and \[\ket*{\widetilde{\psi}},(\FSim_1[x^*][\cdot],\dots,\FSim_t[x^*][\cdot],\GSim) \in \Sim[x^*](1^\secp,\ket*{\psi[x^*,\{v^*_\iota,r^*_\iota\}_{\iota \in [\tau]}]}),\] it holds that
\begin{align*}
U^{\FSim_1[x^*][w^*],\dots,\GSim}\ket*{\widetilde{\psi}^*}\ket{+_H}^\regD = \Sigma[w^*] U^{\FSim_1[x^*][\emptyset],\dots,\GSim}\Sigma[w^*]\ket*{\widetilde{\psi}^*} \ket{+_H}^\regD,
\end{align*}

where the unitary $\Sigma[w^*]$ is applied to registers $\{\regD_a\}_{a \in A[\neq x^*]}$.
\end{fact}

We will also use the following immediate fact.

\begin{fact}\label{fact:2}
    For any $w^*$, \[\Sigma[w^*]\ket{+_\kappa}^{\otimes\{\regD_a\}_{a \in A[\neq x^*]}} = \ket{+_\kappa}^{\otimes\{\regD_a\}_{a \in A[\neq x^*]}}.\]
\end{fact}

Now, write $\E_{\ParSim_3}$ as shorthand for the expectation over \[\ket*{\widetilde{\psi}},(\FSim_1[\cdot],\dots,\FSim_t[\cdot],\GSim) \gets \ParSim_3(1^\secp,\ket*{\psi[x^*,\{v^*_\iota,r^*_\iota\}_{\iota \in [\tau]}]}),\] write $\E_{\ParSim_4}$ as shorthand for the expectation over \[\ket*{\widetilde{\psi}},(\FSim_1[\cdot],\dots,\FSim_t[\cdot],\GSim) \gets \ParSim_4(1^\secp,\ket*{\psi[x^*,\{v^*_\iota,r^*_\iota\}_{\iota \in [\tau]}]}),\] write $\E_{\ParSim_5}$ as shorthand for the expectation over \[\ket*{\widetilde{\psi}},(\FSim_1,\dots,\FSim_t,\GSim) \gets \ParSim_5(1^\secp,\ket*{\psi[x^*,\{v^*_\iota,r^*_\iota\}_{\iota \in [\tau]}]}),\] and write $\E_{\Sim[x^*]}$ as shorthand for \[\ket*{\widetilde{\psi}},(\FSim_1[x^*][\cdot],\dots,\FSim_t[x^*][\cdot],\GSim) \gets \Sim[x^*](1^\secp,\ket*{\psi[x^*,\{v^*_\iota,r^*_\iota\}_{\iota \in [\tau]}]}).\]

Then,

\begin{align*}
    &\E_{\ParSim_3}\left[\bigg\| \left(\Pi[\neg (v^*_\tau,r^*_\tau)] \otimes \ketbra{+_\kappa}{+_\kappa}^{\otimes \{\regD_a\}_{a \in A[\neq x^*]}}\right)U^{\FSim_1[x^*],\dots,\GSim}\ket*{\widetilde{\psi}^*}\ket{+_H}^\regD\bigg\|^2\right] \\
    &= \E_{\ParSim_3}\left[\bigg\| \sum_{w^*}\left(\Pi[\neg (v^*_\tau,r^*_\tau)] \otimes \ketbra{+_\kappa}{+_\kappa}^{\otimes \{\regD_a\}_{a \in A[\neq x^*]}}\right)U^{\FSim_1[x^*],\dots,\GSim}\Pi[w^*]\ket*{\widetilde{\psi}^*}\ket{+_H}^\regD\bigg\|^2\right] \\
    &= \E_{\ParSim_3}\left[\bigg\|\Pi[\neg (v^*_\tau,r^*_\tau)] \sum_{w^*} \left(\ketbra{+_\kappa}{+_\kappa}^{\otimes \{\regD_a\}_{a \in A[\neq x^*]}}\right)U^{\FSim_1[x^*],\dots,\GSim}\Pi[w^*]\ket*{\widetilde{\psi}^*}\ket{+_H}^\regD\bigg\|^2\right] \\
    &\leq \E_{\ParSim_4}\left[\bigg\|\Pi[\neg (v^*_\tau,r^*_\tau)] \sum_{w^*} \left(\ketbra{+_\kappa}{+_\kappa}^{\otimes \{\regD_a\}_{a \in A[\neq x^*]}}\right)U^{\FSim_1[x^*,w^*],\dots,\GSim}\Pi[w^*]\ket*{\widetilde{\psi}^*}\ket{+_H}^\regD\bigg\|^2\right] + 2^n \cdot 2^{-\Omega(\kappa)} \\
    &= \E_{\Sim[x^*]}\left[\bigg\|\Pi[\neg (v^*_\tau,r^*_\tau)] \sum_{w^*} \left(\ketbra{+_\kappa}{+_\kappa}^{\otimes \{\regD_a\}_{a \in A[\neq x^*]}}\right)U^{\FSim_1[x^*][w^*],\dots,\GSim}\Pi[w^*]\ket*{\widetilde{\psi}^*}\ket{+_H}^\regD\bigg\|^2\right] + 2^{-\Omega(\kappa)} \\
    &= \E_{\Sim[x^*]}\left[\bigg\|\Pi[\neg (v^*_\tau,r^*_\tau)] \sum_{w^*} \left(\ketbra{+_\kappa}{+_\kappa}^{\otimes \{\regD_a\}_{a \in A[\neq x^*]}}\right)\Sigma[w^*]U^{\FSim_1[x^*][\emptyset],\dots,\GSim}\Sigma[w^*]\Pi[w^*]\ket*{\widetilde{\psi}^*}\ket{+_H}^\regD\bigg\|^2\right]\\ 
    &~~~~~~ + 2^{-\Omega(\kappa)} \\
    &= \E_{\Sim[x^*]}\left[\bigg\|\Pi[\neg (v^*_\tau,r^*_\tau)] \sum_{w^*} \left(\ketbra{+_\kappa}{+_\kappa}^{\otimes \{\regD_a\}_{a \in A[\neq x^*]}}\right)U^{\FSim_1[x^*][\emptyset],\dots,\GSim}\Pi[w^*]\ket*{\widetilde{\psi}^*}\ket{+_H}^\regD\bigg\|^2\right] + 2^{-\Omega(\kappa)} \\
    &= \E_{\ParSim_5}\left[\bigg\|\Pi[\neg (v^*_\tau,r^*_\tau)] \sum_{w^*} \left(\ketbra{+_\kappa}{+_\kappa}^{\otimes \{\regD_a\}_{a \in A[\neq x^*]}}\right)U^{\FSim_1,\dots,\GSim}\Pi[w^*]\ket*{\widetilde{\psi}^*}\ket{+_H}^\regD\bigg\|^2\right] + 2^{-\Omega(\kappa)} \\
    &= \E_{\ParSim_5}\left[\bigg\|\left(\Pi[\neg (v^*_\tau,r^*_\tau)] \otimes \ketbra{+_\kappa}{+_\kappa}^{\otimes \{\regD_a\}_{a \in A[\neq x^*]}}\right)U^{\FSim_1,\dots,\GSim}\ket*{\widetilde{\psi}^*}\ket{+_H}^\regD\bigg\|^2\right] + 2^{-\Omega(\kappa)} \\
    &\leq 2^{-\Omega(\kappa)},
\end{align*}

where 

\begin{itemize}
    \item The first inequality follows from \cref{claim:parsim3-4} proven above.
    \item The following equality is by definition of $\Sim[x^*]$.
    \item The following equality is \cref{fact:1}.
    \item The following equality is \cref{fact:2}.
    \item The following equality is by definition of $\Sim[x^*]$.
    \item The final inequality follows from \cref{claim:parsim-auth} proven above.
\end{itemize}

\end{proof}

Now, assume for contradiction that the lemma is false. Combined with the fact that \cref{claim:extract} is true, this implies that 

\begin{align*}
    \E\left[\big\| \left(\Pi[x^*,\vk] \otimes \left(\cI - \ketbra{+_\kappa}{+_\kappa}^{\otimes \{\regD_a\}_{a \in A[\neq x^*]}}\right)\right)U^{\FSim_1[x^*],\dots,\GSim}\ket*{\widetilde{\psi}^*}\ket{+_H}^\regD\big\|^2\right] = 2^{-o(\kappa)},
\end{align*}

where the expectation is over

 \[\ket*{\widetilde{\psi}},(\FSim_1[\cdot],\dots,\FSim_t[\cdot],\GSim) \gets \ParSim_3(1^\secp,\ket*{\psi[x^*,\{v^*_\iota,r^*_\iota\}_{\iota \in [\tau]}]}).\]

However, this would violate the security of the signature token scheme, which we now show. Consider the following reduction that takes as input the signing key $\ket{\sk}$ for a signature token scheme, and has oracle access to $\TokVer[\vk]$.

\begin{itemize}
    \item Prepare $\ket*{\widetilde{\psi}^*},\ket{+_H}$, and $(\FSim_1[\cdot],\dots,\FSim_t[\cdot],\GSim)$ as in the description of \cref{lemma:mapping-hardness-3}, and run $U^{\FSim_1[\cdot],\dots,\GSim}\ket*{\widetilde{\psi}^*}\ket{+_H}$, except that $\TokVer$ queries  are computed by forwarding them to $\TokVer[\vk]$.
    \item Measure the final state of $U$ in the standard basis, and parse the outcome as $(\sigma_{x^*},\cdot)$.
    \item Measure the final state on registers $\{\regD_a\}_{a \in A[\neq x^*]}$ in the Hadamard basis. If any register $\regD_a$ gives a result other than $0^\kappa$, then parse $a = (x,\sigma_x,\cdot)$ for some $x \neq x^*$. 
    \item Output $(\sigma_{x^*},\sigma_x)$.
\end{itemize}

Then, by the definition of $\Pi[x^*,\vk]$ and the fact that the random oracle $H$ is only ever queried on inputs that begin with $(x,\sigma_x)$ such that $\TokVer(\vk,x,\sigma_x) = \top$, we have that with probability $2^{-o(\kappa)}$, $\TokVer(\vk,x^*,\sigma_{x^*}) = \TokVer(\vk,x,\sigma_x) = \top$, which violates security of the signature token scheme (\cref{def:unforgeability}). Indeed, note that the signature token scheme is secure against $\poly(\secp)$ query bounded adversaries that otherwise have unlimited time and space, which is satisfied by the reduction given above.

\end{proof}

\section{Acknowledgements}

We thank Sam Gunn for correspondence throughout this project, including several illuminating discussions.

\ifsubmission
\bibliographystyle{plain}
\else
\bibliographystyle{alpha}
\fi

\bibliography{abbrev3,crypto,main}
\end{document}